\PassOptionsToPackage{unicode}{hyperref}
\PassOptionsToPackage{hyphens}{url}
\PassOptionsToPackage{dvipsnames,svgnames,x11names}{xcolor}
\documentclass[12pt]{article}
\usepackage{enumerate}
\usepackage{natbib}
\usepackage{tabularx}
\usepackage{epsfig}
\usepackage{tikz}
\usetikzlibrary{positioning}
\usetikzlibrary{arrows.meta}
\usetikzlibrary{calc}
\usepackage{multirow}
\usepackage{color}
\usepackage{soul}
\usepackage{amsmath}
\usepackage{amssymb}
\usepackage{bookmark}
\usepackage{rotating, amsthm}
\usepackage{float}
\usepackage{latexsym}
\usepackage{graphicx}
\usepackage{amsfonts}
\usepackage{euscript}
\usepackage{color}
\usepackage{iftex}
\usepackage{enumitem}
\setlist[enumerate]{itemsep=0pt,parsep=0pt,topsep=0pt}
\setlist[itemize]{itemsep=0pt,parsep=0pt,topsep=0pt}
\everymath{\displaystyle}
\ifPDFTeX
  \usepackage[T1]{fontenc}
  \usepackage[utf8]{inputenc}
  \usepackage{textcomp}
\else
  \usepackage{unicode-math}
  \defaultfontfeatures{Scale=MatchLowercase}
  \defaultfontfeatures[\rmfamily]{Ligatures=TeX,Scale=1}
\fi
\usepackage{lmodern}
\ifPDFTeX\else
\fi
\IfFileExists{upquote.sty}{\usepackage{upquote}}{}
\IfFileExists{microtype.sty}{
  \usepackage[]{microtype}
  \UseMicrotypeSet[protrusion]{basicmath} 
}{}
\makeatletter
\@ifundefined{KOMAClassName}{
  \IfFileExists{parskip.sty}{
    \usepackage{parskip}
  }{
    \setlength{\parindent}{0pt}
    \setlength{\parskip}{6pt plus 2pt minus 1pt}}
}{
  \KOMAoptions{parskip=half}}
\makeatother
\usepackage{xcolor}
\makeatletter
\ifx\paragraph\undefined\else
  \let\oldparagraph\paragraph
  \renewcommand{\paragraph}{
    \@ifstar
      \xxxParagraphStar
      \xxxParagraphNoStar
  }
  \newcommand{\xxxParagraphStar}[1]{\oldparagraph*{#1}\mbox{}}
  \newcommand{\xxxParagraphNoStar}[1]{\oldparagraph{#1}\mbox{}}
\fi
\ifx\subparagraph\undefined\else
  \let\oldsubparagraph\subparagraph
  \renewcommand{\subparagraph}{
    \@ifstar
      \xxxSubParagraphStar
      \xxxSubParagraphNoStar
  }
  \newcommand{\xxxSubParagraphStar}[1]{\oldsubparagraph*{#1}\mbox{}}
  \newcommand{\xxxSubParagraphNoStar}[1]{\oldsubparagraph{#1}\mbox{}}
\fi
\makeatother
\usepackage{longtable,booktabs,array}
\usepackage{calc} 

\usepackage{etoolbox}
\makeatletter
\patchcmd\longtable{\par}{\if@noskipsec\mbox{}\fi\par}{}{}
\makeatother
\IfFileExists{footnotehyper.sty}{\usepackage{footnotehyper}}{\usepackage{footnote}}
\makesavenoteenv{longtable}
\usepackage{graphicx}
\makeatletter
\def\maxwidth{\ifdim\Gin@nat@width>\linewidth\linewidth\else\Gin@nat@width\fi}
\def\maxheight{\ifdim\Gin@nat@height>\textheight\textheight\else\Gin@nat@height\fi}
\makeatother

\setkeys{Gin}{width=\maxwidth,height=\maxheight,keepaspectratio}

\makeatletter
\def\fps@figure{htbp}
\makeatother

\makeatletter
\@ifpackageloaded{caption}{}{\usepackage{caption}}
\AtBeginDocument{
\ifdefined\contentsname
  \renewcommand*\contentsname{Table of contents}
\else
  \newcommand\contentsname{Table of contents}
\fi
\ifdefined\listfigurename
  \renewcommand*\listfigurename{List of Figures}
\else
  \newcommand\listfigurename{List of Figures}
\fi
\ifdefined\listtablename
  \renewcommand*\listtablename{List of Tables}
\else
  \newcommand\listtablename{List of Tables}
\fi
\ifdefined\figurename
  \renewcommand*\figurename{Figure}
\else
  \newcommand\figurename{Figure}
\fi
\ifdefined\tablename
  \renewcommand*\tablename{Table}
\else
  \newcommand\tablename{Table}
\fi
}
\@ifpackageloaded{float}{}{\usepackage{float}}
\floatstyle{ruled}
\@ifundefined{c@chapter}{\newfloat{codelisting}{h}{lop}}{\newfloat{codelisting}{h}{lop}[chapter]}
\floatname{codelisting}{Listing}

\makeatother
\makeatletter
\@ifpackageloaded{caption}{}{\usepackage{caption}}
\@ifpackageloaded{subcaption}{}{\usepackage{subcaption}}
\makeatother

\ifLuaTeX
  \usepackage{selnolig}  
\fi

\IfFileExists{xurl.sty}{\usepackage{xurl}}{} 
\hypersetup{
  pdftitle={Title},
  pdfauthor={Author 1; Author 2},
  pdfkeywords={3 to 6 keywords, that do not appear in the title},
  colorlinks=true,
  linkcolor={Black},
  filecolor={Maroon},
  citecolor={Black},
  urlcolor={Blue},
  pdfcreator={LaTeX via pandoc}}
\allowdisplaybreaks

\numberwithin{equation}{section}

\newcommand{\cov}{\mbox{cov}}

\newcommand{\bA}{\mbox{$\boldsymbol{A}$}}

\newcommand{\beqn}{\begin{eqnarray}}
\newcommand{\eeqn}{\end{eqnarray}}

\newtheorem{theorem}{\bf Theorem}[section]

\newtheorem{lemma}{Lemma}[section]
\newtheorem{remark}{Remark}[section]

\newcommand{\nn}{\nonumber}

\newcommand{\ignore}[1]{}{}

\newcommand{\anon}{1}

\begin{document}
\def\spacingset#1{\renewcommand{\baselinestretch}%
{#1}\small\normalsize} \spacingset{1}

\if1\anon
{
  \title{\bf Feature Screening for High-Dimensional Structural Break Predictive Regression}
  \author{
    Zhenjie Qin\\
    School of Mathematical Sciences, Zhejiang University
    \\
    and\\
    Rongmao Zhang\footnote{Corresponding author}\\
    School of Statistics and Data Sciences, Zhejiang Gongshang University\\
    and \\
    Wenyang Zhang \\
    Faculty of Business Administration and APAEM, University of Macau\\
    and\\
    Yang Zu\\
    Department of Economics and APAEM, University of Macau
    }
  \maketitle
} \fi

\if0\anon
{
  \bigskip
  \bigskip
  \bigskip
  \begin{center}
    {\LARGE\bf Feature Screening for High-Dimensional Structural Break Predictive Regression}
\end{center}
  \medskip
} \fi

\bigskip
\begin{abstract}
Predictive regression is a crucial tool for exploring return predictability. In this study, we introduce an efficient procedure for selecting and estimating active predictors and change points in 
structural break predictive regression. Our approach allows the number of change points to increase with the sample size and accommodates sparse active predictors that may be stationary or cointegrated. 
We begin by identifying the active predictors using a Sure Independence Canonical Screening (SICS) procedure. Next, we estimate the change points through a Ratio-Controlled Regression Screening (RCRS) method. Finally, we reduce redundancy by eliminating unnecessary breakpoints and predictors using information criteria (IC). This approach allows for consistent estimation and selection of true breakpoints and active predictors. Our simulations and empirical studies demonstrate that the proposed procedure performs effectively.
\end{abstract}

\noindent%
\vfill

\newpage
\spacingset{1.8}
\section{INTRODUCTION}
\label{sec:intro}

The predictability of asset returns has long been a focal research question, dating back to \cite{dow1920scientific}. Predictive regression is a widely used methodological framework in this field. Important contributions to univariate predictive regressions include \cite{campbellEfficientTestsStock2006}, \cite{welchComprehensiveLookEmpirical2008}, \cite{elliottControlFunctionApproach2011}, \cite{caiTestingPredictiveRegression2014a}, among others. Recent extensions to multiple predictive regressions include \cite{kooHighdimensionalPredictiveRegression2020a}, \cite{leeLASSOPredictiveRegression2022a},  \cite{tuPenetratingSporadicReturn2023}, \cite{fanPredictiveQuantileRegression2023}, \cite{mei2024lasso}, \cite{fangDeterminationEffectiveCointegration}, \cite{xie2026regime}, \cite{gao2026lasso}, among others.

Most existing studies focus on predictive regressions without accounting for structural breaks. However, various abrupt events, such as economic crises, policy instabilities, or environmental changes, can significantly alter the predictability of the predictors. Neglecting structural changes in predictive regression models is ineffective and can lead to inaccurate results, as suggested by \cite{welchComprehensiveLookEmpirical2008}. In a recent study, \cite{tuPenetratingSporadicReturn2023} introduced a multiple predictive regression model that includes finite change points and predictors. Their framework can accommodate both weakly stationary (denoted as \(I(0)\)) and nonstationary (\(I(1)\)) predictors. In empirical finance, asset returns are typically treated as \(I(0)\) series, which means that the right-hand side of the regression must maintain a stationary balance. To achieve this, \cite{tuPenetratingSporadicReturn2023} suggests performing cointegration analysis among \(I(1)\) regressors and incorporating the resulting cointegration relationships as stationary covariates. 

In practical situations where both the number of predictors and the number of change points are large, the method proposed by \cite{tuPenetratingSporadicReturn2023} encounters two major limitations. First, having many redundant regressors can distort least square estimates (LSE) in regimes with small sample size. Second, multiple unit root testings and cointegration analysis are in need before detecting change points.
These challenges lead us to propose a procedure before change points detection, which can not only effectively reduce dimensionality within a high-dimensional predictive regression framework that accounts for structural breaks,
but also avoid the need for preliminary unit root testings and cointegration analysis.

Developing an effective procedure for reducing dimensions before estimating change points, particularly in the context of cointegration relationships, is important for practical applications yet challenging. Key studies, such as those by  \cite{leeLASSOPredictiveRegression2022a} and \cite{mei2024lasso}, utilize the Least Absolute Shrinkage and Selection Operator (LASSO) for variable selection in predictive regressions without structural breaks, while accounting for both fixed and diverging dimensions. However, they observe that LASSO often fails to consistently identify active cointegrated predictors, especially when the number of covariates increases alongside the sample size.

     Feature Screening methods, initially proposed by \cite{fanSureIndependenceScreening2008}, have been widely adopted for dimension reduction and machine learning. See also \cite{buhlmannVariableSelectionHighdimensional2010}, \cite{zhuModelFreeFeatureScreening2011}, \cite{liFeatureScreeningDistance2012a}, \cite{liuFeatureSelectionVarying2014}, \cite{liVariableSelectionPartial2017}, \cite{maVariableScreeningQuantile2017a}, \cite{chenErrorVarianceEstimation2018a}, \cite{zhongModelfreeVariableScreening2023}, \cite{fanAreLatentFactor2024}, among others.
The original Sure Independence Screening (SIS) method in \cite{fanSureIndependenceScreening2008} selects predictors that exhibit large absolute marginal correlations with the response variable. However, in predictive regression models that include potentially active nonstationary predictors, the original SIS cannot be applied directly. This limitation arises because the correlations between the response and any $I(1)$ predictors, even those that are active, tend to be small. One can see \eqref{pear0}, Figures \ref{fig:sis2} in Section \ref{sec:sics} for more details. This challenge motivates us to propose an effective procedure for selecting active predictors, whether it's $I(0)$ or cointegrated $I(1)$.

Specifically, we introduce a method for dimension reduction called Sure Independence Canonical Screening (SICS). The idea behind SICS is simple: the canonical correlation between the response variable and the complement of an active predictor (i.e., all predictors excluding this active one) will be much smaller. In contrast, the canonical correlation between the response variable and the complement of an inactive predictor will be larger, as the active predictors remain in the complement set.

Instead of relying on large absolute correlations between the response variable and individual predictors for screening as in classical SIS, our SICS utilize small canonical correlations between the response variable and the group of complement variables associated with a specific predictor for selection purpose. We demonstrate that the proposed SICS method is effective for predictor selection, even in the presence of potential structural breaks. After selecting the variables, we use a Ratio-Controlled Regression Screening (RCRS) procedure to estimate the locations of change points within a dimension-reduced model. Furthermore, we apply backward elimination based on information criteria (IC) to refine the model by removing unnecessary change points and predictors.

Our contributions can be categorized into three main aspects. 

First, we introduce a latent factor structure to model high-dimensional unit-root predictors. By utilizing a low-dimensional unit-root factor, we generate a high-dimensional unit-root predictor vector, allowing both dimensions to increase infinitely as the sample size grows. Our model accommodates multiple cointegration relationships, whereas \cite{kooHighdimensionalPredictiveRegression2020a}, \cite{zhouSemiparametricSingleindexPredictive2024} only consider single relationship. Further, our work overcomes the theoretical problem in \cite{kooHighdimensionalPredictiveRegression2020a} and 
\cite{fanPredictiveQuantileRegression2023} that the number of $I(1)$ predictors is constrained to be finite.

Second, we propose an effective method for selecting active predictors in a predictive regression model that addresses the limitation of LASSO when it comes to choosing cointegrated \(I(1)\) predictors. Our SICS procedure demonstrates the sure screening property, meaning that it can successfully identify all active predictors, particularly those that are cointegrated, even in the presence of multiple change points within the model. Additionally, a backward elimination approach helps ensure that active predictors can be selected precisely. Our predictor selection method does not require preliminary unit root testing and cointegration analysis.

Third, we propose a feasible method for estimating change points in structural break predictive regression with increasing dimension. Our simulation study demonstrates that the Group Forward Regression Screening (GFRS) method presented in \cite{tuPenetratingSporadicReturn2023} exhibits poor finite-sample performance in a predictive regression with many covariate variables. Our RCRS procedure selects change points in a dimension-reduced model one at a time, aiming to maximize the local Residual Sum of Squares (RSS). We construct a ratio for the RSS and prove that it converges to zero, provided that all true change points are included within a neighborhood set of the estimated change points. Additionally, we employ a backward elimination strategy using a different information criterion (IC) to remove unnecessary breaks and enhance the accuracy of the estimation.

The remainder of this paper is structured as follows: In Section 2, we introduce our model and methodologies. Section 3 discusses the asymptotic properties of our methods. In Section 4, we present the results of our simulation studies. Section 5 addresses the real-world applications of the proposed procedure. Finally, Section 6 concludes our work. All proofs of the main results are included in the supplementary material.

{\bf Notational Conventions.} We call a vector process $\boldsymbol{u}_t$ weakly stationary if
(i) ${\rm E}(\boldsymbol{u}_t)$ is a constant vector independent of $t$, and (ii) ${\rm E}\|\boldsymbol{u}_t\|^2 < \infty$,
and $\cov(\boldsymbol{u}_t, \boldsymbol{u}_{t+s}) $ depends on $s$ only for any integers $t, s$, where
$\| \cdot \|$ denotes the Euclidean norm.
 Further, if $\boldsymbol{u}_{t}$ has spectral density matrix that is finite and
positive definite at zero frequency we say $\boldsymbol{u}_{t}$ is an $I\left(
0\right) $ process. 
If $\nabla \boldsymbol{u}_t = \boldsymbol{u}_t - \boldsymbol{u}_{t-1}$ 
 is an $I(0)$ process, then we say $\boldsymbol{u}_{t}$ is an $I(1)$ process.
For a given matrix $\boldsymbol{A}\in \mathbb{R}^{m\times n},$ we denote $\boldsymbol{A}'$ be the transpose of $\boldsymbol{A}$; $\|\boldsymbol{A}\|, \|\boldsymbol{A}\|_F$ be its 2-norm and F-norm. We denote $\boldsymbol{I}_p$ be the identity matrix with dimension $p;~ \boldsymbol{0}_p, \boldsymbol{1}_p$ be $p-$dimensional column vectors whose elements are all $0, 1,$ respectively. If $m=n$, we use $\lambda_{\max}(\boldsymbol{A}), \lambda_{\min}(\boldsymbol{A})$ to represent the maximum and minimum eigenvalues of $\bA$. For $x\in \mathbb{R}, [x]:= n$ if $n\leq x <n+1.$ For a set $S,$ let $1_{S}$ be the indicate function of $S$, i.e. $1_S = 1$ if and only if $S$ holds; let $|S|$ be the cardinality of $S$; let  $S^c$ be the complementary set of $S.$ We use $(s:l)$ to represent a time period containing $s,s+1,\cdots,l-1.$
For any two processes $\boldsymbol{\xi}_{t}$ and $\boldsymbol{\eta}_{t},$ we use $\bar{\boldsymbol{\xi}}^{(s:l)}=(l-s)^{-1}\sum_{t=s}^{l-1}\boldsymbol{\xi}_t$ to denote the sample mean of $\boldsymbol{\xi}_t$ in the period $(s:l)$; $\widehat{\boldsymbol{\Omega}}_{\xi,\eta}^{(s:l)} = (l-s)^{-1}\sum_{t=s}^{l-1} (\boldsymbol{\xi}_t-\bar{\boldsymbol{\xi}}^{(s:l)})(\boldsymbol{\eta}_t-\bar{\boldsymbol{\eta}}^{(s:l)})'$ to denote the sample covariance of $\boldsymbol{\xi}_{t}$ and $\boldsymbol{\eta}_{t}$ in the period $(s:l)$, respectively. For simplicity, we use $\widehat{\boldsymbol{\Omega}}_{\xi}^{(s:l)}$ to denote $\widehat{\boldsymbol{\Omega}}_{\xi,\xi}^{(s:l)};$ we use $\bar{\boldsymbol{\xi}}, \widehat{\boldsymbol{\Omega}}_{\xi,\eta}$ to denote $\bar{\boldsymbol{\xi}}^{(1:T+1)},~\widehat{\boldsymbol{\Omega}}_{\xi,\eta}^{(1:T+1)}$.
\vspace{-0.60cm}
\section{MODEL AND ESTIMATION}
\label{sec:meth}

\subsection{Models}
\label{sec:model}
Let 
$\boldsymbol{z}_{t} = (z_{1t}, \cdots, z_{p_z,t})'$ be the observed  $I(0)$ predictors, $\boldsymbol{w}_t = (w_{1t}, \cdots, w_{p_w,t})'$ be the observed $I(1)$ predictors with possible cointegration, and $\boldsymbol{x}_t = (\boldsymbol{z}_t', \boldsymbol{w}_t')'$ with dimension $p = p_z + p_w$. In this paper, we consider the following high-dimensional structural break predictive regression model:
\beqn
y_{t+1} &=& \sum_{i=1}^{m_0+1} \boldsymbol{\gamma}_i' \boldsymbol{x}_t 1_{\{t_{i-1}^0 \leq t < t_i^0\}} + u_{t+1}, ~~ t = 1, 2, \cdots, T,
\label{model}\\
\boldsymbol{w}_t &= &\boldsymbol{Q} \boldsymbol{F}_{t} + \boldsymbol{e}_t, \,\quad \boldsymbol{F}_{t} = \boldsymbol{F}_{t-1} + \boldsymbol{v}_t,\label{panic}
\eeqn
where 
 $\boldsymbol{F}_{t}$ is an $r_F$-dimensional unit-root factor, $\boldsymbol{Q}$ is an orthogonal loading matrix with $\boldsymbol{Q}'\boldsymbol{Q}=\boldsymbol{I}_{r_F},  \, $  both noises $ \boldsymbol{v}_t$ and $\boldsymbol{e}_t$ are $I(0)$ processes, ${\rm Cov}(\boldsymbol{e}_t)=\sigma_e^2\boldsymbol{I}_{p_w}$ for some $\sigma_e^2>0$;  $S^0 = \{t_1^0, \cdots, t_{m_0}^0\}$ denotes the true locations set of breaks (change points), $t_0^0 = 1, t_{m_0+1}^0 = T+1, m_0 = |S^0|,$ and $\{\boldsymbol{\gamma}_i\}_{i=1}^{m_0 + 1}$ are the true coefficient vectors of $\boldsymbol{x}_t$ in each regime. Model~(2.2) has been widely used in cointegrated nonstationary time series, see for example,  \cite{baiPANICAttackUnit2004}, \cite{onatskiSpuriousFactorAnalysis2021}, \cite{zhangIdentifyingHighDimensionalb}. \cite{fangDeterminationEffectiveCointegration} uses the similar structure for predictive regression where $r_F$ is fixed, while we allow $r_F\rightarrow \infty$ as $T\rightarrow \infty$. If we set $\boldsymbol{Q}^\perp \in \mathbb{R}^{p_w \times (p_w - r_F)}$ be the orthogonal-complementary matrix of $\boldsymbol{Q},$ i.e. $(\boldsymbol{Q}^\perp)' \boldsymbol{Q}^\perp = \boldsymbol{I}_{p_w - r_F}, ~\boldsymbol{Q}' \boldsymbol{Q}^\perp = \boldsymbol{O},$ then by model~(2.2), we have $(\boldsymbol{Q}^{\perp})'\boldsymbol{w}_t=(\boldsymbol{Q}^{\perp})'\boldsymbol{e}_t$ is an $(p_w - r_F)$-dimensional $I(0)$ process and $\boldsymbol{Q}'\boldsymbol{w}_t=\boldsymbol{F}_{t}+\boldsymbol{Q}'\boldsymbol{e}_t$ is an $r_F$-dimensional $I(1)$ process.  This implies that $\boldsymbol{w}_t$ is a cointegrated series with cointegration rank $p_w-r_F$.

\vspace{-0.60cm}
\subsection{SICS Predictor Selection}\label{sec:sics}

In this subsection, we propose a Sure Independence Canonical Screening (SICS) procedure to select the active predictors.
We firstly illustrate the motivation of using canonical correlation for screening briefly.
By standard calculations, it can be shown that for stationary response $y_t$ and any unit-root predictor $w_{i,t},$ the absolute correlation satisfies
\begin{equation}
    \begin{aligned}
        &\quad~
         \bigg(\sum_{t=1}^T (y_{t+1}-\bar{y})^2\sum_{t=1}^T (w_{i,t}-\bar{w}_{i})^2\bigg)^{-{1\over 2}}\bigg|\sum_{t=1}^T (w_{i,t}-\bar{w}_{i}) (y_{t+1}-\bar{y})\bigg|
        \\
         &= \bigg({1\over T}\sum_{t=1}^T (y_{t+1}-\bar{y})^2\bigg)^{-{1\over 2}}\bigg({1\over T^2}\sum_{t=1}^T (w_{i,t}-\bar{w}_{i})^2\bigg)^{-{1\over 2}}\bigg|{1\over T^{3\over 2}}\sum_{t=1}^T (w_{i,t}-\bar{w}_{i}) (y_{t+1}-\bar{y})\bigg| =o_p(1).
    \end{aligned}\label{pear0}
\end{equation}
This shows that we cannot directly rank the absolute correlation to screen unit-root predictors due to the nonstationarity. Figure \ref{fig:sis2}
shows the boxplot of simulated absolute correlations between $y_t$ and $x_{i,t}, ~i=1,2,\cdots,p$ using the Data Generation Process 2 (DGP2) in Section \ref{sec:4.2}, where active predictors are set as $x_{15,t}, x_{29,t}, x_{44,t}, x_{58, t}~ (I(0)); ~x_{72+i,t},~i=1,2,\cdots,8~(I(1)),$ and $p_z = p_w=72.$ It is evident that the original SIS procedure will fail to select active $I(1)$ predictors in most replications.
\begin{figure}
    \centering
    \begin{subfigure}[b]{0.48\textwidth}
        \centering
        \includegraphics[width=\textwidth]{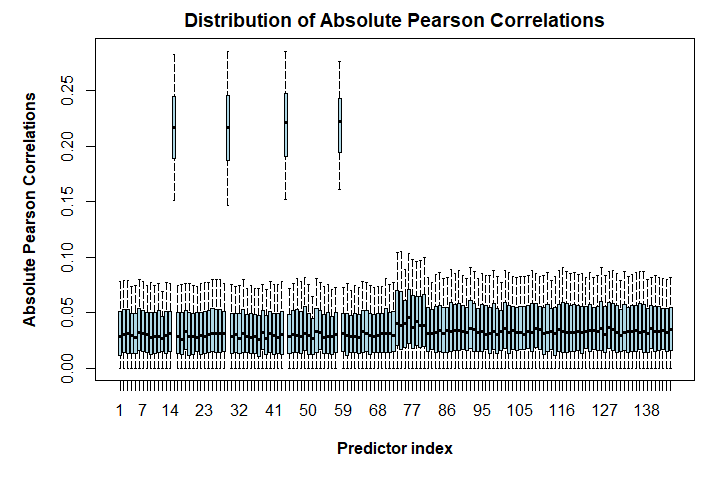}
        \caption{Absolute Correlations}
        \label{fig:sis2}
    \end{subfigure}
    \hfill 
    \begin{subfigure}[b]{0.48\textwidth}
        \centering
        \includegraphics[width=\textwidth]{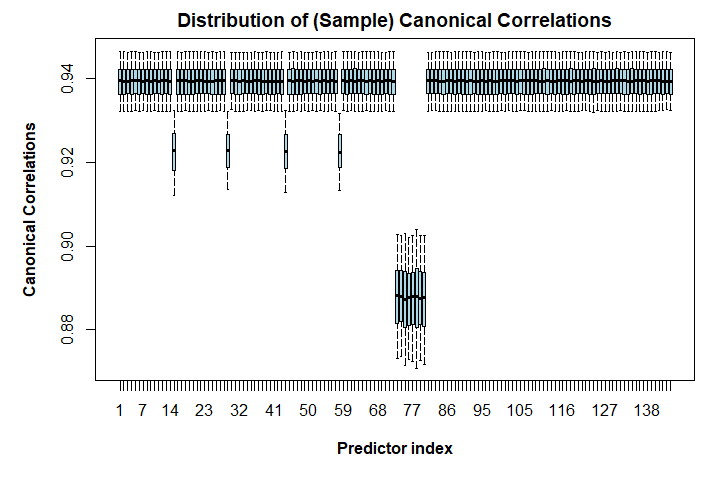}
        \caption{Canonical Correlations}
        \label{fig:sis}
    \end{subfigure}
    \caption{Boxplots of Correlations under DGP 2 in Section 4, with $T = 600,~ S^0 = \{150,300,450\},~ p_z= p_w=72$. Active predictors are $x_{15,t}, x_{29,t}, x_{44,t}, x_{58, t}~ (I(0)); ~x_{72+i,t},~i=1,2,\cdots,8~(I(1))$.}
    \label{fig:combined_sis}
\end{figure}

We therefore explore using the principle of canonical correlation analysis for our purpose. For the predictors vector $\boldsymbol{x}_t = (x_{1t}, \cdots, x_{pt})',$ denote $\boldsymbol{x}_{-i,t} = (x_{1t}, \cdots, x_{i-1,t},x_{i+1,t}, \cdots, x_{pt})'$ as the complement vector of predictors for $x_{i,t}$. The sample canonical correlation between $y_t$ and $\boldsymbol{x}_{-i,t}$ is given by
\begin{equation}
    \widehat{{\rm CC}}_{i} := \widehat{\boldsymbol{\Omega}}^{-1}_y \widehat{\boldsymbol{\Omega}}_{y, x_{-i}}
    \widehat{\boldsymbol{\Omega}}_{x_{-i}}^{-1}
    \widehat{\boldsymbol{\Omega}}_{x_{-i},y}
    ,\label{CCiexpress}
\end{equation}
where $\widehat{\boldsymbol{\Omega}}_y = T^{-1}\sum_{t=1}^T (y_{t+1}-\bar{y})^2,~ \widehat{\boldsymbol{\Omega}}_{y, x_{-i}} = T^{-1}\sum_{t=1}^T (y_{t+1}-\bar{y}) (\boldsymbol{x}_{-i,t}-\bar{\boldsymbol{x}}_{-i})' = \widehat{\boldsymbol{\Omega}}_{x_{-i},y}',$ and $\widehat{\boldsymbol{\Omega}}_{x_{-i}} = T^{-1}\sum_{t=1}^T (\boldsymbol{x}_{-i,t}-\bar{\boldsymbol{x}}_{-i}) (\boldsymbol{x}_{-i,t}-\bar{\boldsymbol{x}}_{-i})'.$ 
It turns out the above defined canonical correlation $\widehat{{\rm CC}}_{j}$'s will provide useful information for predictor selection when cointegration relationships are possible to present in the predictors. Figure \ref{fig:sis} shows the boxplot of simulated $\widehat{\rm CC}_i$'s using the same DGP as that in generating Figure \ref{fig:sis2}. It is noticed that the corresponding canonical correlations of active predictors, whether they are $I(0)$ or $I(1)$, are in general obviously smaller than those of the inactive ones, even in the presence of structural breaks. Our SICS predictor selection procedure, which is defined in the following, is motivated by these observations.

Let
\begin{equation}
    G^0 = \{1\leq j \leq p:~\exists~1\leq i\leq m_0+1,~\text{such that} ~\gamma_{i,j}\neq 0\}\label{truemodel}
\end{equation}
be the predictor set in the true sparse model, where $\boldsymbol{\gamma}_i=(\gamma_{i,1},\cdots,\gamma_{i,p})'$ is given by \eqref{model}. Our SICS procedure estimate $G^0$ by
\begin{equation}
    \widehat{G} = \{1\leq j \leq p:~\widehat{{\rm CC}}_j~\text{is among the first}~d_T~\text{smallest of all}\},\label{dt}
\end{equation}
for some positive integer $d_T\rightarrow\infty$ as $T\rightarrow\infty.$

\begin{remark}
    For the choice of $d_T$ in practice, we can rank $\{\widehat{{\rm CC}}_i\}_{i=1}^p$ by $\widehat{{\rm CC}}_{j_1}\leq \cdots\leq \widehat{{\rm CC}}_{j_p}$ and plot them in a figure. 
    The differences $\widehat{{\rm CC}}_{j_{i+1}}-\widehat{{\rm CC}}_{j_i},~i=1,2,\cdots,p-1$ will tend to be negligible for sufficiently large $i$, and one can select $d_T$ based on this figure.
\end{remark}

\begin{remark}
     The similar canonical correlation screening method has also been discussed by \cite{kongSureScreeningRanking2017a}. They examine situations where some variables are marginally uncorrelated with the response but are jointly correlated with the response when combined with other variables. Our SICS procedure differs from theirs in that we utilize canonical correlations to address cointegration relationships in the context of non-stationary time series, whereas they focus on independent and identically distributed (i.i.d.) covariate variables and noises. Further, our method allows the existence of change points in a regression model, while \cite{kongSureScreeningRanking2017a} considers a model without change points.
\end{remark}

\begin{remark}
      We cannot compute canonical correlations using $k$ predictors for some fixed $k$ as in \cite{kongSureScreeningRanking2017a}, rather than using $(p-1)$ ones as in our method. For example, if $k=2,~p>4,$ and $(x_{1t}, x_{2t}, x_{3t})'$ are cointegrated active predictors with cointegration rank $1$ and cointegration vector $({1/\sqrt{3}})(1,1,1)'$, any linear combinations within groups with two of $(x_{1t}, x_{2t}, x_{3t})'$ will be $I(1).$ Then, the corresponding canonical correlations, say, between $y_t$ and $(x_{1t},x_{2t})',$ will be near zero.
\end{remark}

\subsection{RCRS Change Point Detection}

Let $x_{j_1,t},\cdots,x_{j_{d_T},t}$ be the estimated active predictors   in \eqref{dt} and $\tilde{\boldsymbol{x}}_t := (x_{j_1,t},\cdots,x_{j_{d_T},t})'.$ We propose a novel Ratio-Controlled Regression Screening (RCRS) procedure to detect the change points based on  $\tilde{\boldsymbol{x}}_t$.  We first give some notations. \begin{itemize}
    \item For a given time period $(s:l)$, we define the (local) Residual Sum of Squares (RSS) for the least square estimators by
\begin{equation}
    {\rm RSS}(s:l) = \sum_{t=s}^{l-1} [y_{t+1} - \bar{y}^{(s:l)} -( \widehat{\boldsymbol{\gamma}}^{(s:l)})' (\tilde{\boldsymbol{x}}_{t} - \bar{\tilde{\boldsymbol{x}}}^{(s:l)})]^2, \label{RSSsl}
\end{equation}
where
$
    \widehat{\boldsymbol{\gamma}}^{(s:l)} = [\widehat{\boldsymbol{\Omega}}^{(s:l)}_{\tilde{x}}]^{-1}[\widehat{\boldsymbol{\Omega}}^{(s:l)}_{\tilde{x}, y}]
$ denotes the least square estimator of coefficient of $\tilde{\boldsymbol{x}}_t$.
\item For a given set $S=\{t_1, \cdots, t_k\},$ denote $t_0=1, t_{k+1}=T+1,$ we define the global RSS with locations set $S$ by
\begin{equation}
    {\rm RSS}(S) = \sum_{i=1}^{k+1} {\rm RSS}(t_{i-1}:t_i); \label{RSSofS}
\end{equation}
we define the $r-$neighborhood set of $S$ by \begin{equation}\mathcal{U}(S, r) =  \{t: \exists t_i \in S\cup \{1,T\}, t_i - r \leq t < t_i + r\}.\label{Usr}\end{equation}
\end{itemize}

The RCRS procedure relies on the RSS calculated over a sequence of local moving windows for initial identification of potential break points. Specifically, let \( h_T \) be a positive integer representing the radius of the local window. 
For a given interior time \( l \), consider the local window defined by \( (l - h_T : l + h_T) \). Intuitively, when this local window contains no breaks, the local RSS calculated over this window, denoted as \( \text{RSS}(l - h_T : l + h_T) \), will be small. Conversely, if there are true change points within this window, the RSS will be large. To identify potential break points, we can search for the local maxima of the RSS over all possible interior times \( l \). 
We further propose the RSS ratio (RSSR) statistic, which computes the (ratio of) global RSS reductions when an additional break point is included. This allows us to retain only those local maxima that lead to significant reductions in the global RSS in our estimated locations set of breaks.

To be specific, given constant $C_h>0$, and $M_T\rightarrow \infty$ as $T\rightarrow\infty$, our 
RCRS procedure is defined as follows:
\begin{enumerate}
    \item Set $\widehat{S}_0 = \phi$, the empty set; and set $k = 1.$
    \item Compute all ${\rm RSS}(l-h_T:l+h_T), ~l = h_T+1, \cdots, T+1-h_T.$ 
    \item Set $\widehat{S}_k = \widehat{S}_{k-1} \cup \{\widehat{t}_k\},$ where
    \begin{equation}
        \widehat{t}_k = \arg\max\limits_{l \notin \mathcal{U}(\widehat{S}_{k-1}, C_h h_T)} {\rm RSS}(l-h_T:l+h_T). \label{hat_tk}
    \end{equation}
    If $k\geq2 ,$ compute the ratio of RSS difference ${\rm RSSR}_k = \dfrac{{\rm RSS}(\widehat{S}_{k-1})-{\rm RSS}(\widehat{S}_{k})}{{\rm RSS}(\widehat{S}_{k-2})-{\rm RSS}(\widehat{S}_{k-1})}.$
    \item
    \begin{itemize}
        \item If $k < M_T,$ take $k \leftarrow k + 1$ and return to step 3.
        \item If $k = M_T,$ stop the procedure and take $\widehat{S} = \widehat{S}_{k_0},$ where \beqn \label{k0}k_0=(\arg\min\limits_{i}{\rm RSSR}_{i})-1.\eeqn
    \end{itemize}
\end{enumerate}
When applying the above procedure, 
the positive constant $C_h$ specifies that two estimated break points cannot be too close to each other; $M_T$ controls the maximum total number of steps.
Figure \ref{fig:rssr} demonstrates how the RSS ratios can be utilized to identify breakpoints using a simulated example. For the DGP3 described in Section \ref{sec:4.2}, where four true breaks 
\begin{figure}[H]
    \centering
    \includegraphics[width=0.5\textwidth]{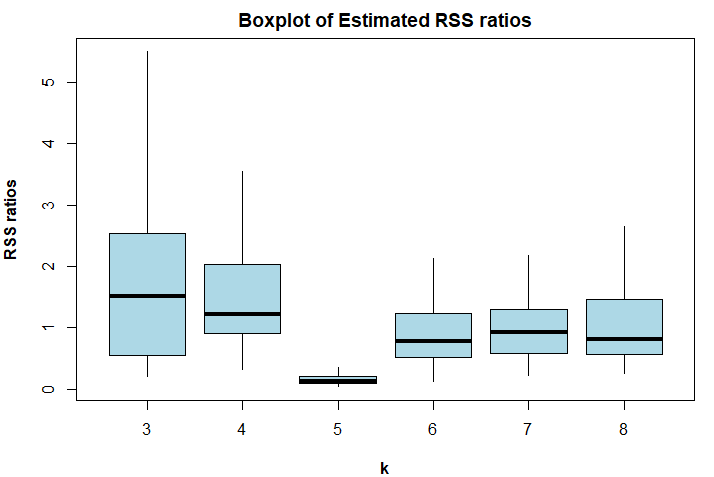}
    \caption{Boxplot of RSS ratios under DGP3 in Section \ref{sec:4.2}, with $T = 600,~ S^0 = \{120,240,360,480\}$. Here, $k$ denotes the steps of RCRS procedure.
    }
    \label{fig:rssr}
\end{figure} 
\vspace{-0.5cm}

exist (denoted by \(m_0 = 4\)), we present the boxplot of the RSS ratios ${\rm RSSR}_{k},~ k = 3,4,\cdots,8$. Note that the RSS ratio decreases sharply when \(k = 5 = m_0 + 1\). This indicates that RSS ratios can exclude many redundant points. 
During the first $m_0$ steps of RCRS procedure in most simulated cases, we can usually select one point within a small neighbourhood of one true break at each step.

\begin{remark}\label{remarktx}
Our SICS-RCRS procedure differs from the method described in \cite{tuPenetratingSporadicReturn2023} in two ways. First, we perform predictor selection before change points detection, whereas \cite{tuPenetratingSporadicReturn2023} directly searches for break points. To ensure the feasibility of LSE, The distance between two break points estimated by their procedure cannot be smaller than $p,$ where \( p \) represents the number of predictors. When \( p \) is large, their approach may overlook some true change points that are located close to each other. One can see Section \ref{comparetx} for more details. Second, \cite{tuPenetratingSporadicReturn2023} employs the extended Bayesian information criterion (EBIC) as outlined in \cite{chen2008extended} to determine when to stop screening, while we utilize RSS ratios for this purpose.
\end{remark}

\vspace{-0.60cm}
\subsection{Refining Procedures for the Break Points and Predictors}

The previously estimated sets of active predictors, denoted as \(\widehat{G}\), and breakpoints, denoted as \(\widehat{S}\), may contain redundant points. 
In this subsection, we propose two information criteria (IC) to eliminate these redundancies: one for the breakpoints and another for the predictors. First, we will eliminate the redundant breakpoints using a backward algorithm based on one information criterion. Next, we will address the redundant predictors in each segment defined by the estimated change points, using a different information criterion.

For a given locations set $S,$ we define
\begin{equation}
    {\rm IC}_1(S) := {\rm RSS}(S) + |S|\omega_{1T},
\end{equation}
where ${\rm RSS}(S)$ is given in \eqref{RSSofS}, 
$\omega_{1T}$ denotes a penalty term. The backward elimination algorithm for breaks goes as follows.
\begin{enumerate}
\item Let $K = |\widehat{S}|, ~{\rm IC}_{1,K} = {\rm IC}_1(\widehat{S}).$ Let
$
    {\bf t}_K = \{t_{K,1}, t_{K,2}, \cdots, t_{K,K}\} = \widehat{S}.
$

\item For $i = 1, 2, \cdots, K,$ compute ${\rm IC}_1({\bf t}_K \backslash t_{K,i})$ and set ${\rm IC}_{1,K-1} = \min\limits_{1\leq i \leq K}{\rm IC}_{1}({\bf t}_K\backslash t_{K,i}),$ where ${\bf t}_K \backslash t_{K,i} = \{t_{K,1},\cdots, t_{K,i-1}, t_{K,i+1}, \cdots, t_{K,K}\}.$

\item If ${\rm IC}_{1,K-1} > {\rm IC}_{1,K},$ then the estimated locations set of breaks is given by $\widehat{\widehat{S}}:={\bf t}_K.$

If ${\rm IC}_{1,K-1} \leq {\rm IC}_{1,K},~ K = 1,$ we conclude there exists no structural breaks.

If ${\rm IC}_{1,K-1} \leq {\rm IC}_{1,K}$ and $K > 1,$ set $
        j \leftarrow \arg\min\limits_i {\rm IC}_1({\bf t}_K \backslash t_{K,i}),~
        {\bf t}_{K-1} \leftarrow {\bf t}_K \backslash t_{K-1,j},~
        K \leftarrow K - 1,
$
then go to step 2.
\end{enumerate}
Finally, we get an estimated locations set of breaks $\widehat{\widehat{S}}= \Big\{\widehat{\widehat{t}}_1, \cdots, \widehat{\widehat{t}}_{\widehat{\widehat{m}}}\Big\}.$

Then, we give the following predictor elimination algorithm in the estimated regime $(\widehat{\widehat{t}}_{i-1}:\widehat{\widehat{t}}_{i})$ for some $1\leq i\leq \widehat{\widehat{m}}+1,$ where $\widehat{\widehat{t}}_0 = 1, \widehat{\widehat{t}}_{\widehat{\widehat{m}}+1}=T+1.$
For any given covariate variable set $G$, we define the Residual Sum of Squares using predictors in $G$ and samples in $(\widehat{\widehat{t}}_{i-1}:\widehat{\widehat{t}}_{i})$ by
\begin{equation}
    {\rm RSS}_i(G) = \sum_{t= \widehat{\widehat{t}}_{i-1}}^{\widehat{\widehat{t}}_{i}-1} \bigg(y_{t+1}-\bar{y}^{(\widehat{\widehat{t}}_{i-1}:\widehat{\widehat{t}}_{i})}-(\widehat{\boldsymbol{\gamma}}_i^G)' (\boldsymbol{x}_{Gt} - \bar{\boldsymbol{x}}_G^{(\widehat{\widehat{t}}_{i-1}:\widehat{\widehat{t}}_{i})})\bigg)^2,
\end{equation}
where $\boldsymbol{x}_{Gt}$ gathers the variables in $G,$ and
$
    \widehat{\boldsymbol{\gamma}}_i^G= \bigg(\widehat{\boldsymbol{\Omega}}^{(\widehat{\widehat{t}}_{i-1}:\widehat{\widehat{t}}_{i})}_{x_G}\bigg)^{-1} \bigg(\widehat{\boldsymbol{\Omega}}^{(\widehat{\widehat{t}}_{i-1}:\widehat{\widehat{t}}_{i})}_{x_G,y}\bigg)
$ denotes the least square estimator of coefficient of $\boldsymbol{x}_{Gt}$.
Define
$$
    {\rm IC}_2^i(G) = {\rm RSS}_i(G) + |G| \omega_{2T}^i,
$$
where $\omega_{2T}^i$ denotes the penalty term.
Our predictor elimination algorithm goes as follows.
\begin{enumerate}
\item Let $K = |\widehat{G}|, ~{\rm IC}_{2,K}^i = {\rm IC}_2^i(\widehat{G}),
    ~G_K := \{x_{K,1}, x_{K,2}, \cdots, x_{K,K}\} = \widehat{G}.
$

\item For $j = 1, 2, \cdots, K,$ compute ${\rm IC}_2^i(G_K \backslash x_{K,j})$ and set ${\rm IC}_{2, K-1}^i=\min\limits_{1\leq j \leq K}{\rm IC}_2^i(G_K \backslash x_{K,j}),$ where $G_K\backslash x_{K,j}=G_K\cap\{x_{K,j}\}^c.$

\item If ${\rm IC}_{2, K-1}^i > {\rm IC}_{2, K}^i,$ the estimated set of predictors in $(\widehat{\widehat{t}}_{i-1}:\widehat{\widehat{t}}_{i})$ is given by $ \widehat{\widehat{G}}_i := G_K.$

If ${\rm IC}_{2, K-1}^i \leq {\rm IC}_{2, K}^i$ and $K = 1,$ we conclude that $ \widehat{\widehat{G}}_i = \phi$.

If ${\rm IC}_{2,K-1}^i \leq {\rm IC}_{2,K}^i, K > 1,$ we set $
        l \leftarrow \arg\min\limits_j {\rm IC}_2^i(G_K \backslash x_{K,j}),
        ~{G}_{K-1} \leftarrow  G_K \backslash x_{K,l},\\
        ~K \leftarrow K - 1.
$
Then, go to step 2.
\end{enumerate}
Finally, the estimator of predictor set in the time period $(\widehat{\widehat{t}}_{i-1}:\widehat{\widehat{t}}_{i})$ is given by $\widehat{\widehat{G}}_i.$ 

Our procedure can be summarized in the flowchart in Figure \ref{fig:procedure}.
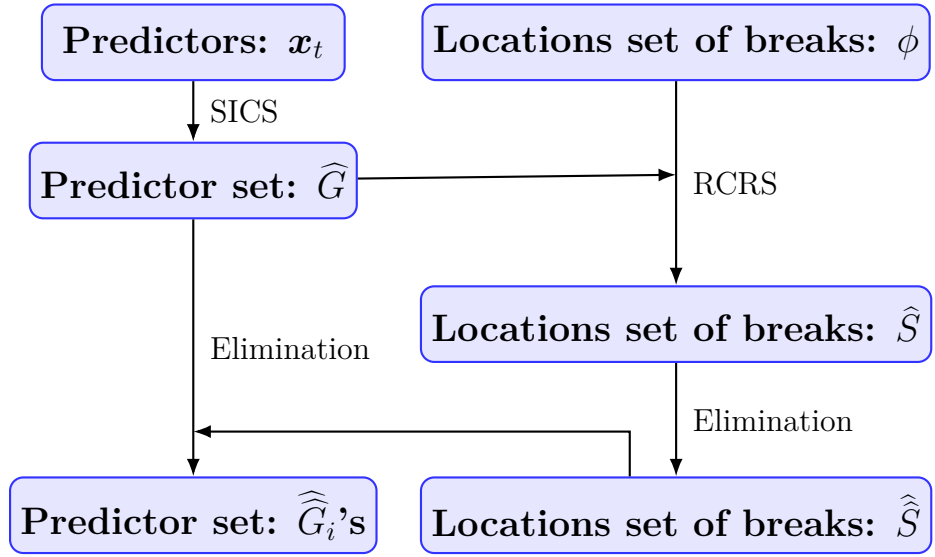
\begin{figure}[htbp]
  \centering
\begin{tikzpicture}[
    scale=0.4,
    box/.style={
        draw=blue!80,
        thick,
        fill=blue!10,
        minimum width=4cm,
        minimum height=1cm,
        rounded corners=5pt,
        font=\large\bfseries
    },
    node distance=0.5cm and 1cm,
    arrow/.style={thick, -Latex}
]

\node[box] (Q) {Predictors: $
\boldsymbol{x}_t$};
\node[box, right=of Q] (phi) {Locations set of breaks: $\phi$};
\node[box, below=0.8cm of Q] (Qhat) {Predictor set: $\widehat{G}$};
\node[box, below=3.4cm of Qhat] (Qhat2) {Predictor set: $\widehat{\widehat{G}}_i$'s};
\node[box, below=2.7cm of phi] (Shat) {Locations set of breaks: $\widehat{S}$};
\node[box, below=1.5cm of Shat] (Shat2) {Locations set of breaks: $\widehat{\widehat{S}}$};

\draw[arrow] (Q) -- node[right=2pt] {SICS} (Qhat);
\draw[arrow] (Qhat) -- node[right=2pt] {Elimination} (Qhat2);
\draw[arrow] (phi) -- node[right=2pt] {RCRS} (Shat);
\draw[arrow] (Shat) -- node[right=2pt] {Elimination} (Shat2);

\coordinate (midElimS) at ($(phi)!0.47!(Shat)$);
\draw[arrow] (Qhat) -- (midElimS);


\coordinate (target-point) at ($(Qhat)!0.75!(Qhat2)$);
\draw[arrow] (Shat2.500) |- (target-point);
\end{tikzpicture}

\caption{Flowchart of our procedure.}
\label{fig:procedure}
\end{figure}
\begin{remark}
We can combine our SICS with predictor elimination as an effective predictor selection method. Our approach first reduces the dimensionality to a moderately large one, and then refines the predictor set through elimination. This procedure can overcome the weakness of LASSO methods by  \cite{leeLASSOPredictiveRegression2022a} and \cite{mei2024lasso}, which fails to select effective cointegrated predictors. We will also demonstrate this in the simulation experiment presented in Section \ref{sec:4.4}.
\end{remark}

\vspace{-0.60cm}
\section{ASYMPTOTIC THEORY}
\label{sec:asym}

In this section, we investigate the asymptotic properties of the estimated  predictors and change points. 
In Section \ref{sec:3.1}, we show that the sure screening property holds for SICS, i.e., all the true active predictors are contained in the selected set $\widehat{G}$ defined in (\ref{dt}) with probability tending to one.
In Section \ref{sec:3.2}, we show that each true change point can be consistently estimated by one element in the estimated set $\widehat{S}$ drawn by the proposed RCRS procedure. 
In Section \ref{sec:3.3}, we show that our refinement procedures based on information criteria can eliminate redundancies for both change points and predictors.
\vspace{-0.60cm}
\subsection{On the SICS Estimated Predictor Set $\widehat{G}$}
\label{sec:3.1}

We first introduce regularity conditions for the generating processes.

{\noindent\bf Assumption 1. }
\vspace{-0.35cm}
\begin{itemize}

    \item[(1)] Suppose that there exists an $r_F-$dimensional vector $\boldsymbol{g}_t$ with mean zero and independent components such that $\boldsymbol{v}_t = \boldsymbol{\Xi} \boldsymbol{g}_t,$ where $\boldsymbol{\Xi}$ has size $r_F\times r_F, ~\|\boldsymbol{\Xi}\| < \infty.$ For each component $g_t^i$ of $\boldsymbol{g}_t,$ there exists an independent and standard normal sequence $\{\xi_t^i\}$ and $\tau \in (0,1/2),$ for which as $T \rightarrow \infty,$
    \begin{equation}
        \max\limits_{1\leq i \leq r_F} \max\limits_{1\leq s<l\leq T, r_F(l-s)^{\tau-1/2}\rightarrow 0} {\rm E}\biggl[\sum_{t=s}^{l-1} \{g_t^i - \bar{g}_i^{(s:l)} - \sigma_{ii}(\xi_t^i - \bar{\xi}_i^{(s:l)})\}
            \biggr]^2 = O((l-s)^{2\tau}), \label{strapp_cond}
    \end{equation}
    where $b_1 \leq \sigma_{ii}^2 := \lim\limits_{T\rightarrow \infty}{\rm Var}\bigg(T^{-{1\over 2}}\sum_{s=1}^T g_s^i\bigg)\leq b_2$ holds for any $i$ and some $0<b_1<b_2.$

    \item[(2)]  Define $\boldsymbol{\zeta}_t := (\boldsymbol{z}_t', u_t, \boldsymbol{e}_t')' := (\zeta_t^1,\cdots,\zeta_t^{p_z+p_w+1}). ~{\rm E} \boldsymbol{\zeta}_t = \boldsymbol{0},~ \max\limits_{i}{\rm E}|\zeta_t^i|^{4+2\kappa} < \infty$
    for some $\kappa > 0. ~\boldsymbol{\zeta}_t$ is $\alpha-$mixing with mixing coefficients $\alpha(k)$ satisfying
    $\sum_{k=1}^\infty [\alpha(k)]^{1-1/( 2+\kappa)} < \infty.~ \lambda_{\min}({\rm Cov}(\boldsymbol{\zeta}_t))\asymp\lambda_{\max}({\rm Cov}(\boldsymbol{\zeta}_t)) \asymp 1,$ where $a\asymp b$ denotes that $a$ and $b$ have the same order, i.e., $a=O(b)$ and $b=O(a)$.

    \item[(3)] $\{\boldsymbol{v}_t\}, \{\boldsymbol{z}_t\}, \{\boldsymbol{e}_t\}, \{u_t\}$ are independent from each other.
\end{itemize}

\begin{remark}
    Assumption 1~(1) is essential in building the limit performance of $(l-s)^{-2}\sum_{t=s}^{l-1} (\boldsymbol{F}_t - \bar{\boldsymbol{F}}^{(s:l)})(\boldsymbol{F}_t - \bar{\boldsymbol{F}}^{(s:l)})'$ using strong approximation method.
    One can refer to Condition 2 and Lemma 3 in \cite{zhangIdentifyingCointegrationEigenanalysis2019} for more details. We allow the dimensionality of $\boldsymbol{F}_t,~ i.e.~ r_F,$ goes to infinity. Assumption 1~(2) is standard and satisfied by many common processes. 
    Assumption 1 (3) can be replaced by more general conditions to allow for mild dependence between $(\{\boldsymbol{w}_t\},\{u_t\})$ and between $(\{\boldsymbol{z}_t\}, \{\boldsymbol{w}_t\})$.
    \label{remarkF}
\end{remark}

Let $\boldsymbol{A} = \begin{bmatrix}
        \boldsymbol{I}_{p_z} & {\boldsymbol{O}} \\
        {\boldsymbol{O}} & (\boldsymbol{Q}^{\perp})' \\
        {\boldsymbol{O}} & \boldsymbol{Q}'
    \end{bmatrix}$, then $\boldsymbol{A}'\boldsymbol{A} = \boldsymbol{A}\boldsymbol{A}'=\boldsymbol{I}_p$ and
\begin{equation}
    \boldsymbol{A} \boldsymbol{x}_t  
    =(\boldsymbol{z}'_t, \,\boldsymbol{w}_t'\boldsymbol{Q}^{\perp},\, \boldsymbol{w}'_t\boldsymbol{Q})'
    :=(\boldsymbol{z}'_t, \, \boldsymbol{\psi}_{1t}', \, \boldsymbol{\psi}_{2t}')'
 := \boldsymbol{\psi}_t,
    \label{cointA}
\end{equation}
where $\boldsymbol{\psi}_{1t} $ represents a $(p_w - r_F)-$dimensional $I(0)$ process and $\boldsymbol{\psi}_{2t} $ represents an $r_F-$dimensional $I(1)$ process.
Using \eqref{cointA}, we can rewrite Model~(\ref{model}) as
\begin{equation}
    \begin{aligned}
        y_{t+1}
        &= \sum_{i=1}^{m_0+1} \boldsymbol{\gamma}_i'\boldsymbol{A}'\boldsymbol{A} \boldsymbol{x}_t 1_{\{t_{i-1}^0 \leq t < t_i^0\}} + u_{t+1}\\
        &:= \sum_{i=1}^{m_0+1} (\boldsymbol{\alpha}_i',\boldsymbol{\beta}_i',\boldsymbol{0}') \boldsymbol{\psi}_t 1_{\{t_{i-1}^0 \leq t < t_i^0\}} + u_{t+1},
    \end{aligned}\label{modelf}
\end{equation}
where $\boldsymbol{\alpha}_i,~\boldsymbol{\beta}_i$ denote the coefficient vector of $\boldsymbol{z}_t,~ \boldsymbol{\psi}_{1t},$ respectively.
\begin{remark}
    In \eqref{modelf}, the coefficient of $\boldsymbol{\psi}_{2t},$ the non-cointegrated $I(1)$ vector, is assumed to be zero. In fact, this condition can be relaxed to allow for some nonzero but shrink-to-zero coefficients for $\boldsymbol{\psi}_{2t}$. It will not affect the main theoretical results of this paper.
\end{remark}
The following Assumption is used for deriving the sure screening property of SICS procedure.

{\noindent\bf Assumption 2.}
\vspace{-0.35cm}
\begin{itemize}
\item[(1)] $m_0p T^{-1/2} \rightarrow 0,~m_0=O(\log T).$ 

\item[(2)]  $\lim\limits_{T\rightarrow\infty}\sum_{i=1}^{m_0+1} {t_i^0-t_{i-1}^0\over T}  (\boldsymbol{\alpha}_i', \boldsymbol{\beta}_i')'= (\tilde{\boldsymbol{\alpha}}', \tilde{\boldsymbol{\beta}}')',$ where $\|(\tilde{\boldsymbol{\alpha}}', \tilde{\boldsymbol{\beta}}')'\|\geq C_3$ for some $C_3>0$.

\item[(3)] $|G^0| < \infty,$ where $G^0$ is defined in \eqref{truemodel}. For some $C_0> 0,~ M_\delta := \min\limits_{ k \in G^0, l \notin G^0} ({\rm CC}_l-{\rm CC}_k) > C_0,$ where \begin{equation}
    {\rm CC}_l := {1\over \sigma_y^2} [{\rm Cov}(\tilde{y}_t, \boldsymbol{x}_{-l,t})] [{\rm Cov}(\boldsymbol{x}_{-l,t})]^{-1}[{\rm Cov}(\boldsymbol{x}_{-l,t}, \tilde{y}_t)],~l=1,2,\cdots,p, \label{pcc0}
\end{equation}
$\tilde{y}_t = \tilde{\boldsymbol{\alpha}}'\boldsymbol{z}_t+\tilde{\boldsymbol{\beta}}'\boldsymbol{\psi}_{1t}, ~\sigma_y^2 = \lim\limits_{T\rightarrow\infty}\sum_{i=1}^{m_0+1}{t_i^0-t_{i-1}^0\over T}\bigg[\boldsymbol{\alpha}_i'{\rm Cov}(\boldsymbol{z}_t)\boldsymbol{\alpha}_i + \boldsymbol{\beta}_i'{\rm Cov}(\boldsymbol{\psi}_{1t})\boldsymbol{\beta}_i + {\rm E}u_{t}^2\bigg]<\infty.$
\end{itemize}

\begin{remark}
    Assumption 2 (1) gives the maximum order for $m_0$ and $p.$
     Assumption 2 (2) excludes the case where the limit of  weighted-average coefficient vector, i.e., $(\tilde{\boldsymbol{\alpha}}', \tilde{\boldsymbol{\beta}}')'$, is a zero vector.
    Assumption 2 (3) on $M_\delta$  constrains that the total signal strength cannot be too small.
    Similar Assumptions on signal strength are common in screening literature, see Condition 3 in \cite{fanSureIndependenceScreening2008} for example. It will be shown in the supplementary material that $\widehat{\rm CC}_l$ converges to ${\rm CC}_l$ for $l=1,2,\cdots,p$ in probability. Assumption 2 (3) can be relaxed to allow $|G^0|$ to grow with $T$, and $C_0$ can be replaced by $C_0T^{-\tau_0}$ for some $\tau_0>0.$ When there are no change points, Assumption 2 (2) holds naturally if $\boldsymbol{x}_t$ exists predictability, and ${\rm CC}_l$ in Assumption 2 (3) is equal to the classical canonical correlation between $\{y_t\}$ and $\{\boldsymbol{x}_{-l,t}\}.$
\end{remark}
We have the following Theorem 3.1 for  the sure screening property of SICS procedure.
\begin{theorem}
    Suppose $r_F=o(T^{1/2-\tau})$ holds for $\tau$ in Assumption 1 (1). Under Assumptions 1,2, as $T \rightarrow \infty,$ we have
        \begin{equation}P(G^0 \subset\widehat{G})\longrightarrow 1.\end{equation}
    \label{thm3-1}
\end{theorem}
\vspace{-0.5cm}
\begin{remark}
Theorem 3.1 indicates that with probability tending to one, our SICS procedure can select all the predictors which enjoy predictability in at least one regime.
It can be noted that SICS procedure greatly reduces the dimension of covariate variables, leading to a more accurate estimate for locations of breaks.
\end{remark}
\vspace{-0.60cm}
\subsection{On the RCRS Estimated Breaks Set  $\widehat{S}$}
\label{sec:3.2}
To derive the asymptotic properties of the estimated change points, we need the following further condition.

{\noindent\bf Assumption 3.}
As $T\rightarrow \infty,~ r_F^2 h_T^{\tau-{1\over 2} }+r_Fd_Th_T^{-{1\over 2}} \rightarrow 0,~~r_F^2 h_T^2 \log T (\min\limits_{1\leq i \leq m_0+1}|t_i^0-t_{i-1}^0|)^{-1}\rightarrow 0
,$ where $d_T$ is given in \eqref{dt}, $h_T$ is given in \eqref{hat_tk}, and $\tau$ is given in Assumption 1 (1).
\begin{remark}
    Assumption 3 means that: (i) the number of predictors given by SICS procedure~i.e.~$d_T,$~cannot be too large, which is
    essential in estimating structural breaks; (ii) the bandwidth of change point screening~i.e.~$h_T,$~cannot be too large, in order to avoid omitting true points; (iii) the bandwidth $h_T$ cannot be too small to ensure the consistency of local least square estimates.
    The second condition in Assumption 3 is essential to give the limit property of RSS ratios as well. 
    Assumption 3 is mild. For instance, Assumption 3 holds when $m_0<\infty, ~r_F=O(\log T),~ d_T=O(\log T),~h_T\asymp T^{2/5}.$
\end{remark}
The sure screening property for RCRS procedure is given by the following Theorem 3.2.
\begin{theorem}\label{thm3-2}
    Under Assumptions 1-3, we have
    \begin{enumerate}
        \item for any $i = 1,2,\cdots,m_0,$ there exists $\widehat{t}_i \in \widehat{S}$, such that $$P(|\widehat{t}_i - t_i^0|\leq h_T)\rightarrow 1,~~ \text{as}~ T\rightarrow \infty;$$
        \item for $k'$ satisfying $S^0 \nsubseteq \mathcal{U}(\widehat{S}_{k'-1}, h_T),~ S^0\subset\mathcal{U}(\widehat{S}_{k'},h_T)$, \begin{equation}{\rm RSSR}_{k'+1} 
            = o_p(1);~~
          {\rm RSSR}_{k'+i} \asymp 1, ~\forall i\geq 2.\label{3.7}
        \end{equation}
    \end{enumerate}
\end{theorem}

\begin{remark}
    Theorem 3.2 (1) indicates that with probability tending to 1, for any $t_i^0\in S^0,$ there exists at least one consistent estimator $\widehat{t}_i \in \widehat{S}$. 
    The condition $S^0 \nsubseteq \mathcal{U}(\widehat{S}_{k'-1}, h_T), S^0\subset\mathcal{U}(\widehat{S}_{k'},h_T)$ in Theorem 3.2 (2) means that: 
    (i) In the $(k'-1)-$th step, there exists one true change point $t_j^0$ satisfying $|t_j^0-\widehat{t}_l|\geq h_T$ for any $\widehat{t}_l\in \widehat{S}_{k'-1};$ 
    (ii) In the $k'-$th step, our selected $\widehat{t}_{k'}$ satisfies $|\widehat{t}_{k'}-t_j^0|\leq h_T$ with probability tending to one. 
    As long as $S^0 \nsubseteq \mathcal{U}(\widehat{S}_k,h_T)$ for some $k,$ it is shown in the proof of Theorem 3.2 that $P(\widehat{t}_{k+1}\in \mathcal{U}(S^0, h_T))\rightarrow 1,$ which implies that our RCRS procedure can give a consistent estimator for one unidentified true change point in the $(k+1)$-th step. This fact ensures the existence of $k'$ in \eqref{3.7}.
    It is shown in Figure \ref{fig:rssr} that: when $S^0 \nsubseteq \mathcal{U}(\widehat{S}_{k'-1}, h_T), S^0\subset\mathcal{U}(\widehat{S}_{k'},h_T)$, ${\rm RSSR}_k$ will drop sharply when $k=k'+1;$
    when $k>k'+1, {\rm RSSR}_k$ will be away from zero. 
    This result motivates us to use $\{\widehat{t}_1,\cdots,\widehat{t}_{k'}\}$ as estimated locations set of breaks, avoiding many redundant points.
\end{remark}
\vspace{-0.60cm}
\subsection{On the Refinement Procedures}
\label{sec:3.3}
In this subsection, we establish the consistency of the second-step estimators for both change points and active predictors. To this end, we first give Assumptions for the parameters $\omega_{1T}$ and $\omega_{2T}^i$ defined in the information criteria.

{\noindent\bf Assumption 4.}
\vspace{-0.35cm}
\begin{itemize}

    \item[(1)] $m_0 \omega_{1T} r_F(\min\limits_{1\leq i\leq m_0+1}|t_{i}^0 - t_{i-1}^0|)^{-1} \rightarrow 0, ~m_0r_Fh_T^2 \omega_{1T}^{-1} \rightarrow 0$ as $T\rightarrow\infty.$

    \item[(2)] $\omega_{2T}^i(\min\limits_{1\leq i\leq m_0+1}|t_{i}^0 - t_{i-1}^0|)^{-1} \rightarrow 0,~r_Fh_T^2(\omega_{2T}^i)^{-1}\rightarrow 0$ holds for all $i$ as $T\rightarrow \infty.$

\end{itemize}

The following Theorem 3.3 is about the consistency of the second step estimators for  both the  number and locations  of the true breaks.

\begin{theorem}
    Under Assumptions 1-3 and 4(1), 
    \begin{itemize}
    \item[(1)] as $T\rightarrow\infty,$ we have
     $$P(|\widehat{\widehat{S}}| = m_0) \rightarrow 1;$$

    \item[(2)] for any $i = 1,2,\cdots,m_0,$ there exists $\widehat{\widehat{t}}_i \in \widehat{\widehat{S}}$, such that
        $$P(|\widehat{\widehat{t}}_i - t_i^0|\leq h_T)\rightarrow 1,~~\text{as}~T\rightarrow \infty.$$
    \end{itemize}
\end{theorem}
\begin{remark}
    Theorem 3.3~(1) states that we are able to estimate the true number of breaks $m_0$ with probability tending to one after backward elimination. Further, Theorem 3.3~(2) concludes that with probability tending to one, one can find a consistent estimator for any true change point $t_i^0 \in S^0$ in its neighborhood.
\end{remark}
Similar to \eqref{truemodel}, denote
$$
    G_i^0 = \{1\leq j \leq p: \gamma_{i,j} \neq 0\},~i=1,2,\cdots,m_0+1,
$$
and
$\widehat{\widehat{G}}_i$ be the estimator of predictor set given by predictor elimination procedure in the regime $(\widehat{\widehat{t}}_{i-1}: \widehat{\widehat{t}}_i)$. The next Theorem 3.4 gives the consistency of the second step estimators for the active predictor sets.
\begin{theorem}
     Under Assumptions 1-4, as $T \rightarrow \infty,$ we have
    \begin{equation}
        P(\widehat{\widehat{G}}_i = G_i^0) \rightarrow 1, ~~i=1,~2,~\cdots,~m_0+1.
    \end{equation}
\end{theorem}

\begin{remark}
    Theorem 3.4 states that our elimination procedure can select all active predictors in each regime consistently, as if the locations of change points are known. Here, $\widehat{\widehat{G}}_i = G_i^0$ denotes that the predictors in $\widehat{\widehat{G}}_i$ and $G_i^0$ are the same. We allow $G_1^0,G_2^0,  \cdots, G_{m_0+1}^0$ to be not exactly equal.
\end{remark}
\vspace{-0.60cm}
\section{SIMULATION}
\label{sec:simu}
In this section, we conduct simulation studies to show the performance of our 
procedure.

\vspace{-0.60cm}
\subsection{Parameters and Evaluation Criteria}\label{sec:4.1}

We first discuss how we set the tuning parameters in our simulation and define the criteria necessary to evaluate the effectiveness of our method. All simulations will be conducted using $R=2,000$ Monte Carlo replications.
\begin{itemize}
    \item[(1)] For SICS procedure, we use $d_{T} = [j_p\log T]+1$ for some prescribed $j_p.$ Denote $\widehat{G}^{(r)}$ be the estimated predictor set selected by SICS in the $r-$th replication.
    We use coverage rate (CR) defined by $
    {\rm CR} = {1\over R}\sum_{r=1}^R \dfrac{| G^0 \cap \widehat{G}^{(r)}|}{|G^0|}$
    to measure how likely we can capture all relevant predictors.
    \item[(2)] Then, in our RCRS procedure, we set $h_T = [10T^{1/5}], M_T= 2\log T, C_h=1.25.$ These settings satisfy our Assumptions, as we will set $\min\limits_{1\leq i \leq m_0+1}|t_i^0-t_{i-1}^0|\asymp T$ in the following data generation processes (DGPs). 
    In break elimination, we choose $\omega_{1T} = \sqrt{T}.$
    For the eliminated locations set of breaks in the $r$-th replication $\widehat{\widehat{S}}^{(r)}$ and the true locations set $S^0,$ we define the Hausdorff distance (HD) between $\widehat{\widehat{S}}^{(r)}$ and $S^0$ by $${\rm HD} = {1\over R}\sum_{r=1}^R {\max}\bigg\{d(\widehat{\widehat{S}}^{(r)}, S^0),~d(S^0,\widehat{\widehat{S}}^{(r)})\bigg\},~ d(S_1, S_2)=\max\limits_{a\in S_1}\min\limits_{b\in S_2}|a-b|.$$ Similar to \cite{tuPenetratingSporadicReturn2023}, we report $100T^{-1} \times {\rm HD}$ in our following tables. 
    \item[(3)] To measure how likely we can choose breaks precisely, we record the 
    percentages of correct estimation (PCE) defined by 
    $
    {\rm PCE} = {1\over R} \sum_{r=1}^{R} 1_{\{\widehat{\widehat{m}}^{(r)} = m_0\}},$ where $
    \widehat{\widehat{m}}^{(r)}=|\widehat{\widehat{S}}^{(r)}|$ denotes the estimated number of breaks after IC elimination.
    \item[(4)] For the elimination of redundant regressors in the period $(\widehat{\widehat{t}}_{i-1}: \widehat{\widehat{t}}_i)$ determined by $\widehat{\widehat{S}}$, we choose $\omega_{2T}^i = 1.5(\widehat{\widehat{t}}_i-\widehat{\widehat{t}}_{i-1})^{1/ 2}.$ We use the root mean squared error (RMSE) to measure the performance of coefficient estimates. When $\widehat{\widehat{m}}^{(r)} = m_0,$ i.e. we capture the number of breaks correctly, we define $ {\rm RMSE} = \bigg[{1\over R}\sum_{r=1}^R \sum_{i=1}^{m_0+1}||\widehat{\boldsymbol{\gamma}}_{i}^{(r)} - \boldsymbol{\gamma}_{i}||^2\bigg]^{1/2},$ where $\widehat{\boldsymbol{\gamma}}_{i}^{(r)}$ denotes the estimated coefficient vectors of $\boldsymbol{x}_{t}$ 
    in the regime $(\widehat{\widehat{t}}_{i-1}: \widehat{\widehat{t}}_i)$ of $r-$th replication.
\end{itemize}
\vspace{-0.35cm}
\subsection{Data Generation Processes and Results}\label{sec:4.2}

For \( j_p = 2, 3, 4 \), we define \( p_z = p_w = \left[(j_p + 1) T^{0.45}\right] + 1 \) and \( r_F = \left[(j_p + 3) \log T\right] + 1 \). We then generate independent vector innovation series \( (u_t, \boldsymbol{v}_t', \boldsymbol{e}_t', \boldsymbol{e}_{z,t}')' \) from the multivariate normal distribution \( \boldsymbol{N}_{1 + r_F + p_w + p_z}({\bf 0}, 2\boldsymbol{I}) \) for \( t = 1, 2, \cdots, T \). In this context, \( u_t \), \( \boldsymbol{v}_t \), \( \boldsymbol{e}_t \), and \( \boldsymbol{e}_{z,t} \) have respective dimensionalities of \( 1 \), \( r_F \), \( p_w \), and \( p_z \).

We generate $\boldsymbol{z}_{t}$ by $\boldsymbol{z}_{t} = {\rm diag}\{\kappa_1, \cdots, \kappa_{p_z}\} \boldsymbol{z}_{t-1} + \boldsymbol{e}_{z,t},$ where $\boldsymbol{z}_{0} = {\bf 0}_{p_z},~\kappa_1, \kappa_2, \cdots, \kappa_{p_z}$ are drawn from a uniform distribution $U[0.4,0.6].$
For $\boldsymbol{w}_t$, we first generate matrices $\boldsymbol{Q}$ and $\boldsymbol{Q}^\perp.$
Define two $p_w$ dimensional vector $$\boldsymbol{q}_1^*=(1/2,0,1/2,0,1/2,0,1/2,0,0,\cdots,0)',
\boldsymbol{q}_2^*= (0,1/2,0,1/2,0,1/2,0,1/2,0,\cdots,0)',$$ We generate $(p_w-2)$ vectors $\boldsymbol{a}_3, \cdots, \boldsymbol{a}_{p_w}\in \mathbb{R}^{p_w}$ whose elements are drawn from $N(0,1)$ distribution, and then compute $(\boldsymbol{q}_1^*, \boldsymbol{q}_2^*, \cdots, \boldsymbol{q}_{p_w}^*)$ using $(\boldsymbol{q}_1^*, \boldsymbol{q}_2^*,\boldsymbol{a}_3, \cdots, \boldsymbol{a}_{p_w})$ by Gram-Schmidt Algorithm. 
We then set $\boldsymbol{Q}^\perp=(\boldsymbol{q}_1^*, \cdots, \boldsymbol{q}_{p_w-r_F}^*),~ \boldsymbol{Q}= (\boldsymbol{q}_{p_w-r_F+1}^*, \cdots, \boldsymbol{q}_{p_w}^*).$ 
Then, we generate $\boldsymbol{F}_t$ and $\boldsymbol{w}_t = \boldsymbol{Q}\boldsymbol{F}_t + \boldsymbol{e}_t$ using $(\boldsymbol{v}_t',\boldsymbol{e}_t')'$ generated above.
Note that $\boldsymbol{\psi}_{1t}$ defined in \eqref{cointA} can be represented by $\boldsymbol{\psi}_{1t} = (\boldsymbol{Q}^\perp)' \boldsymbol{w}_t = (\boldsymbol{Q}^\perp)' \boldsymbol{e}_t$.
Under the above specifications, $(w_{1t}, w_{3t}, w_{5t}, w_{7t})'$ and $(w_{2t}, w_{4t}, w_{6t}, w_{8t})'$ are cointegrated respectively with a same cointegration vector $(1/2,1/2,1/2,1/2)'$.

Let $\boldsymbol{z}_{pre,t} = (z_{[p_z/5]+1,t}, z_{[2p_z/5]+1,t}, z_{[3p_z/5]+1,t}, z_{[4p_z/5]+1,t})', \boldsymbol{w}_{pre,t} = (w_{1t}, \cdots, w_{8t})'.$ Our DGPs are defined
as follows.

{\bf DGP 1 with two real breaks $S^0=\{[{T/ 3}], [2T/ 3]\}:$}
\begin{equation}
    y_{t+1} = \left\{
        \begin{array}{lll}
        (1, -0.5, 0.5, 1) \boldsymbol{z}_{pre,t} +
        (0.75, 0.75, \cdots, 0.75) \boldsymbol{w}_{pre,t} + u_{t+1},
         &0 < t \leq [T/3];\\
        (0.5, 1, 1, 0.5) \boldsymbol{z}_{pre,t} +
        (1.5, 1.5, \cdots, 1.5) \boldsymbol{w}_{pre,t} + u_{t+1}, &[T/3] < t \leq [2T/3];\\
        (1, -0.5, 0.5, 1) \boldsymbol{z}_{pre,t} + (0.75, 0.75, \cdots, 0.75) \boldsymbol{w}_{pre,t} + u_{t+1}, &[2T/3] < t \leq T.
    \end{array}\right.
    \nn
    \label{dgp4.2}
\end{equation}

{\bf DGP 2 with three real breaks $S^0=\{[T/4],[T/2],[3T/4]\}:$}
\begin{equation}
y_{t+1} = \left\{
        \begin{array}{llll}
       (1, -0.5, 0.5, 1) \boldsymbol{z}_{pre,t} +
        (0.75, 0.75, \cdots, 0.75) \boldsymbol{w}_{pre,t} + u_{t+1},
        &0< t \leq [T/4];\\
        (0.5, 1, 1, 0.5) \boldsymbol{z}_{pre,t} +
        (1.5, 1.5, \cdots, 1.5) \boldsymbol{w}_{pre,t} + u_{t+1}, &[T/4] < t \leq [T/2];\\
        (1, -0.5, 0.5, 1) \boldsymbol{z}_{pre,t} +
        (0.75, 0.75, \cdots, 0.75) \boldsymbol{w}_{pre,t} + u_{t+1},  &[T/2] < t \leq [3T/4],\\
        (0.5, 1, 1, 0.5) \boldsymbol{z}_{pre,t} +
        (1.5, 1.5, \cdots, 1.5) \boldsymbol{w}_{pre,t} + u_{t+1}, & [3T/4] < t \leq T.
    \end{array}\right.
    \nn
\label{dgp4.22}
\end{equation}

{\bf DGP 3 with four real breaks $S^0=\{[T/5],[2T/5],[3T/5],[4T/5]\}:$}
\begin{equation}
y_{t+1} = \left\{
        \begin{array}{lllll}
        (1, -0.5, 0.5, 1) \boldsymbol{z}_{pre,t} +
        (0.75, 0.75, \cdots, 0.75) \boldsymbol{w}_{pre,t} + u_{t+1},
        &0<t \leq [T/5];\\
        (0.5, 1, 1, 0.5) \boldsymbol{z}_{pre,t} +
        (1.5, 1.5, \cdots, 1.5) \boldsymbol{w}_{pre,t} + u_{t+1}, &[T/5] < t \leq [2T/5];\\
        (1, -0.5, 0.5, 1) \boldsymbol{z}_{pre,t} +
       (0.75, 0.75, \cdots, 0.75) \boldsymbol{w}_{pre,t} + u_{t+1}, &[2T/5] < t \leq [3T/5],\\
        (0.5, 1, 1, 0.5) \boldsymbol{z}_{pre,t} +
        (1.5, 1.5, \cdots, 1.5) \boldsymbol{w}_{pre,t} + u_{t+1}, & [3T/5] < t \leq [4T/5];\\
        (1, -0.5, 0.5, 1) \boldsymbol{z}_{pre,t} +
        (0.75, 0.75, \cdots, 0.75) \boldsymbol{w}_{pre,t} + u_{t+1}, &[4T/5] < t \leq T.
    \end{array}\right.\nn
\label{dgp4.23}
\end{equation}
Table \ref{tab123} reports the performance of our methods based on the specified parameter settings and defined evaluation criteria. Our SICS procedure effectively captures all significant predictors, achieving a coverage rate (CR) of 100\% across all cases. 
Furthermore, as demonstrated in Table \ref{tab123} and Figure \ref{fig:combined_breaks_vs}, our RCRS and breaks elimination method estimate both the number and locations of breaks well, resulting in high percentages of correct estimation (PCEs) and low Hausdorff distances (HDs), even in the presence of many redundant regressors. Additionally, Table \ref{tab123} and Figure \ref{fig:combined_breaks_vs} indicate that our elimination method successfully identifies effective predictors, resulting in low root mean squared errors (RMSEs). When cointegrated \(I(1)\) predictors are present, our approach not only selects them well but also estimates their coefficients accurately.
\begin{table}
    \centering
    \resizebox{0.7\textwidth}{!}{

\begin{tabular}{|cc|ccc|ccc|ccc|}
\hline
\multicolumn{2}{|c|}{\multirow{2}{*}{}}                                   & \multicolumn{3}{c|}{$T=400$}                                    & \multicolumn{3}{c|}{$T=500$}                                    & \multicolumn{3}{c|}{$T=600$}                                    \\ \cline{3-11} 
\multicolumn{2}{|c|}{}                                                    & \multicolumn{1}{c|}{DGP1}  & \multicolumn{1}{c|}{DGP2}  & DGP3  & \multicolumn{1}{c|}{DGP1}  & \multicolumn{1}{c|}{DGP2}  & DGP3  & \multicolumn{1}{c|}{DGP1}  & \multicolumn{1}{c|}{DGP2}  & DGP3  \\ \hline
\multicolumn{1}{|c|}{\multirow{3}{*}{CR}}                       & $j_p=2$ & \multicolumn{1}{c|}{100.0} & \multicolumn{1}{c|}{100.0} & 100.0 & \multicolumn{1}{c|}{100.0} & \multicolumn{1}{c|}{100.0} & 100.0 & \multicolumn{1}{c|}{100.0} & \multicolumn{1}{c|}{100.0} & 100.0 \\ \cline{2-11} 
\multicolumn{1}{|c|}{}                                          & $j_p=3$ & \multicolumn{1}{c|}{100.0} & \multicolumn{1}{c|}{100.0} & 100.0 & \multicolumn{1}{c|}{100.0} & \multicolumn{1}{c|}{100.0} & 100.0 & \multicolumn{1}{c|}{100.0} & \multicolumn{1}{c|}{100.0} & 100.0 \\ \cline{2-11} 
\multicolumn{1}{|c|}{}                                          & $j_p=4$ & \multicolumn{1}{c|}{100.0} & \multicolumn{1}{c|}{100.0} & 100.0 & \multicolumn{1}{c|}{100.0} & \multicolumn{1}{c|}{100.0} & 100.0 & \multicolumn{1}{c|}{100.0} & \multicolumn{1}{c|}{100.0} & 100.0 \\ \hline
\multicolumn{1}{|c|}{\multirow{3}{*}{PCE}}                      & $j_p=2$ & \multicolumn{1}{c|}{94.2}  & \multicolumn{1}{c|}{95.8}  & 82.4  & \multicolumn{1}{c|}{98.6}  & \multicolumn{1}{c|}{94.6}  & 94.8  & \multicolumn{1}{c|}{96.4}  & \multicolumn{1}{c|}{96.8}  & 96.0  \\ \cline{2-11} 
\multicolumn{1}{|c|}{}                                          & $j_p=3$ & \multicolumn{1}{c|}{95.0}  & \multicolumn{1}{c|}{94.0}  & 80.4  & \multicolumn{1}{c|}{94.4}  & \multicolumn{1}{c|}{96.2}  & 83.0  & \multicolumn{1}{c|}{94.6}  & \multicolumn{1}{c|}{97.2}  & 93.2  \\ \cline{2-11} 
\multicolumn{1}{|c|}{}                                          & $j_p=4$ & \multicolumn{1}{c|}{97.9}  & \multicolumn{1}{c|}{86.5}  & 70.8  & \multicolumn{1}{c|}{98.2}  & \multicolumn{1}{c|}{90.0}  & 87.6  & \multicolumn{1}{c|}{95.3}  & \multicolumn{1}{c|}{95.6}  & 94.1  \\ \hline
\multicolumn{1}{|c|}{\multirow{3}{*}{${100\text{HD} \over T}$}} & $j_p=2$ & \multicolumn{1}{c|}{2.7}   & \multicolumn{1}{c|}{3.1}   & 7.0   & \multicolumn{1}{c|}{2.2}   & \multicolumn{1}{c|}{2.7}   & 3.4   & \multicolumn{1}{c|}{2.0}   & \multicolumn{1}{c|}{2.2}   & 2.3   \\ \cline{2-11} 
\multicolumn{1}{|c|}{}                                          & $j_p=3$ & \multicolumn{1}{c|}{3.0}   & \multicolumn{1}{c|}{3.4}   & 7.5   & \multicolumn{1}{c|}{2.6}   & \multicolumn{1}{c|}{2.9}   & 3.8   & \multicolumn{1}{c|}{2.3}   & \multicolumn{1}{c|}{2.3}   & 2.6   \\ \cline{2-11} 
\multicolumn{1}{|c|}{}                                          & $j_p=4$ & \multicolumn{1}{c|}{3.2}   & \multicolumn{1}{c|}{3.9}   & 9.9   & \multicolumn{1}{c|}{2.6}   & \multicolumn{1}{c|}{3.3}   & 4.1   & \multicolumn{1}{c|}{2.0}   & \multicolumn{1}{c|}{2.6}   & 2.9   \\ \hline
\multicolumn{1}{|c|}{\multirow{3}{*}{RMSE}}                     & $j_p=2$ & \multicolumn{1}{c|}{0.3}   & \multicolumn{1}{c|}{0.3}   & 0.4   & \multicolumn{1}{c|}{0.2}   & \multicolumn{1}{c|}{0.3}   & 0.3   & \multicolumn{1}{c|}{0.2}   & \multicolumn{1}{c|}{0.3}   & 0.3   \\ \cline{2-11} 
\multicolumn{1}{|c|}{}                                          & $j_p=3$ & \multicolumn{1}{c|}{0.3}   & \multicolumn{1}{c|}{0.3}   & 0.4   & \multicolumn{1}{c|}{0.2}   & \multicolumn{1}{c|}{0.3}   & 0.4   & \multicolumn{1}{c|}{0.2}   & \multicolumn{1}{c|}{0.3}   & 0.3   \\ \cline{2-11} 
\multicolumn{1}{|c|}{}                                          & $j_p=4$ & \multicolumn{1}{c|}{0.3}   & \multicolumn{1}{c|}{0.4}   & 0.4   & \multicolumn{1}{c|}{0.3}   & \multicolumn{1}{c|}{0.3}   & 0.4   & \multicolumn{1}{c|}{0.2}   & \multicolumn{1}{c|}{0.3}   & 0.3   \\ \hline
\end{tabular}}
\caption{The CR, PCE, HD and RMSE for DGP1, DGP2, DGP3 with different $T, p$ and $S^0$.}
\label{tab123}
\end{table}

\begin{figure}
    \centering
    \begin{minipage}{0.48\textwidth}
        \centering
        \includegraphics[width=\textwidth]{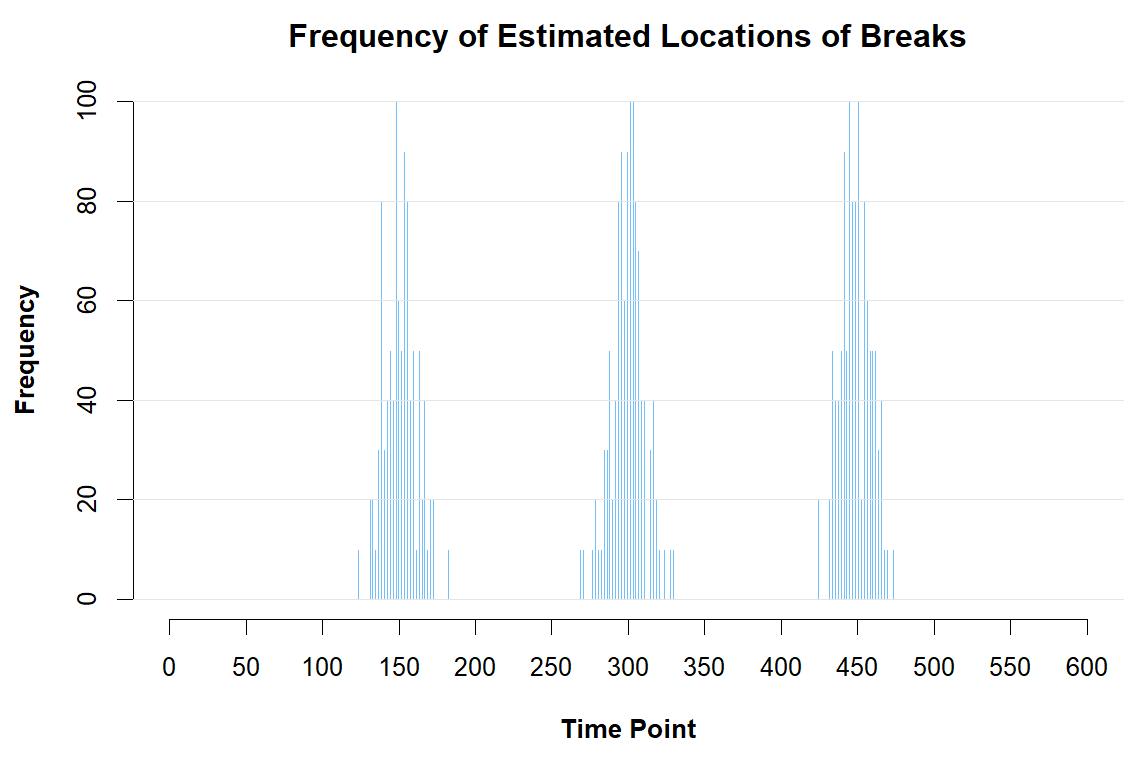}
    \end{minipage}
    \hfill 
    \begin{minipage}{0.48\textwidth}
        \centering
        \includegraphics[width=\textwidth]{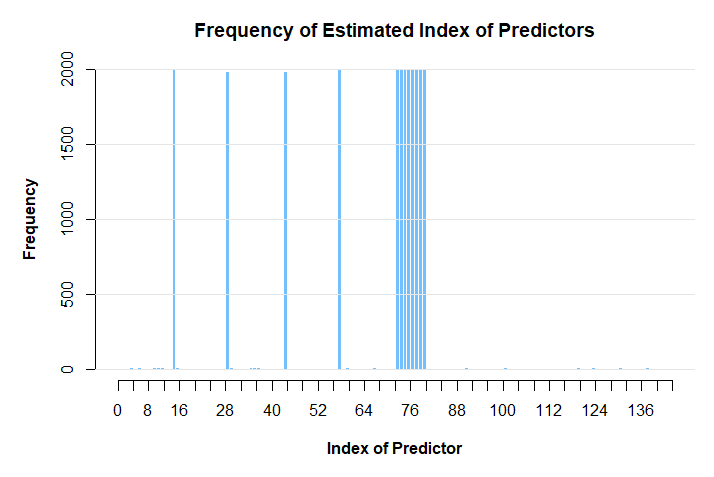}
    \end{minipage}
    
    \caption{Left: Locations Frequency of estimated change points. Right: Frequency of active predictor index in the first estimated regime. Both panels are under DGP 2, with $T = 600,~ S^0 = \{150,300,450\},~ p_z= p_w=72,~ r_F = 16$. The active predictors are set as $x_{15,t}, x_{29,t}, x_{44,t}, x_{58, t}~ (I(0)); ~x_{72+i,t},~i=1,2,\cdots,8~(I(1))$.}
    \label{fig:combined_breaks_vs}
\end{figure}
\vspace{-0.60cm}
\subsection{Comparison with Tu and Xie (2023)}\label{comparetx}
In this subsection, we consider a simulation example to compare the performance of our method with that of \cite{tuPenetratingSporadicReturn2023}. When the number of variables ($p$) is large, the method proposed by \cite{tuPenetratingSporadicReturn2023} may encounter issues with missed change point estimation, as discussed in Remark \ref{remarktx}.

We use DGP 1 and 2 described in Section \ref{sec:4.2}, with the same parameters except for the following: \(d_T = [4\log T]\), \(r_F = [\log T]\), and \(p_w = \max\{8, [j_p T^{0.45}]\}\), where \(j_p\) takes values of 0.3, 0.4, 0.5, 0.6, and 0.7. All tuning parameter choices for our method remain the same as specified in Section \ref{sec:4.1}. 


We follow all the suggested choices of the tuning parameters as in  \cite{tuPenetratingSporadicReturn2023}, except that we set $\varepsilon = \max\{0.05, 1.5pT^{-1}\}$ rather than $0.05$ in their paper. The reason is as follows: suppose that in the $k-$th step of their Group Forward Regression Screening (GFRS) procedure, one gets the locations set of breaks $\widehat{S}_k = \{\widehat{t}_1, \cdots, \widehat{t}_k\}$. The global Residual Sum of Squares with $\widehat{S}_k$ is given by
$
    {\rm RSS}(\widehat{S}_k) := \sum_{i=1}^{k+1} \sum_{t=\widehat{t}_{i-1}}^{\widehat{t}_i-1} (y_{t+1} - [\widehat{\boldsymbol{\gamma}}(\widehat{t}_{i-1}: \widehat{t}_{i})]'
    \boldsymbol{x}_t)^2,
$
where $\widehat{t}_0=1, \widehat{t}_{k+1}=T+1,$ and $
\widehat{\boldsymbol{\gamma}}(\widehat{t}_{i-1}: \widehat{t}_{i})$ denotes the least sqaure estimator of coefficient vector of $\boldsymbol{x}_t$ in the period $(\widehat{t}_{i-1}: \widehat{t}_{i}).$
In the $(k+1)-$th step, one aims to find a point $t$ maximizing
$
    {\rm RSS}(\widehat{S}_k) - {\rm RSS}(\widehat{S}_k\cup\{t\})
$
over $t$ satisfying $|t-\widehat{t}_l|\geq \varepsilon T,~l=1,\cdots,k.$ 
When $p$ is large, we can only select the new change points outside $\mathcal{U}(\widehat{S}_k, p)$ to ensure the feasibility of least square estimators, where $\mathcal{U}(S,r)$ is given in \eqref{Usr}.

\begin{table}[H]
    \centering\resizebox{0.65\textwidth}{!}{
\begin{tabular}{|cc|cc|cc|cc|cc|cc|}
\hline
\multicolumn{2}{|c|}{$j_p$}                              & \multicolumn{2}{c|}{0.3}           & \multicolumn{2}{c|}{0.4}           & \multicolumn{2}{c|}{0.5}          & \multicolumn{2}{c|}{0.6}         & \multicolumn{2}{c|}{0.7}         \\ \hline
\multicolumn{2}{|c|}{Method}                             & \multicolumn{1}{c|}{TX}    & Ours  & \multicolumn{1}{c|}{TX}    & Ours  & \multicolumn{1}{c|}{TX}   & Ours  & \multicolumn{1}{c|}{TX}  & Ours  & \multicolumn{1}{c|}{TX}  & Ours  \\ \hline
\multicolumn{1}{|c|}{\multirow{2}{*}{$T=400$}} & DGP1 & \multicolumn{1}{c|}{81.8}  & 100.0 & \multicolumn{1}{c|}{63.5}  & 100.0 & \multicolumn{1}{c|}{16.0} & 100.0 & \multicolumn{1}{c|}{1.0} & 100.0 & \multicolumn{1}{c|}{0.0} & 98.0 \\ \cline{2-12}
\multicolumn{1}{|c|}{}                         & DGP2 & \multicolumn{1}{c|}{0.0}   & 94.4  & \multicolumn{1}{c|}{0.0}   & 95.6  & \multicolumn{1}{c|}{0.0}  & 96.2  & \multicolumn{1}{c|}{0.0} & 95.3  & \multicolumn{1}{c|}{0.0} & 96.1  \\ \hline
\multicolumn{1}{|c|}{\multirow{2}{*}{$T=500$}} & DGP1 & \multicolumn{1}{c|}{100.0} & 100.0 & \multicolumn{1}{c|}{90.7}  & 100.0 & \multicolumn{1}{c|}{30.0} & 100.0 & \multicolumn{1}{c|}{1.0} & 100.0 & \multicolumn{1}{c|}{0.0} & 100.0 \\ \cline{2-12}
\multicolumn{1}{|c|}{}                         & DGP2 & \multicolumn{1}{c|}{6.6}   & 100.0 & \multicolumn{1}{c|}{0.0}   & 100.0 & \multicolumn{1}{c|}{0.0}  & 100.0 & \multicolumn{1}{c|}{0.0} & 100.0 & \multicolumn{1}{c|}{0.0} & 100.0 \\ \hline
\multicolumn{1}{|c|}{\multirow{2}{*}{$T=600$}} & DGP1 & \multicolumn{1}{c|}{100.0} & 100.0 & \multicolumn{1}{c|}{100.0} & 100.0 & \multicolumn{1}{c|}{43.8} & 100.0 & \multicolumn{1}{c|}{4.0} & 100.0 & \multicolumn{1}{c|}{0.0} & 100.0 \\ \cline{2-12}
\multicolumn{1}{|c|}{}                         & DGP2 & \multicolumn{1}{c|}{4.2}  & 100.0 & \multicolumn{1}{c|}{2.0}  & 100.0 & \multicolumn{1}{c|}{0.0}  & 100.0 & \multicolumn{1}{c|}{0.0} & 100.0 & \multicolumn{1}{c|}{0.0} & 100.0 \\ \hline
\end{tabular}}\caption{Percentages of Correct Estimation (PCE) of $m_0$ Comparison between the method in \cite{tuPenetratingSporadicReturn2023} and our paper using DGP1 and DGP2.}
    \label{fdcom}
\end{table}
We present the PCE 
in Table \ref{fdcom} and Figure \ref{figtx}. We refer to the method from \cite{tuPenetratingSporadicReturn2023} as "TX" and our own approach as "Ours". From Table \ref{fdcom} and Figure \ref{figtx}, it is evident that as the dimensionality of the predictors increases (i.e., when $j_p$ becomes larger), the effectiveness of the method in \cite{tuPenetratingSporadicReturn2023} significantly decreases. When $j_p$ exceeds 0.6 or $m_0=3$, indicating that $p$ is moderately large or the minimum distance between two adjacent change points is small, change points generally become unidentifiable by their method. In contrast, our method remains effective in these scenarios, consistently providing accurate estimates.
\begin{figure}
\centering
\begin{subfigure}[b]{0.495\textwidth}
    \centering
    \includegraphics[width=\linewidth]{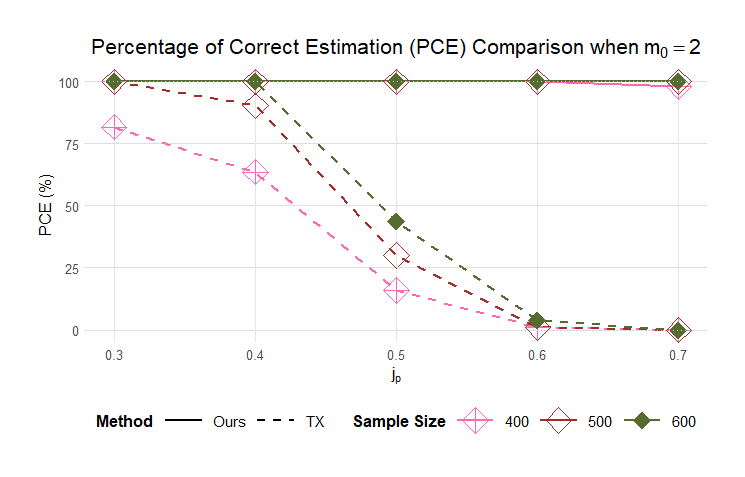}
    \caption{}
    \label{fig:sub1}
\end{subfigure}
\hfill
\begin{subfigure}[b]{0.495\textwidth}
    \centering
    \includegraphics[width=\linewidth]{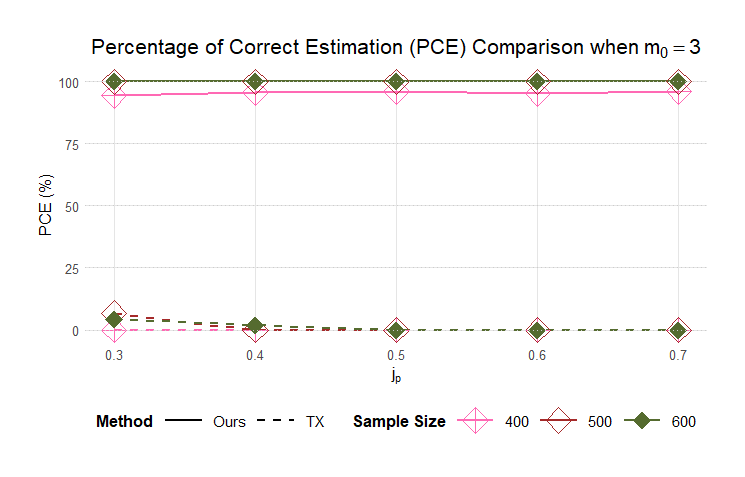}
    \caption{}
    \label{fig:sub2}
\end{subfigure}
\caption{Percentages of Correct Estimation (PCE) of $m_0$ Comparison between the method in \cite{tuPenetratingSporadicReturn2023} and our paper using DGP1 and DGP2.}
\label{figtx}
\end{figure}
\vspace{-0.60cm}
\subsection{Predictor Selection Comparison with \cite{mei2024lasso}}
\label{sec:4.4}

In this subsection, we perform a simulation to compare the selection performance with the LASSO method proposed by \cite{mei2024lasso}, in a 
predictive regression model with no structural breaks. As previously mentioned, our SICS procedure, when combined with the predictor elimination, serves as an effective strategy for predictor selection. In this example, we demonstrate that our method effectively identifies active predictors, especially cointegrated \(I(1)\) ones that cannot be selected well by LASSO method.

All tuning parameters are set the same as in Section \ref{sec:4.1} except $d_T=[4\log T]$. 
 Predictor vector $\boldsymbol{x}_t=(\boldsymbol{z}_t',\boldsymbol{w}_t')'$ is generated as described in Section \ref{sec:4.2}. We set the active predictors as $w_{i,t},~i=1,2,\cdots,8,$ 
 and set the true coefficient of them by $0.75b_0$ for all $i,$ 
where we take $b_0= 0.2,~ 0.3,~ 0.4$ to control the signal to noise ratio.

We present the RMSE defined in Section \ref{sec:4.1} for the Slasso described in \cite{mei2024lasso} as well as our own method. Given the estimated predictor set $\widehat{G}^{(r)}$ in the $r$-th replication and the true set $G^0,$ We also report the average selection coverage rate (SCR) given by ${\rm SCR} = {1\over R}\sum_{r=1}^R\frac{|\widehat{G}^{(r)} \cap G^0|}{|G^0|}$ in percentage.
The results are  summarized in Table \ref{psperform}. We will refer to the method from \cite{mei2024lasso} as "MS" and our method as "Ours". It can be noted that our method consistently demonstrates superior performance in selecting active predictors across all parameter settings.

\begin{table}[H]
\centering
    \resizebox{\textwidth}{!}{
\begin{tabular}{|c|c|cccccc|cccccc|cccccc|}
\hline
\multirow{3}{*}{$b_0$} & \multirow{3}{*}{Method} & \multicolumn{6}{c|}{$j_p=2$}                                                                                                                              & \multicolumn{6}{c|}{$j_p=3$}                                                                                                                              & \multicolumn{6}{c|}{$j_p=4$}                                                                                                                              \\ \cline{3-20} 
                       &                         & \multicolumn{2}{c|}{$T=400$}                             & \multicolumn{2}{c|}{$T=500$}                             & \multicolumn{2}{c|}{$T=600$}        & \multicolumn{2}{c|}{$T=400$}                             & \multicolumn{2}{c|}{$T=500$}                             & \multicolumn{2}{c|}{$T=600$}        & \multicolumn{2}{c|}{$T=400$}                             & \multicolumn{2}{c|}{$T=500$}                             & \multicolumn{2}{c|}{$T=600$}        \\ \cline{3-20} 
                       &                         & \multicolumn{1}{c|}{RMSE} & \multicolumn{1}{c|}{SCR} & \multicolumn{1}{c|}{RMSE} & \multicolumn{1}{c|}{SCR} & \multicolumn{1}{c|}{RMSE} & SCR & \multicolumn{1}{c|}{RMSE} & \multicolumn{1}{c|}{SCR} & \multicolumn{1}{c|}{RMSE} & \multicolumn{1}{c|}{SCR} & \multicolumn{1}{c|}{RMSE} & SCR & \multicolumn{1}{c|}{RMSE} & \multicolumn{1}{c|}{SCR} & \multicolumn{1}{c|}{RMSE} & \multicolumn{1}{c|}{SCR} & \multicolumn{1}{c|}{RMSE} & SCR \\ \hline
\multirow{2}{*}{0.2}   & Ours                    & \multicolumn{1}{c|}{0.06} & \multicolumn{1}{c|}{91.2}    & \multicolumn{1}{c|}{0.06} & \multicolumn{1}{c|}{97.2}    & \multicolumn{1}{c|}{0.06} & 97.6    & \multicolumn{1}{c|}{0.08} & \multicolumn{1}{c|}{69.4}    & \multicolumn{1}{c|}{0.05} & \multicolumn{1}{c|}{93.2}    & \multicolumn{1}{c|}{0.06} & 100.0   & \multicolumn{1}{c|}{0.05} & \multicolumn{1}{c|}{65.0}    & \multicolumn{1}{c|}{0.06} & \multicolumn{1}{c|}{88.2}    & \multicolumn{1}{c|}{0.05} & 95.0    \\ \cline{2-20} 
                       & MS                      & \multicolumn{1}{c|}{0.42} & \multicolumn{1}{c|}{13.8}    & \multicolumn{1}{c|}{0.42} & \multicolumn{1}{c|}{12.6}    & \multicolumn{1}{c|}{0.40} & 32.0    & \multicolumn{1}{c|}{0.43} & \multicolumn{1}{c|}{3.2}     & \multicolumn{1}{c|}{0.42} & \multicolumn{1}{c|}{10.6}    & \multicolumn{1}{c|}{0.42} & 12.0    & \multicolumn{1}{c|}{0.43} & \multicolumn{1}{c|}{20.0}    & \multicolumn{1}{c|}{0.43} & \multicolumn{1}{c|}{21.8}    & \multicolumn{1}{c|}{0.43} & 23.8    \\ \hline
\multirow{2}{*}{0.3}   & Ours                    & \multicolumn{1}{c|}{0.08} & \multicolumn{1}{c|}{100.0}   & \multicolumn{1}{c|}{0.06} & \multicolumn{1}{c|}{100.0}   & \multicolumn{1}{c|}{0.07} & 100.0   & \multicolumn{1}{c|}{0.07} & \multicolumn{1}{c|}{100.0}   & \multicolumn{1}{c|}{0.06} & \multicolumn{1}{c|}{100.0}   & \multicolumn{1}{c|}{0.06} & 100.0   & \multicolumn{1}{c|}{0.08} & \multicolumn{1}{c|}{100.0}   & \multicolumn{1}{c|}{0.06} & \multicolumn{1}{c|}{100.0}   & \multicolumn{1}{c|}{0.06} & 100.0   \\ \cline{2-20} 
                       & MS                      & \multicolumn{1}{c|}{0.47} & \multicolumn{1}{c|}{70.2}    & \multicolumn{1}{c|}{0.38} & \multicolumn{1}{c|}{94.2}    & \multicolumn{1}{c|}{0.39} & 95.0    & \multicolumn{1}{c|}{0.56} & \multicolumn{1}{c|}{36.6}    & \multicolumn{1}{c|}{0.57} & \multicolumn{1}{c|}{30.0}    & \multicolumn{1}{c|}{0.45} & 71.8    & \multicolumn{1}{c|}{0.61} & \multicolumn{1}{c|}{28.8}    & \multicolumn{1}{c|}{0.58} & \multicolumn{1}{c|}{33.2}    & \multicolumn{1}{c|}{0.55} & 46.2    \\ \hline
\multirow{2}{*}{0.4}   & Ours                    & \multicolumn{1}{c|}{0.08} & \multicolumn{1}{c|}{100.0}   & \multicolumn{1}{c|}{0.05} & \multicolumn{1}{c|}{100.0}   & \multicolumn{1}{c|}{0.06} & 100.0   & \multicolumn{1}{c|}{0.07} & \multicolumn{1}{c|}{100.0}   & \multicolumn{1}{c|}{0.05} & \multicolumn{1}{c|}{100.0}   & \multicolumn{1}{c|}{0.06} & 100.0   & \multicolumn{1}{c|}{0.06} & \multicolumn{1}{c|}{100.0}   & \multicolumn{1}{c|}{0.06} & \multicolumn{1}{c|}{100.0}   & \multicolumn{1}{c|}{0.06} & 100.0   \\ \cline{2-20} 
                       & MS                      & \multicolumn{1}{c|}{0.39} & \multicolumn{1}{c|}{95.2}    & \multicolumn{1}{c|}{0.38} & \multicolumn{1}{c|}{100.0}   & \multicolumn{1}{c|}{0.35} & 100.0   & \multicolumn{1}{c|}{0.44} & \multicolumn{1}{c|}{100.0}   & \multicolumn{1}{c|}{0.42} & \multicolumn{1}{c|}{100.0}   & \multicolumn{1}{c|}{0.39} & 100.0   & \multicolumn{1}{c|}{0.62} & \multicolumn{1}{c|}{70.4}    & \multicolumn{1}{c|}{0.55} & \multicolumn{1}{c|}{85.6}    & \multicolumn{1}{c|}{0.48} & 94.4    \\ \hline
\end{tabular}}
    \caption{RMSE and SCR Comparison with \cite{mei2024lasso}.
    }
    \label{psperform}
\end{table}

\vspace{-0.60cm}
\section{EMPIRICAL ANALYSIS}
\label{sec:emp}

We apply the procedure to 
U.S.\ CPI inflation, 
which is a 
primary target of the large-panel macro forecasting literature
as in \cite{stock2002forecasting}. We work with the
FRED-MD database (\cite{mccrackenFREDMDMonthlyDatabase2016}),
which compiles 126 monthly U.S.\ macroeconomic series over
1960:01--2024:12. The
predictor panel has the mixed-integration structure the
methodology is designed to handle: it contains both stationary
$I(0)$ series (detrended real quantities, spreads, survey
measures) and cointegrated $I(1)$ series (price levels,
monetary aggregates, nominal interest rates), and both targets
exhibit multiple documented historical regime changes over the
sample.

Each series in the FRED-MD database comes with a
series-specific \texttt{tcode} column prescribing
transformations to render every series approximately stationary.
This default stationary transformation is unnecessary for our  procedure, which permits both cointegrated and redundant persistent variables. We therefore use a modified rule
\texttt{tcode\_mod} that applies the minimum transformation
needed to bring each series into either $I(0)$ or $I(1)$. The detailed constructions are reported in Supplementary
Material Section~\ref{sec:supp_tcode}.

We have two purposes in this empirical exercise. First, on the
discovery side, we examine the breaks detected and the
predictors selected in each regime, and assess whether these
findings are consistent with the documented macroeconomic
history. 
Second, on the predictive
side, we consider out-of-sample forecasting comparisons with \cite{mei2024lasso}, \cite{tuPenetratingSporadicReturn2023}.

\subsection{Structural Break Detection and Predictor Selection}\label{sec:emp_discovery}

Figure~\ref{fig:discovery_cpi_pipeline} displays the full
procedure diagnostic. Panel~(a) plots the sample canonical correlations sorted in ascending order;
the retained predictors cluster at the left tail of the
distribution, well separated from the bulk. Panel~(b) traces the
local-window RSS over the sample as RCRS adds break candidates
sequentially: the first four forward picks (red) generate the
largest RSS reductions and fall in or near the orange-shaded
windows of documented historical events, while picks 5 and 6
(dark orange) fall within the slack zone and are carried
forward for IC evaluation. Panel~(c) shows the RSSR$_k$
trajectory.  
 Panel~(d) traces the IC$_1$
backward-elimination path, which minimizes at four
breaks, confirming the final count.

\begin{figure}[!ht]
  \centering
  \includegraphics[width=\textwidth, trim={0 0 0 30}, clip]{
  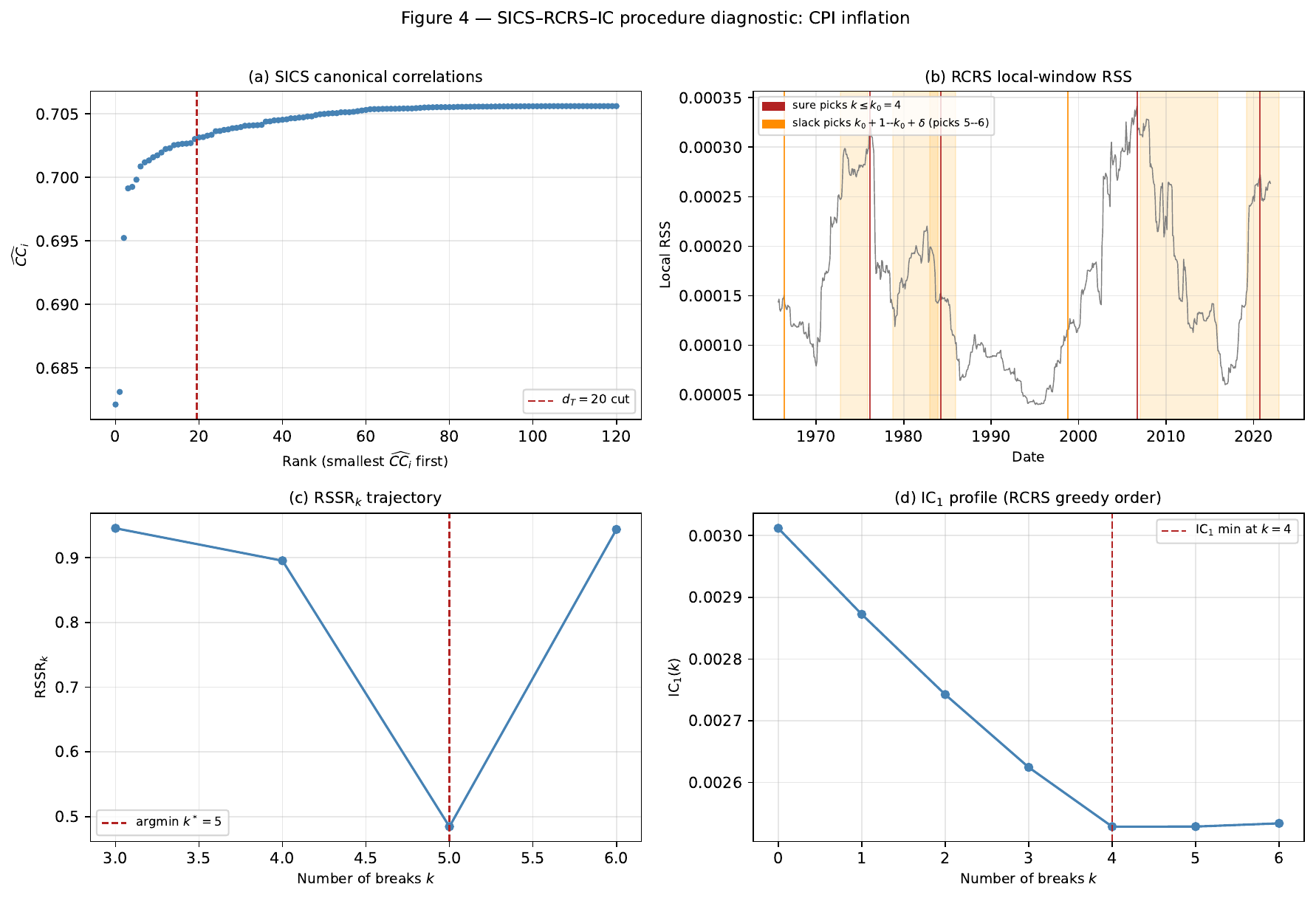}
  \caption{SICS--RCRS--IC procedure diagnostic for CPI inflation.}
  \label{fig:discovery_cpi_pipeline}

  \parbox{\linewidth}{\scriptsize\linespread{0.85}\selectfont\textit{Note:} \textbf{(a)}~SICS sample canonical correlations, sorted, with the $d_T = 20$ cut. \textbf{(b)}~Local-window RSS trajectory over time. 
  \textbf{(c)}~RSSR$_k$ trajectory with argmin at $k^* = 5$.
  \textbf{(d)}~IC$_1$ trajectory across backward-elimination steps, minimized at $|\widehat{\widehat S}| = 4$.}
\end{figure}

The procedure detects four breaks in headline CPI inflation: 
1976-03, 1984-04, 2006-09, and
2020-10. 
They align with widely documented inflation
regimes. The 1976-03 detection is consistent with the end of the
1973--1974 oil-shock inflation surge. 
The 1984-04 detection coincides with the end of the
Volcker disinflation (\cite{goodfriend2005incredible}). 
The 2006-09 break is economically interpretable
as the pre-financial-crisis credit
expansion. The 2020-10 break captures the
onset of the COVID inflation surge.

The per-regime IC$_2$-selected active sets, shown in Table \ref{tab:regime_summary}, are parsimonious
and economically interpretable. 
For instance, the
1984--2006 post-Volcker regime retains one active predictor, oil price, which indicates that 
supply-side commodity shocks are the dominant source of
short-run inflation variation once expectations are anchored.
The
2006--2020 regime expands the active set to
three monetary predictors, consistent with
the credit and risk-spread channels through which the 2007--2009
financial disruption transmitted to prices. 
After the 2020
COVID break, three predictors -- retail sales, industrial production, and capacity
utilization -- reflect the unprecedented monetary expansion of
2020--2021 and its transmission to consumer prices through both
demand and supply channels.

\begin{table}[!ht]
\centering
\scriptsize\setlength{\tabcolsep}{3pt}
\renewcommand{\arraystretch}{0.95}
\caption{Detected regimes and IC$_2$-selected active predictor sets.}
\label{tab:regime_summary}
\begin{tabularx}{\linewidth}{@{}lcX@{}}
\toprule
Regime & Number of Predictors & Active predictors \\
\midrule
1962:08--1976:02 & 4 & cap.\ util., inv./sales, S\&P P/E, manuf.\ empl. 
\\
1976:03--1984:03 & 5 & manuf. empl., dur.-goods empl., inv./sales, oil, cpi. dur. \\
1984:04--2006:08 & 1 & oil \\
2006:09--2020:09 & 3 & mon.\ base, oil, reserves\\
2020:10--2024:12 & 3 & retail sales, IP, cap.\ util.\\
\bottomrule
\end{tabularx}

\vspace{4pt}
\parbox{\linewidth}{\scriptsize\linespread{0.85}\selectfont\textit{Note:} 
Abbreviations: oil $=$ \texttt{OILPRICEx}; cap.\ util.\ $=$ \texttt{CUMFNS}; 
mon.\ base $=$ \texttt{BOGMBASE}; 
manuf. empl $=$ \texttt{MANEMP}; inv./sales $=$ \texttt{ISRATIOx}; S\&P P/E $=$ \texttt{S\&P PE ratio}; dur.-goods empl.$=$ \texttt{DMANEMP}; cpi. dur. $=$ \texttt{CUSR0000SAD}; reserves $=$ \texttt{TOTRESNS}; retail sales $=$\texttt{RETAILx}; IP $=$ \texttt{INDPRO}. For details on the predictors, refer to the FRED-MD database documentation (\cite{mccrackenFREDMDMonthlyDatabase2016}).}
\end{table}

\subsection{Out-of-sample Forecast Comparison}\label{sec:emp_oos}

To quantify the practical payoff of the detected breaks and predictors for
forecasting, we consider forecasting exercises comparison with \cite{mei2024lasso} and \cite{tuPenetratingSporadicReturn2023}. 
The comparison is on monthly forecast origins from January 1990
through December 2019, for forecast horizons
$h \in \{1, 2, 3\}$ months.  
At each forecast origin $s$, we run the procedures in our paper  and \cite{tuPenetratingSporadicReturn2023} using data available up to $s$, identify the most recent detected break that is at
least $40$ months before $s$, and forecast
$y_{s+h+1}$ from a centered OLS on the last-regime window using
the regime's active predictor set, respectively. We also run the SLasso procedure in \cite{mei2024lasso} using data available up to $s$, and forecast $y_{s+h+1}$ from a centered OLS without considering structural breaks.

We report the root-mean-squared forecast error (RMSFE) of these methods during four forecasting windows: the 1990-1999 period (1990:01--1999:12), the 2000-2009 period (2000:01--2009:12), the 2010-2023 period (2010:01--2023:12), and the full period (1990:01--2023:12). 
The RMSFE is defined as:
\begin{equation}
	\text{RMSFE} = \left(\frac{1}{n} \sum_{s= s_0}^{s_0 + n-1} (\widehat{y}_{s+h+1}- y_{s+h+1})^2\right)^{{1\over 2}},
\end{equation}
where $s_0,~n$ denote the start point and number of origins in a forecasting window. 
The results in Table \ref{tab:oos_ablation_transposed} demonstrate that our approach achieves the lowest RMSFE across all horizons in the full-sample, post-2000 and after-2010 periods, highlighting the predictive advantage of our proposed method in predictor selection as well as change point detection. 
On one hand, our dimension reduction approach extracts more precise structural changes and predictive information compared to \cite{tuPenetratingSporadicReturn2023}. On the other hand, when structural breaks occur in a high-dimensional setting, our method outperforms \cite{mei2024lasso} by identifying regime changes to deliver more accurate forecasts.

\begin{table}[H]
\centering
\scriptsize
\renewcommand{\arraystretch}{1.1}
\setlength{\tabcolsep}{3.5pt}
\setlength{\aboverulesep}{0pt}
\setlength{\belowrulesep}{0pt}
\caption{RMSFE Comparisons.}
\label{tab:oos_ablation_transposed}

\vspace{0.6ex}
\resizebox{\textwidth}{!}{ 
\begin{tabular}{@{}l *{16}{c} @{}}
\toprule
& \multicolumn{3}{c}{\textit{1990--1999 Period
}} 
& \multicolumn{3}{c}{\textit{2000--2009 Period 
}} 
& \multicolumn{3}{c}{\textit{2010--2023 Period
}} 
& \multicolumn{3}{c}{\textit{Full Period
}} \\
& \multicolumn{3}{c}{$n=120$} 
& \multicolumn{3}{c}{$n=120$} 
& \multicolumn{3}{c}{$n=167$} 
& \multicolumn{3}{c}{$n=407$} \\
\cmidrule(lr){2-4} \cmidrule(lr){5-7} \cmidrule(lr){8-10}\cmidrule(lr){11-13} 
$h$ & 1 & 2 & 3  &1 & 2 & 3  & 1 & 2 & 3 & 1 & 2 & 3 \\
\midrule
Ours 
& \textbf{0.161} & \textbf{0.154} & \textbf{0.160}  
& 0.394 & \textbf{0.374} & 0.401 
& \textbf{0.253} & \textbf{0.268} & \textbf{0.274} 
& \textbf{0.283} & \textbf{0.279} & \textbf{0.293}
   \\
\addlinespace[2pt]

MS
& 0.189 & 0.189 & 0.190 
& 0.453 & 0.454 & 0.450
& 0.274 & 0.274 & 0.275 
& 0.322 & 0.320 & 0.318  \\
\addlinespace[2pt]

TX
& 0.222 & 0.222 & 0.224 
& \textbf{0.384} & 0.390 & \textbf{0.387}
& 0.313 & 0.310 & 0.321 
& 0.313 & 0.314 & 0.319  \\
\addlinespace[2pt]

\bottomrule
\end{tabular}}

\vspace{4pt}
\parbox{\linewidth}{\scriptsize\linespread{0.95}
\selectfont\textit{Note:} 
We report $10^2\times$ RMSFE in the table. We use bold font to indicate the smallest RMSFE computed by all three methods, and refer to the results from the methods by \cite{mei2024lasso}, \cite{tuPenetratingSporadicReturn2023} and our paper as "MS",  "TX" and "Ours", respectively. 
}
\end{table}

\vspace{-0.60cm}
\section{CONCLUSION}

In this paper, we examine a predictive regression model characterized by high dimensionality, cointegrated \(I(1)\) regressors, and structural breaks. We begin by applying a Sure Independence Canonical Screening (SICS) procedure to identify active predictors. Next, we utilize a Ratio-Controlled Regression Screening (RCRS) procedure to select change points. We further refine the detection of change points and active predictors through two elimination steps based on information criteria (IC).
	
Our results demonstrate that the proposed methods effectively identify active predictors, particularly those that are cointegrated, within a predictive regression model that experiences structural breaks. Additionally, our approach consistently estimates the number and locations of change points, even in the presence of numerous redundant regressors. Simulation results and an empirical study 
indicate that our method delivers exceptional performance.

\bibliographystyle{apalike}
\bibliography{Bibliography-MM-MC}

\newpage
\appendix
\begin{center}
{\Large\bf Supplementary Material for ``Feature Screening for High-Dimensional Structural Break Predictive Regression''}
\end{center}

\section{Proofs of main results}

\subsection{The Proof of Theorem 3.1}

\begin{lemma}
    Suppose Assumption 1 (1) holds. Let $\boldsymbol{W}_t = (W^1_t,\cdots,W^{r_F}_t)', W^i_t = \sigma_{ii} \bigg(B_t^i-\int_0^1 B_s^i {\rm d}s\bigg),$ where $\sigma_{ii}$'s are given in Assumption 1 (1), and $\{B_t^i\}_{i=1}^{r_F}$ are independent Brownian motions. Suppose $r_F = o((l-s)^{{1\over 2}-\tau}),$ then we have
    \begin{equation}
        \bigg\| {1\over l-s} \boldsymbol{\widehat{\Omega}}_F^{(s:l)}
        -\int_0^1 \boldsymbol{W}_t \boldsymbol{W}_t' {\rm d}t\bigg\| = O_p(r_F (l-s)^{\tau-{1\over 2}}).\label{S13}
    \end{equation}\label{lemma_strapp}
\end{lemma}

\begin{proof}
    Let $\boldsymbol{\eta}_t = (\eta_t^1, \cdots, \eta_t^{r_F})'$ be an integrated process satisfying $\boldsymbol{\eta}_0=\boldsymbol{0}, \eta_t^i - \eta_{t-1}^i=\sigma_{ii} \xi_t^i,$ where $\xi_t^i$'s are given in Assumption 1 (1). We have
    \begin{equation}
        \begin{aligned}
            &\quad~{1\over (l-s)^2}\sum_{t=s}^{l-1} [(F_t^i-\bar{F}_i^{(s:l)})(F_t^j-\bar{F}_j^{(s:l)}) - (\eta_t^i-\bar{\eta}_i^{(s:l)})(\eta_t^j-\bar{\eta}_j^{(s:l)})]\\
        &= {1\over (l-s)^2}\sum_{t=s}^{l-1} [(F_t^i-\eta_t^i)-(\bar{F}_i^{(s:l)}-\bar{\eta}_i^{(s:l)})](F_t^j-\bar{F}_j^{(s:l)}) \\ &\quad~+{1\over (l-s)^2}\sum_{t=s}^{l-1}(\eta_t^i-\bar{\eta}_i^{(s:l)}) [(F_t^j-\eta_t^j)-(\bar{F}_j^{(s:l)}-\bar{\eta}_j^{(s:l)})]\\
        &:= r_{ij}^1 + r_{ij}^2.
        \end{aligned}
    \end{equation}
    By induction, we can show that under Assumption 1 (1), 
    \begin{equation}
        \sup\limits_{1\leq i \leq r_F}\sup\limits_{1\leq s<l \leq T} {\rm E}\{[(F_t^i-\eta_t^i)-(\bar{F}_i^{(s:l)}-\bar{\eta}_i^{(s:l)})]/(l-s)^{1/2}\}^2 = O((l-s)^{2\tau-1}),\label{S15}
    \end{equation}
    \begin{equation}
        \sup\limits_{1\leq i \leq r_F}\sup\limits_{1\leq s<l \leq T} {\rm E}\bigg[{F_t^i-\bar{F}_i^{(s:l)}\over (l-s)^{1/2}}\bigg]^2 = O(1), ~and~ \sup\limits_{1\leq i \leq r_F}\sup\limits_{1\leq s<l \leq T} {\rm E}\bigg[{\eta_t^i-\bar{\eta}_i^{(s:l)}\over(l-s)^{1/2}}\bigg]^2 = O(1).\label{S16}
    \end{equation}
    Thus, by equations \eqref{S15}, \eqref{S16} and the independence of the components, we have
    \begin{equation}
        \sum_{i,j=1}^{r_F} [{\rm E}(r_{ij}^1)^2 + {\rm E}(r_{ij}^2)^2] = O(r_F^2 (l-s)^{2\tau-1}),
    \end{equation}
    which implies that
    \begin{equation}
        \bigg\|{1\over (l-s)^2} \sum_{t=s}^{l-1} (\boldsymbol{F}_t - \bar{\boldsymbol{F}}^{(s:l)})(\boldsymbol{F}_t - \bar{\boldsymbol{F}}^{(s:l)})'-{1\over (l-s)^2} \sum_{t=s}^{l-1}(\boldsymbol{\eta}_t - \bar{\boldsymbol{\eta}}^{(s:l)})(\boldsymbol{\eta}_t - \bar{\boldsymbol{\eta}}^{(s:l)})'\bigg\| = O_p(r_F (l-s)^{\tau-{1\over 2}}).
    \end{equation}
    Thus, it suffices to show that
    \begin{equation}
        \sup\limits_{1\leq i,j\leq r_F} \bigg\|{1\over (l-s)^2} \sum_{t=s}^{l-1} (\eta_t^i-\bar{\eta}_i^{(s:l)}) (\eta_t^j-\bar{\eta}_j^{(s:l)}) - \int_0^1 W^i_t W^j_t {\rm d}t\bigg\| = o_{a.s.}((l-s)^{\tau-{1\over 2}}).\label{S17}
    \end{equation}
    Note that
    \begin{equation}
        \begin{aligned}
            &\quad~{1\over l-s}\sum_{t=s}^{l-1} \bigg({\eta_t^i-\bar{\eta}_i^{(s:l)}\over (l-s)^{1\over 2}}\bigg)\bigg({\eta_t^j-\bar{\eta}_j^{(s:l)}\over (l-s)^{1\over 2}}\bigg) - \sum_{t=s}^{l-1} \int_{(t-s)/(l-s)}^{{(t-s+1)/(l-s)}} W^i(a)W^j(a){\rm d}a\\
            &= {1\over l-s}\sum_{t=s}^{l-1} \bigg({\eta_t^i-\bar{\eta}_i^{(s:l)}\over (l-s)^{1\over 2}} -W^i\bigg({t-s+1 \over l-s}\bigg) \bigg)\bigg({\eta_t^j-\bar{\eta}_j^{(s:l)}\over (l-s)^{1\over 2}}\bigg) + {1\over l-s}\sum_{t=s}^{l-1}W^i\bigg({t-s+1 \over l-s}\bigg) \bigg({\eta_t^j-\bar{\eta}_j^{(s:l)}\over (l-s)^{1\over 2}}\\ &\quad~-W^j\bigg({t-s+1 \over l-s}\bigg)\bigg)
            - \sum_{t=s}^{l-1} \int_{(t-s)/(l-s)}^{{(t-s+1)/(l-s)}} \bigg\{\bigg[W^i(a) - W^i\bigg({t-s+1 \over l-s}\bigg)\bigg]W^j(a) \\ &\quad~+ W^i\bigg({t-s+1 \over l-s}\bigg) \bigg[W^j(a)- W^j\bigg({t-s+1 \over l-s}\bigg)\bigg]\bigg\}
            {\rm d}a
            \\
            &:= J_{1}(i,j) + J_2(i,j) + J_3(i,j).
        \end{aligned}
    \end{equation}
    For $t=s,\cdots,l-1,$ let $t'=t-s+1, \bar{\eta}_i'= {1\over l-s}\sum_{t'=1}^{l-s} \eta_{t'},$ we have
    \begin{equation}
        {\eta_t^i-\bar{\eta}_i^{(s:l)}\over (l-s)^{1\over 2}}-W^i\bigg({t-s+1 \over l-s}\bigg) \stackrel{d}{=} {\eta_{t'}^i-\bar{\eta}_i'\over (l-s)^{1\over 2}}-W^i\bigg({t'\over l-s}\bigg),\label{S21}
    \end{equation}
By similar argument to the proof of Lemma 3 in \cite{zhangIdentifyingCointegrationEigenanalysis2019}, we have
   \beqn
         &&\bigg|{\eta_{t'}^i-\bar{\eta}_i'\over (l-s)^{1\over 2}}-W^i\bigg({t'\over l-s}\bigg)\bigg|\nn\\
          &=& \sigma_{ii}\bigg|\bigg[{1\over (l-s)^{1\over 2}}\sum_{s'=1}^{t'}\xi_{s'}^i - B^i\bigg({t'\over l-s}\bigg) \bigg] + \bigg[{1\over (l-s)^{3\over 2}} \sum_{t'=1}^{l-s} \sum_{s'=1}^{t'} \xi_{s'}^i - \int_0^1 B_t^i {\rm d}t\bigg]
         \bigg|\nn\\
         &=& O_{a.s.}((l-s)^{-{1\over 2}} \log^2(l-s)).
   \label{ilaw}
    \eeqn
    Combining \eqref{S21}, \eqref{ilaw} yields \begin{equation}
        \bigg| {\eta_t^i-\bar{\eta}_i^{(s:l)}\over (l-s)^{1\over 2}}-W^i\bigg({t-s+1 \over l-s}\bigg)\bigg| = O_{a.s.}((l-s)^{-{1\over 2}} \log^2(l-s)).
    \end{equation}
    Then we have
    \begin{equation}
        \sup\limits_{1\leq i,j\leq r_F}|J_1(i,j)| = O_{a.s.}((l-s)^{-{1\over 2}} \log^2(l-s)).\label{S25}
    \end{equation}
    Similarly, we have
    \begin{equation}
        \sup\limits_{1\leq i,j\leq r_F}|J_2(i,j)| = O_{a.s.}((l-s)^{-{1\over 2}} \log^2(l-s)).\label{S26}
    \end{equation}
    By continuity of Brownian motion, we can get that
    \begin{equation}
        \sup\limits_{1\leq i,j\leq r_F} |J_3(i,j)| = O_{a.s.}((l-s)^{-{1\over 2}}\log^2 (l-s)).\label{S20}
    \end{equation}
    Combining \eqref{S25}, \eqref{S26} and \eqref{S20}, we have \eqref{S17} and conclude \eqref{S13}.
\end{proof}
In the following Lemmas \ref{lemma_cov} and \ref{lemma_ux}, we denote
\begin{equation}
    \boldsymbol{P} = {\rm diag}\{\boldsymbol{I}_{p-r_F}, (l-s)^{-{1\over 2}} \boldsymbol{I}_{r_F}\}.\label{Pls}
\end{equation}
\begin{lemma}
    Suppose Assumption 1 holds and $r_F(l-s)^{2\tau-1}\rightarrow 0, p(l-s)^{-{1\over 2}}\rightarrow 0$ as $T\rightarrow \infty,$ where $p$ is the dimension of $\boldsymbol{x}_t$ in Model \eqref{model}, 
   then
\begin{equation}\begin{aligned}
&~\quad\bigg\|\boldsymbol{P} \boldsymbol{A}
        \widehat{\boldsymbol{\Omega}}_x^{(s:l)} \boldsymbol{A}'
        \boldsymbol{P}-{\rm diag}\bigg\{{\rm Cov}(\boldsymbol{z}_t),{\rm Cov}(\boldsymbol{\psi}_{1t}), \int_0^1 \boldsymbol{W}_t \boldsymbol{W}_t'{\rm d}t\bigg\}\bigg\|\\ &=O_p(r_F(l-s)^{\tau-{1\over 2}}+p(l-s)^{-{1\over 2}}) = o_p(1),\end{aligned}
    \end{equation}
    where $\boldsymbol{A}, \boldsymbol{\psi}_{1t}$ are given in Section 3.1.
    \label{lemma_cov}
\end{lemma}

\begin{proof}
    By \eqref{cointA}, \eqref{Pls}, we have
    \begin{equation}
        \boldsymbol{P} \boldsymbol{A}
        \widehat{\boldsymbol{\Omega}}_x^{(s:l)} \boldsymbol{A}'
        \boldsymbol{P} = \begin{bmatrix}
            \widehat{\boldsymbol{\Omega}}_z^{(s:l)} & \widehat{\boldsymbol{\Omega}}_{z,\psi_1}^{(s:l)} & (l-s)^{-{1\over 2}}\widehat{\boldsymbol{\Omega}}_{z,\psi_2}^{(s:l)}\\

            [\widehat{\boldsymbol{\Omega}}_{z,\psi_1}^{(s:l)}]' & \widehat{\boldsymbol{\Omega}}_{\psi_1}^{(s:l)} &(l-s)^{-{1\over 2}}\widehat{\boldsymbol{\Omega}}_{\psi_1,\psi_2}^{(s:l)} \\

           [(l-s)^{-{1\over 2}}\widehat{\boldsymbol{\Omega}}_{z,\psi_2}^{(s:l)}]'  & [(l-s)^{-{1\over 2}}\widehat{\boldsymbol{\Omega}}_{\psi_1,\psi_2}^{(s:l)}]' & (l-s)^{-1}\widehat{\boldsymbol{\Omega}}_{\psi_2}^{(s:l)}
        \end{bmatrix},
    \end{equation}
    We first focus on $(l-s)^{-1}\widehat{\boldsymbol{\Omega}}_{\psi_2}^{(s:l)}.$ Note that
    \begin{equation}
        \begin{aligned}
            \widehat{\boldsymbol{\Omega}}_{\psi_2}^{(s:l)} &= \dfrac{1}{l-s}\sum_{t=s}^{l-1} [(\boldsymbol{F}_t-\bar{\boldsymbol{F}}^{(s:l)}) + \boldsymbol{Q}' (\boldsymbol{e}_t-\bar{\boldsymbol{e}}^{(s:l)})] [(\boldsymbol{F}_t-\bar{\boldsymbol{F}}^{(s:l)}) + \boldsymbol{Q}' (\boldsymbol{e}_t-\bar{\boldsymbol{e}}^{(s:l)})]'\\
            &= \widehat{\boldsymbol{\Omega}}_F^{(s:l)} +\widehat{\boldsymbol{\Omega}}_{Fe}^{(s:l)} \boldsymbol{Q}
            +  \boldsymbol{Q}'[\widehat{\boldsymbol{\Omega}}_{Fe}^{(s:l)}]'  + \boldsymbol{Q}' \widehat{\boldsymbol{\Omega}}_{e}^{(s:l)} \boldsymbol{Q}.
        \end{aligned}\label{wtB}
    \end{equation}
    By Lemma \ref{lemma_strapp}, we have
    \begin{equation}
        \begin{aligned}
        \bigg\|{\widehat{\boldsymbol{\Omega}}_{F}^{(s:l)} \over l-s}- \int_0^1 \boldsymbol{W}(s) \boldsymbol{W}(s)
        '{\rm d}s \bigg\|
        = O_p(r_F (l-s)^{\tau-1/2})
        \stackrel{P}{\longrightarrow} 0.
        \label{FF}
        \end{aligned}
    \end{equation}
    By Assumption 1 (2), let $\boldsymbol{\Omega}_e = {\rm Cov}(\boldsymbol{e}_t),$ we have
    \begin{equation}
        \begin{aligned}
            \|\widehat{\boldsymbol{\Omega}}_e^{(s:l)} - \boldsymbol{\Omega}_e\| &= \bigg\|{1\over l-s}\sum_{t=s}^{l-1} \boldsymbol{e}_t \boldsymbol{e}_t' - (\bar{\boldsymbol{e}}^{(s:l)})(\bar{\boldsymbol{e}}^{(s:l)})' - {\rm E} \boldsymbol{e}_t \boldsymbol{e}_t'\bigg\|.
        \end{aligned}\label{eet}
    \end{equation}
    Note that
    \begin{equation}
        \begin{aligned}
           {\rm E} \bigg\|{1\over l-s} \sum_{t=s}^{l-1} \boldsymbol{e}_t \boldsymbol{e}_t' - {\rm E}\boldsymbol{e}_t \boldsymbol{e}_t'\bigg\|_F^2 &= {1\over (l-s)^2} \sum_{i=1}^{p_w}\sum_{j=1}^{p_w} {\rm E}\bigg(\sum_{t=s}^{l-1} (e_t^i e_t^j - {\rm E}e_t^i e_t^j)\bigg)^2 \\
           &= O((l-s)^{-1} p_w^2),
        \end{aligned}\label{eet2}
    \end{equation}
    which yields
    \begin{equation}
        \bigg\|{1\over l-s}\sum_{t=s}^{l-1} \boldsymbol{e}_t \boldsymbol{e}_t' - {\rm E}\boldsymbol{e}_t \boldsymbol{e}_t'\bigg\| = O_p((l-s)^{-{1\over 2}} p_w).\label{eet3}
    \end{equation}
    Meanwhile, we have
    \begin{equation}
            {\rm E} \|(\bar{\boldsymbol{e}}^{(s:l)})(\bar{\boldsymbol{e}}^{(s:l)})'\| = {1\over (l-s)^2} {\rm E} \bigg\|\sum_{t=s}^{l-1}\sum_{t=s}^{l-1} \boldsymbol{e}_t\boldsymbol{e}_t'\bigg\|
            = O((l-s)^{-1} p_w^2).\label{bar_ebar_e}
    \end{equation}
    Combining \eqref{eet}, \eqref{eet3} and \eqref{bar_ebar_e} yields
    \begin{equation}
        \|\widehat{\boldsymbol{\Omega}}_e^{(s:l)} - \boldsymbol{\Omega}_e\| = O_p(p_w (l-s)^{-{1 \over 2}}). \label{ee4}
    \end{equation}
    As $\|\boldsymbol{Q}\|=1,$ we have
    \begin{equation}
        \bigg\|{1\over l-s} \boldsymbol{Q}' \widehat{\boldsymbol{\Omega}}_{e}^{(s:l)} \boldsymbol{Q}\bigg\| \leq {1\over l-s} [\|\widehat{\boldsymbol{\Omega}}_e^{(s:l)} - \boldsymbol{\Omega}_e\| + \|\boldsymbol{\Omega}_e\|] = O_p(p_w(l-s)^{-{3\over 2}} + (l-s)^{-1})\label{ee}
    \end{equation}
    holds by the Assumption 1 (4) that $\|\boldsymbol{\Omega}_e\| < \infty.$

    On the other hand, by the independence condition in Assumption 1 (3), we have
    \begin{equation}
        \begin{aligned}
            &\quad~{\rm E}\bigg[\sum_{t=s}^{l-1} (F^{i}_{t} - \bar{F}_i^{(s:l)}) (e^{i}_{t} - \bar{e}_i^{(s:l)})\bigg]^2 \\
            &= \sum_{t=s}^{l-1} {\rm E}(F^{i}_{t} - \bar{F}_i^{(s:l)})^2{\rm E}(e^{i}_{t} - \bar{e}_i^{(s:l)})^2\\
            &+ \sum_{t,t'=s, t\neq t'}^{l-1} {\rm E}(F_t^i - \bar{F}_i^{(s:l)})(F_{t'}^i - \bar{F}_i^{(s:l)}) {\rm E}(e_t^j - \bar{e}_j^{(s:l)})(e_{t'}^j - \bar{e}_j^{(s:l)}).
            \label{crossfe}
        \end{aligned}
    \end{equation}
   From   Assumption 1 (1) and (2), it follows that
    \begin{equation}
        \begin{aligned}
            {\rm E} (F_t^i - \bar{F}_i^{(s:l)})^2 = O(l-s),~{\rm E} (e_t^j - \bar{e}_j^{(s:l)})^2 = O(1),
        \end{aligned}\label{fe2}
    \end{equation}
    which yields
    \begin{equation}
        \sum_{t=s}^{l-1} {\rm E}(F^{i}_{t} - \bar{F}_i^{(s:l)})^2{\rm E}(e^{i}_{t} - \bar{e}_i^{(s:l)})^2 = O((l-s)^2).
    \end{equation}
    For the cross term in \eqref{crossfe}, we have
    \begin{equation}
        \begin{aligned}
            &~\quad\bigg|\sum_{t,t'=s, t\neq t'}^{l-1} {\rm E}(F_t^i - \bar{F}_i^{(s:l)})(F_{t'}^i - \bar{F}_i^{(s:l)}) {\rm E}(e_t^j - \bar{e}_j^{(s:l)})(e_{t'}^j - \bar{e}_j^{(s:l)})\bigg|\\
            &= (l-s) \bigg|\sum_{t,t'=s, t\neq t'}^{l-1}{\rm E}{F_t^i - \bar{F}_i^{(s:l)}\over (l-s)^{1\over 2}}{F_{t'}^i - \bar{F}_i^{(s:l)}\over (l-s)^{1\over 2}}{\rm E}(e_t^j - \bar{e}_j^{(s:l)})(e_{t'}^j - \bar{e}_j^{(s:l)})\bigg|\\
            &= O\bigg((l-s) \sum_{t,t'=s, t\neq t'}^{l-1}\bigg|{\rm E}(e_t^j - \bar{e}_j^{(s:l)})(e_{t'}^j - \bar{e}_j^{(s:l)})\bigg|\bigg).
        \end{aligned}\label{crossA32}
    \end{equation}
    By mixing property of $\boldsymbol{e}_t$ given in Assumption 1 (2), we have for some $C>0,$
    \begin{equation}
        \begin{aligned}
            \sum_{t,t'=s, t\neq t'}^{l-1}\bigg|{\rm E}(e_t^j - \bar{e}_j^{(s:l)})(e_{t'}^j - \bar{e}_j^{(s:l)})\bigg|
            &\leq C\sum_{t,t'=s, t\neq t'}^{l-1} [\alpha(|t-t'|)]^{1-{1\over 2+\kappa}}\\
            &\leq C (l-s)\sum_{k=1}^{l-s}[\alpha(k)]^{1-{1\over 2+\kappa}} = O(l-s).
        \end{aligned}\label{crossbound}
    \end{equation}
    Combining \eqref{crossA32} and \eqref{crossbound} yields
    \begin{equation}
        \bigg|\sum_{t,t'=s, t\neq t'}^{l-1} {\rm E}(F_t^i - \bar{F}_i^{(s:l)})(F_{t'}^i - \bar{F}_i^{(s:l)}) {\rm E}(e_t^j - \bar{e}_j^{(s:l)})(e_{t'}^j - \bar{e}_j^{(s:l)})\bigg| = O((l-s)^2).
    \end{equation}
    Thus,
    \begin{equation}
        \begin{aligned}
            {\rm E}\bigg\|\widehat{\boldsymbol{\Omega}}_{Fe}^{(s:l)}\bigg\|_F^2&= {1\over (l-s)^2}\sum_{i=1}^{r_F} \sum_{j=1}^{p_w} {\rm E}\bigg[\sum_{t=s}^{l-1} (F^{i}_{t} - \bar{F}_i^{(s:l)}) (e^{j}_{t} - \bar{e}_j^{(s:l)})\bigg]^2\\
            &= O(p_w r_F).
        \end{aligned}\label{fe_fnorm}
    \end{equation}
    By Assumption 1 (3), we have ${\rm E}(\boldsymbol{F}_t - \bar{\boldsymbol{F}}^{(s:l)}) (\boldsymbol{e}_t - \bar{\boldsymbol{e}}^{(s:l)})' = \boldsymbol{O}.$ This implies that $${\rm E}[\widehat{\boldsymbol{\Omega}}_{Fe}^{(s:l)}] = {1\over l-s}\sum_{t=s}^{l-1} {\rm E}(\boldsymbol{F}_t - \bar{\boldsymbol{F}}^{(s:l)}) (\boldsymbol{e}_t - \bar{\boldsymbol{e}}^{(s:l)})'=\boldsymbol{O},$$
    which implies that
    \begin{equation}
        \bigg\|\widehat{\boldsymbol{\Omega}}_{Fe}^{(s:l)}\bigg\| = O_p(p_w^{1/2} r_F^{1/2}).\label{fe}
    \end{equation}
    This yields
    \begin{equation}
        \bigg\|\dfrac{1}{l-s} \widehat{\boldsymbol{\Omega}}_{Fe}^{(s:l)} \boldsymbol{Q} \bigg\| = O_p(p_w^{1\over 2}r_F^{1\over 2} (l-s)^{-1}) \stackrel{P}{\longrightarrow} 0.\label{fet}
    \end{equation}
    Combining \eqref{wtB}, \eqref{FF}, \eqref{ee} and \eqref{fet}, we have
    \begin{equation}
        \bigg\|{1\over l-s}\widehat{\boldsymbol{\Omega}}_{\psi_2}^{(s:l)} - \int_0^1 \boldsymbol{W}_t\boldsymbol{W}_t' {\rm d}t\bigg\|
        = O_p(r_F (l-s)^{\tau- {1\over 2}} + p(l-s)^{-1})
        \stackrel{P}{\longrightarrow} 0. \label{f2f2weakl}
    \end{equation}
    For the term $\widehat{\boldsymbol{\Omega}}_{\psi_1}^{(s:l)},$ as $\|\boldsymbol{Q}^\perp\| = 1,$ by \eqref{ee4}, we have
    \begin{equation}
        \begin{aligned}
        \bigg\| \widehat{\boldsymbol{\Omega}}_{\psi_1}^{(s:l)} - {\rm Cov}(\boldsymbol{\psi}_{1t}) \bigg\| 
        &\leq \|\widehat{\boldsymbol{\Omega}}_e - \boldsymbol{\Omega}_e\| \\ &= O_p(p_w (l-s)^{-{1\over 2}}) = O_p(p(l-s)^{-{1\over 2}}) \stackrel{P}{\longrightarrow} 0.
        \end{aligned}\label{f1f1weakl}
    \end{equation}
    For the term $\widehat{\boldsymbol{\Omega}}_z^{(s:l)}$,
    by similar arguments in \eqref{eet} to \eqref{ee4}, we have
    \begin{equation}
        \bigg\|\widehat{\boldsymbol{\Omega}}_z^{(s:l)} - \boldsymbol{\Omega}_z\bigg\| = O_p(p_z(l-s)^{-1/2}) = O_p(p(l-s)^{-{1\over 2}}) \stackrel{P}{\longrightarrow} 0. \label{zzweakl}
    \end{equation}
    For the term $\widehat{\boldsymbol{\Omega}}_{z,\psi_1}^{(s:l)},$ we have ${\rm E}\widehat{\boldsymbol{\Omega}}_{z,\psi_1}^{(s:l)} = \boldsymbol{O}$ by independence of $\boldsymbol{z}_t$ and $(\boldsymbol{v}_t', \boldsymbol{e}_t')'$, and
    \begin{equation}
        \begin{aligned}
        \bigg\|\widehat{\boldsymbol{\Omega}}_{z,\psi_1}^{(s:l)} \bigg\|&= \bigg\|{1\over l-s} \sum_{t=s}^{l-1} (\boldsymbol{z}_t - \bar{\boldsymbol{z}}^{(s:l)})(\boldsymbol{e}_t - \bar{\boldsymbol{e}}^{(s:l)})' \boldsymbol{Q}^\perp\bigg\|\\ &\leq \bigg\|{1\over l-s} \sum_{t=s}^{l-1} (\boldsymbol{z}_t - \bar{\boldsymbol{z}}^{(s:l)})(\boldsymbol{e}_t - \bar{\boldsymbol{e}}^{(s:l)})'\bigg\| = \bigg\|\widehat{\boldsymbol{\Omega}}_{ze}^{(s:l)}\bigg\|.\end{aligned}\label{zf1weakl_0}
    \end{equation}
    By similar arguments in \eqref{crossfe} to \eqref{fe_fnorm}, We have
    \begin{equation}
        \begin{aligned}
        {\rm E} \bigg\|\widehat{\boldsymbol{\Omega}}_{ze}^{(s:l)}\bigg\|_F^2
        &= O(p_z p_w (l-s)^{-1}) 
        .
        \end{aligned}\label{crosscon0}
    \end{equation}
    which yields
    \begin{equation}
        \bigg\|\widehat{\boldsymbol{\Omega}}_{ze}^{(s:l)}\bigg\| = O_p(p_z^{1\over 2}p_w^{1\over 2} (l-s)^{-{1\over 2}}).
        \label{zeorder}
    \end{equation}
    Combining \eqref{zf1weakl_0} and \eqref{zeorder} yields
    \begin{equation}
        \bigg\|\widehat{\boldsymbol{\Omega}}_{z,\psi_1}^{(s:l)}\bigg\| = O_p(p_z^{1\over 2}p_w^{1\over 2} (l-s)^{-{1\over 2}}) = O_p(p(l-s)^{-{1\over 2}}). \label{zf1weakl}
    \end{equation}
    For the term $\widehat{\boldsymbol{\Omega}}_{z,\psi_2}^{(s:l)},$ we have ${\rm E} \widehat{\boldsymbol{\Omega}}_{z,\psi_2} = \boldsymbol{O}$ by Assumption 1 (3) again, and
    \begin{equation}
        \begin{aligned}
            \bigg\|\widehat{\boldsymbol{\Omega}}_{z,\psi_2}^{(s:l)}\bigg\|
            &\leq \bigg\|\widehat{\boldsymbol{\Omega}}_{zF}^{(s:l)}\bigg\|  + \bigg\|\widehat{\boldsymbol{\Omega}}_{ze}^{(s:l)}\bigg\| \|\boldsymbol{Q}\|.
        \end{aligned}\label{zFlimit0}
    \end{equation}
    For $\widehat{\boldsymbol{\Omega}}_{zF}^{(s:l)},$ by similar arguments in \eqref{crossfe} to \eqref{fe_fnorm} again, we have
    \begin{equation}
        \begin{aligned}
            {\rm E} \bigg\|\widehat{\boldsymbol{\Omega}}_{zF}^{(s:l)}\bigg\|_F^2 &= {1\over (l-s)^2}\sum_{i=1}^{p_z} \sum_{j = 1}^{r_F}{\rm E}\bigg(\sum_{t=s}^{l-1} (z_t^i-\bar{z}_i^{(s:l)}) (F_{t}^j - \bar{F}_j^{(s:l)})\bigg)^2 = O(p_zr_F),
        \end{aligned}
    \end{equation}
    which yields
    \begin{equation}
        \bigg\|\widehat{\boldsymbol{\Omega}}_{zF}^{(s:l)}\bigg\| = O_p(p_z^{1\over 2} r_F^{1\over 2}).\label{zForder}
    \end{equation}
    Combining \eqref{zFlimit0} and \eqref{zForder} yields
    \begin{equation}
        \begin{aligned}
            \bigg\|{1\over (l-s)^{1\over 2}}\widehat{\boldsymbol{\Omega}}_{z,\psi_2}^{(s:l)}\bigg\| &= O_p(p_z^{1\over 2} r_F^{1\over 2}(l-s)^{-{1\over 2}}) + O_p(p_z^{1\over 2}p_w^{1\over 2} (l-s)^{-{1}})\\
            &= O_p(p(l-s)^{-{1\over 2}}).\label{zf2weakl}
        \end{aligned}
    \end{equation}
    Finally, we consider $\widehat{\boldsymbol{\Omega}}_{\psi_1,\psi_2}^{(s:l)}.$ Note that
    \begin{equation}
        \begin{aligned}
           \widehat{\boldsymbol{\Omega}}_{\psi_1,\psi_2}^{(s:l)} &= {1\over l-s} \sum_{t=s}^{l-1} (\boldsymbol{Q}^\perp)' (\boldsymbol{e}_t - \bar{\boldsymbol{e}}^{(s:l)}) [(\boldsymbol{F}_t - \bar{\boldsymbol{F}}^{(s:l)}) + \boldsymbol{Q}' (\boldsymbol{e}_t-\bar{\boldsymbol{e}}^{(s:l)})]'\\
            &= (\boldsymbol{Q}^\perp)' [\widehat{\boldsymbol{\Omega}}_{Fe}^{(s:l)}]'  +(\boldsymbol{Q}^\perp)' \widehat{\boldsymbol{\Omega}}_{e}^{(s:l)} \boldsymbol{Q}.
        \end{aligned}\label{f1f2limit0}
    \end{equation}
    By $\|\boldsymbol{Q}^\perp\| = \|\boldsymbol{Q}\| = 1$ again and \eqref{ee4}, \eqref{fe}, we have
    \begin{equation}
        \begin{aligned}
            \bigg\|{1\over (l-s)^{1\over 2}}\widehat{\boldsymbol{\Omega}}_{\psi_1,\psi_2}^{(s:l)}\bigg\|
        &= O_p(p_w^{1\over 2} r_F^{1\over 2} (l-s)^{-{1\over 2}}) + O_p(p_w (l-s)^{-1})\\
            &= O_p(p(l-s)^{-{1\over 2}}).
        \end{aligned}
        \label{f1f2weakl}
    \end{equation}
    Thus,  Lemma \ref{lemma_cov} follows by combining \eqref{f2f2weakl}, \eqref{f1f1weakl}, \eqref{zzweakl}, \eqref{zf1weakl}, \eqref{zf2weakl}, and \eqref{f1f2weakl}.
\end{proof}

\begin{lemma}
    Suppose Assumption 1 holds and $p(l-s)^{-{1\over 2}}\rightarrow 0$ as $T\rightarrow \infty.$ We have
    \begin{equation}
        \bigg\|(l-s)^{1\over 2}\boldsymbol{PA} \widehat{\boldsymbol{\Omega}}_{xu}^{(s:l)}
        \bigg\| 
        = O_p(p_z^{1\over 2} + p_w^{1\over 2} + r_F^{1\over 2}).\label{uxl}
    \end{equation}
    \label{lemma_ux}
\end{lemma}

\begin{proof}
    By \eqref{cointA}, it follows that
    \begin{equation}
        \begin{aligned}
        \bigg\|\boldsymbol{PA} \widehat{\boldsymbol{\Omega}}_{xu}^{(s:l)}
        \bigg\|^2
        &= \|\widehat{\boldsymbol{\Omega}}_{zu}^{(s:l)}\|^2 + \|\widehat{\boldsymbol{\Omega}}_{\psi_1,u}^{(s:l)}\|^2 + \bigg\|{1\over (l-s)^{1\over 2}}\widehat{\boldsymbol{\Omega}}_{\psi_2,u}^{(s:l)}\bigg\|^2.
        \end{aligned}
    \end{equation}
    For $\widehat{\boldsymbol{\Omega}}_{zu}^{(s:l)},$ using similar arguments to \eqref{crossfe}-\eqref{fe_fnorm}, we have
    \begin{equation}
        \begin{aligned}
        {\rm E}\|\widehat{\boldsymbol{\Omega}}_{zu}^{(s:l)}\|^2 &= {1\over l-s}\sum_{i=1}^{p_z} {\rm E}\bigg[{1\over (l-s)^{1/2}} \sum_{t=s}^{l-1} (z_t^i - \bar{z}_i^{(s:l)})  (u_{t+1}-\bar{u}^{(s:l)})\bigg]^2
        = O(p_z (l-s)^{-1}).
    \end{aligned}\label{I0_ux}
    \end{equation}
    Similarly, we can show that
    \begin{equation}
        \begin{aligned}
        {\rm E}\|\widehat{\boldsymbol{\Omega}}_{\psi_1,u}^{(s:l)} \|^2
        &= O(p_w (l-s)^{-1}).
    \end{aligned}\label{I1_ux}
    \end{equation}
    For $\widehat{\boldsymbol{\Omega}}_{\psi_2,u}^{(s:l)}
    ,$ we have
    \begin{equation}
        \begin{aligned}
            {\rm E}\|\widehat{\boldsymbol{\Omega}}_{\psi_2,u}^{(s:l)}\|^2 &\leq   {\rm E}\bigg\|{1\over l-s} \sum_{t=s}^{l-1}  (u_{t+1}-\bar{u}^{(s:l)}) (\boldsymbol{F}_{t}-\bar{\boldsymbol{F}}^{(s:l)}) \bigg\|^2 + \|\boldsymbol{Q}'\|^2 \\ &\quad~ {\rm E}\bigg\|{1\over l-s} \sum_{t=s}^{l-1}  (u_{t+1}-\bar{u}^{(s:l)}) (\boldsymbol{e}_{t}-\bar{\boldsymbol{e}}^{(s:l)}) \bigg\|^2\\
            &=  {\rm E} \|\widehat{\boldsymbol{\Omega}}_{Fu}^{(s:l)}\|^2 + {\rm E} \|\widehat{\boldsymbol{\Omega}}_{eu}^{(s:l)}\|^2.
        \end{aligned}\label{I2_ux1}
    \end{equation}
    By similar arguments to \eqref{crossfe}-\eqref{fe_fnorm} again, we have
    \begin{equation}
        \begin{aligned}
            {\rm E}\|\widehat{\boldsymbol{\Omega}}_{Fu}^{(s:l)}\|^2 &= \sum_{i=1}^{r_F} {\rm E} \bigg({1\over l-s}\sum_{t=s}^{l-1}  (u_{t+1}-\bar{u}^{(s:l)}) (F_t^i - \bar{F}^{(s:l)})\bigg)^2
            = O(r_F)
        \end{aligned}\label{I2_uF}
    \end{equation}
    and
    \begin{equation}
        \begin{aligned}
        {\rm E}\|\widehat{\boldsymbol{\Omega}}_{eu}^{(s:l)}\|^2 &= {1\over l-s} \sum_{i=1}^{p_w} {\rm E} \bigg({1\over (l-s)^{1\over 2}}\sum_{t=s}^{l-1} (u_t - \bar{u}^{(s:l)}) (e_t^i - \bar{e}_i^{(s:l)})\bigg)^2
        = O(p_w(l-s)^{-1}).
        \end{aligned}\label{I2_ue}
    \end{equation}
    Combining \eqref{I2_ux1}, \eqref{I2_uF} and \eqref{I2_ue} yields
    \begin{equation}
        {\rm E}\bigg\|{1\over (l-s)^{1\over 2}}\widehat{\boldsymbol{\Omega}}_{\psi_2,u}^{(s:l)}\bigg\|^2 = O(r_F (l-s)^{-1}+ p_w(l-s)^{-2}) = O(r_F (l-s)^{-1}).\label{I2_ux}
    \end{equation}
    By Assumption 1 (3), we have ${\rm E}\widehat{\boldsymbol{\Omega}}_{xu}^{(s:l)}=\boldsymbol{O}.$ Combining \eqref{I0_ux}, \eqref{I1_ux}, \eqref{I2_ux}  yields  \eqref{uxl}.
\end{proof}

\begin{proof}[\textbf{Proof of Theorem 3.1.}]

    If we leave one $x_{lt}$ out, there are three possible cointegration structure \eqref{cointAl0}, \eqref{cointAl}, \eqref{cointAl2} for $\boldsymbol{x}_{-l,t}$ below:
\begin{itemize}
\item[(1)] {\bf Case 1:} $x_{lt}$ is an $I(0)$ regressor. By \eqref{cointA}, we have
\begin{equation}
    \boldsymbol{A}_l \boldsymbol{x}_{-l,t} := \begin{bmatrix}
        \boldsymbol{I}_{p_z-1}&\boldsymbol{O}\\
        \boldsymbol{O} & (\boldsymbol{Q}^\perp)'\\
        \boldsymbol{O} & \boldsymbol{Q}'
    \end{bmatrix} \begin{bmatrix}
        \boldsymbol{z}_{-l,t} \\
        \boldsymbol{w}_t
    \end{bmatrix} := \begin{bmatrix}
        \boldsymbol{z}_{t}^{(l)}\\
        \boldsymbol{\psi}_{1t}^{(l)}\\
        \boldsymbol{\psi}_{2t}^{(l)}
    \end{bmatrix}:=
    \boldsymbol{\psi}_{t}^{(l)}, \boldsymbol{A}_l'\boldsymbol{A}_l = \boldsymbol{A}_l\boldsymbol{A}_l' = \boldsymbol{I}_{p-1}.\label{cointAl0}
\end{equation}
Here, $\boldsymbol{\psi}_{1t}^{(l)} = \boldsymbol{\psi}_{1t}, \boldsymbol{\psi}_{2t}^{(l)} = \boldsymbol{\psi}_{2t}, $ and $\boldsymbol{\psi}_{1t}, \boldsymbol{\psi}_{2t}$ are given in \eqref{cointA}. The use of superscripts is only for matching the subsequent content.

\item[(2)] When $x_{lt},$ or $w_{l-p_z,t}$ equivalently, is an $I(1)$ regressor, by \eqref{panic}, we have
\begin{equation}
    \boldsymbol{w}_{-(l - p_z),t} = \boldsymbol{Q}_{-(l-p_z)}\boldsymbol{F}_t + \boldsymbol{e}_{-(l-p_z),t},\label{A35}
\end{equation}
where $\boldsymbol{Q}_{-(l-p_z)}, \boldsymbol{e}_{-(l-p_z),t}$ gather all rows except the $(l-p_z)-$th one in $\boldsymbol{Q},\boldsymbol{e}_t$. Further, the rank of $\boldsymbol{Q}_{-(l-p_z)}$ could be $r_F$ or $r_F-1$ due to the deletion of one row. This leads to two different cases: Case 2 and Case 3.
\begin{itemize}
\item[{\bf Case 2:}] ${\rm rank}(\boldsymbol{Q}_{-(l-p_z)}) = r_F.$ We have QR-decomposition for $\boldsymbol{Q}_{-(l-p_z)}$, i.e.
\begin{equation}
    \boldsymbol{Q}_{-(l-p_z)} = \boldsymbol{Q}_{-l} \boldsymbol{R}_{-l},\label{QRB-l}
\end{equation}
where $\boldsymbol{Q}_{-l} \in \mathbb{R}^{(p_w - 1)\times r_F}$ satisfies $\boldsymbol{Q}_{-l}'\boldsymbol{Q}_{-l} = \boldsymbol{I}_{r_F}$, and $\boldsymbol{R}_{-l}\in \mathbb{R}^{r_F\times r_F}$ is an upper-triangular matrix with positive diagnoal entries. Let $\boldsymbol{Q}_{-l}^\perp\in \mathbb{R}^{(p_w - 1) \times (p_w - 1 - r_F)}$ satisfies $(\boldsymbol{Q}_{-l}^\perp)'(\boldsymbol{Q}_{-l}^\perp) = \boldsymbol{I}_{p_w - 1 - r_F}, \boldsymbol{Q}_{-l}' \boldsymbol{Q}_{-l}^\perp = \boldsymbol{O},$ we have
\begin{equation}
    \boldsymbol{A}_l \boldsymbol{x}_{-l,t} := \begin{bmatrix}
        \boldsymbol{I}_{p_z} & {\boldsymbol{O}} \\
        {\boldsymbol{O}} & (\boldsymbol{Q}^{\perp}_{-l})' \\
        {\boldsymbol{O}} & \boldsymbol{Q}_{-l}'
    \end{bmatrix}
    \begin{bmatrix}
        \boldsymbol{z}_t\\
        \boldsymbol{w}_{-(l-p_z),t}
    \end{bmatrix}:=\begin{bmatrix}
        \boldsymbol{z}_t^{(l)}\\
        \boldsymbol{\psi}_{1t}^{(l)}\\
        \boldsymbol{\psi}_{2t}^{(l)}
    \end{bmatrix}:= \boldsymbol{\psi}_{t}^{(l)}, \boldsymbol{A}_l'\boldsymbol{A}_l = \boldsymbol{A}_l\boldsymbol{A}_l' = \boldsymbol{I}_{p-1}.
    \label{cointAl}
\end{equation}
where $\boldsymbol{\psi}_{1t}^{(l)} = (\boldsymbol{Q}^\perp_{-l})'\boldsymbol{w}_{-(l-p_z),t} = (\boldsymbol{Q}^\perp_{-l})'\boldsymbol{e}_{-(l-p_z),t}$ is a $(p_w-1-r_F)-$dimensional $I(0)$ series, $\boldsymbol{\psi}_{2t}^{(l)} = \boldsymbol{Q}_{-l}' \boldsymbol{w}_{-(l-p_z),t} =\boldsymbol{R}_{-l} \boldsymbol{F}_t + \boldsymbol{Q}_{-l}' \boldsymbol{e}_{-(l-p_z),t}$ is an $r_F-$dimensional $I(1)$ series; For convenience, we write $\boldsymbol{z}_t^{(l)}=\boldsymbol{z}_t$.

\item[{\bf Case 3:}] ${\rm rank}(\boldsymbol{Q}_{-(l-p_z)}) = r_F-1.$ We can find matrices $\boldsymbol{Q}_{-(l-p_z)}^*\in \mathbb{R}^{(p_w-1) \times (r_F - 1)}$ with full column-rank, $\boldsymbol{C}^* \in \mathbb{R}^{(r_F - 1)\times r_F}$ with full row-rank, such that $\boldsymbol{Q}_{-(l-p_z)} = \boldsymbol{Q}_{-(l-p_z)}^*\boldsymbol{C}^*.$ There exists a matrix $\boldsymbol{Q}_{-l}^*\in\mathbb{R}^{(p_w-1) \times (r_F-1)}$ satisfying $(\boldsymbol{Q}_{-l}^{*})' \boldsymbol{Q}_{-l}^* = \boldsymbol{I}_{r_F-1}$, and an upper-triangular matrix $\boldsymbol{R}_{-l} \in \mathbb{R}^{(r_F - 1)\times (r_F-1)}$ with positive diagnoal entries, such that $\boldsymbol{Q}_{-(l-p_z)}^* = \boldsymbol{Q}_{-l}^*\boldsymbol{R}_{-l}^*.$ Thus, we have
\begin{equation}
    \boldsymbol{Q}_{-(l-p_z)} = \boldsymbol{Q}_{-l}^* [\boldsymbol{R}_{-l}^* \boldsymbol{C}^*] := \boldsymbol{Q}_{-l} \boldsymbol{R}_{-l}. \label{QR*}
\end{equation} Let $\boldsymbol{Q}_{-l}^\perp\in \mathbb{R}^{(p_w-1)\times (p_w-r_F)}$ be the orthogonal complementary matrix of $\boldsymbol{Q}_{-l}$, then
\begin{equation}
    \boldsymbol{A}_l \boldsymbol{x}_{-l,t} := \begin{bmatrix}
        \boldsymbol{I}_{p_z} & {\boldsymbol{O}} \\
        {\boldsymbol{O}} & (\boldsymbol{Q}^{\perp}_{-l})' \\
        {\boldsymbol{O}} & \boldsymbol{Q}_{-l}'
    \end{bmatrix}
    \begin{bmatrix}
        \boldsymbol{z}_t\\
        \boldsymbol{w}_{-(l-p_z),t}
    \end{bmatrix}:=\begin{bmatrix}
        \boldsymbol{z}_t^{(l)}\\
        \boldsymbol{\psi}_{1t}^{(l)}\\
        \boldsymbol{\psi}_{2t}^{(l)}
    \end{bmatrix}:= \boldsymbol{\psi}_{t}^{(l)},\boldsymbol{A}_l'\boldsymbol{A}_l = \boldsymbol{A}_l\boldsymbol{A}_l' = \boldsymbol{I}_{p-1}.
    \label{cointAl2}
\end{equation}
where $\boldsymbol{\psi}_{1t}^{(l)} = (\boldsymbol{Q}_{-l}^\perp)'\boldsymbol{e}_{-(l-p_z),t}$ is a $(p_w-r_F)-$dimensional $I(0)$ process, $\boldsymbol{\psi}_{2t}^{(l)} = \boldsymbol{R}_{-l}\boldsymbol{F}_t + \boldsymbol{Q}_{-l}' \boldsymbol{e}_{-(l-p_z),t}$ is an $(r_F-1)-$dimensional $I(1)$ process. For convenience, we write $\boldsymbol{z}_t^{(l)}=\boldsymbol{z}_t.$
\end{itemize}
\end{itemize}
Using $\boldsymbol{\psi}_{t}^{(l)}$ given by \eqref{cointAl0}, \eqref{cointAl}, \eqref{cointAl2}, one can verify that for any $l = 1,2,\cdots,p$ and $\tilde{y}_t, \sigma_y^2$ defined in Assumption 2 (3),
\begin{equation}
\begin{aligned}
     {\rm CC}_l& = {1\over \sigma_y^2}[{\rm Cov}(\tilde{y}_t, \boldsymbol{\psi}_{t}^{(l)})] [{\rm Cov}(\boldsymbol{\psi}_{t}^{(l)})]^{-1} [{\rm Cov}(\boldsymbol{\psi}_{t}^{(l)}, \tilde{y}_t)]\\
     &= {1\over \sigma_y^2}\bigg[[{\rm Cov}(\tilde{y}_t, \boldsymbol{z}_{t}^{(l)})] [{\rm Cov}(\boldsymbol{z}_{t}^{(l)})]^{-1} [{\rm Cov}(\boldsymbol{z}_{t}^{(l)}, \tilde{y}_t)]\\
     &\quad +[{\rm Cov}(\tilde{y}_t, \boldsymbol{\psi}_{1t}^{(l)})] [{\rm Cov}(\boldsymbol{\psi}_{1t}^{(l)})]^{-1} [{\rm Cov}(\boldsymbol{\psi}_{1t}^{(l)}, \tilde{y}_t)]\bigg].
     \end{aligned}\label{pcc}
\end{equation}
where ${\rm CC}_l$'s are defined in Assumption 2 (3). 

Next, we first prove that
\beqn\label{ccl} \widehat{{\rm CC}}_l \stackrel{P}{\longrightarrow} {\rm CC}_l,~ l = 1,2,\cdots,p.\eeqn
Since the proof for the above Cases~1-3 are similar, we only give the proof for Case~2 in details. 

    Denote $\boldsymbol{P}_l = {\rm diag}\{\boldsymbol{I}_{p-1-r_F}, T^{-1/2} \boldsymbol{I}_{r_F}\},$ which is different from the matrix in \eqref{Pls}. Then, 
    \begin{equation}
        \begin{aligned}
            \widehat{{\rm CC}}_l
            &= \widehat{\boldsymbol{\Omega}}^{-1}_y \widehat{\boldsymbol{\Omega}}_{y, x_{-l}}
    \widehat{\boldsymbol{\Omega}}_{x_{-l}}^{-1}
    \widehat{\boldsymbol{\Omega}}_{x_{-l},y}
            \\
            &= \widehat{\boldsymbol{\Omega}}^{-1}_y  (\boldsymbol{P}_l \boldsymbol{A}_{l} \widehat{\boldsymbol{\Omega}}_{x_{-l},y})'
            (\boldsymbol{P}_l \boldsymbol{A}_{l}\widehat{\boldsymbol{\Omega}}_{x_{-l}} \boldsymbol{A}_{l}'  \boldsymbol{P}_l)^{-1} (\boldsymbol{P}_l \boldsymbol{A}_{l} \widehat{\boldsymbol{\Omega}}_{x_{-l},y}),
        \end{aligned}
    \end{equation}
    where $\boldsymbol{A}_{l}$ is defined in \eqref{cointAl}. In the following, we show the convergence of each term in $\widehat{{\rm CC}}_l$.

    Let $\boldsymbol{\Sigma}_l^0 = {\rm diag}\bigg\{{\rm Cov}(\boldsymbol{z}_t),{\rm Cov}(\boldsymbol{\psi}_{1t}^{(l)}) ,\int_0^1 (\boldsymbol{R}_{-l}\boldsymbol{W}(t))(\boldsymbol{R}_{-l}\boldsymbol{W}(t))' {\rm d}t\bigg\},$ where $\boldsymbol{R}_{-l}$ is given in \eqref{cointAl}. By similar arguments to the proof of Lemma \ref{lemma_cov}, we have
    \begin{equation}
        \begin{aligned}
        \bigg\|\boldsymbol{P}_l \boldsymbol{A}_{l}\widehat{\boldsymbol{\Omega}}_{x_{-l}} \boldsymbol{A}_{l}'  \boldsymbol{P}_l - \boldsymbol{\Sigma}_l^0\bigg\|
        = O_p(r_F T^{\tau-{1\over 2}} + p T^{-{1\over 2}}) 
        \stackrel{P}{\longrightarrow} 0,
        \end{aligned}\label{limit_CCp_omega}
    \end{equation}
    and $\int_0^1 (\boldsymbol{R}_{-l}\boldsymbol{W}(t))(\boldsymbol{R}_{-l}\boldsymbol{W}(t))' {\rm d}t$ is positive definite almost surely by the fact that $\boldsymbol{R}_{-l}$ has full row-rank, combining with the conclusion that $\int_0^1 \boldsymbol{W}(t) \boldsymbol{W}'(t) {\rm d}t$ is positive definite almost surely in Remark 3.5 of \cite{zhangIdentifyingCointegrationEigenanalysis2019} 

    Next, we focus on the convergence of
    \begin{equation}
        \boldsymbol{P}_l \boldsymbol{A}_{l} \widehat{\boldsymbol{\Omega}}_{x_{-l},y} = \bigg([\widehat{\boldsymbol{\Omega}}_{zy}]', [\widehat{\boldsymbol{\Omega}}_{\psi_1,y}^{(l)}]', [T^{-{1\over 2}}\widehat{\boldsymbol{\Omega}}_{\psi_2, y}^{(l)}]'\bigg)',
    \end{equation}
    where $\widehat{\boldsymbol{\Omega}}_{\psi_1,y}^{(l)} = \dfrac{1}{T}\sum_{t=1}^T  (\boldsymbol{\psi}_{1t}^{(l)} - \bar{\boldsymbol{\psi}}_{1}^{(l)})(y_{t+1}-\bar{y}), \widehat{\boldsymbol{\Omega}}_{\psi_2,y}^{(l)} = \dfrac{1}{T} \sum_{t=1}^T  (\boldsymbol{\psi}_{2t}^{(l)} - \bar{\boldsymbol{\psi}}_{2}^{(l)})(y_{t+1}-\bar{y}).$

    Consider
    \begin{equation}
    \begin{aligned}
        \widehat{\boldsymbol{\Omega}}_{zy} &= {1\over T}\sum_{t=1}^T\sum_{i=1}^{m_0+1}  (\boldsymbol{z}_t - \bar{\boldsymbol{z}})(\boldsymbol{z}_t - \bar{\boldsymbol{z}})' \boldsymbol{\alpha}_{i} 1_{\{t_{i-1}^{0}\leq t < t_i^0\}} + {1\over T}\sum_{t=1}^T\sum_{i=1}^{m_0+1} (\boldsymbol{z}_t - \bar{\boldsymbol{z}}) (\boldsymbol{\psi}_{1t}-\bar{\boldsymbol{\psi}}_1)'\boldsymbol{\beta}_{i} 1_{\{t_{i-1}^{0}\leq t < t_i^0\}} \\ &\qquad+ {1\over T}\sum_{t=1}^T (u_{t+1} - \bar{u})(\boldsymbol{z}_t - \bar{\boldsymbol{z}}).
    \end{aligned}
    \end{equation}
     By Lemma \ref{lemma_ux}, we have
    \begin{equation}
        \bigg\|{1\over T}\sum_{t=1}^T (u_{t+1} - \bar{u})(\boldsymbol{z}_t - \bar{\boldsymbol{z}})\bigg\| = O_p(p_z^{1/2} T^{-1/2}).\label{ut_limit}
    \end{equation}
    From Assumption 1 (3), it follows that
     $${\rm E}\bigg({1\over T} \sum_{t=1}^T\sum_{i=1}^{m_0+1} (\boldsymbol{z}_t - \bar{\boldsymbol{z}}) (\boldsymbol{\psi}_{1t}-\bar{\boldsymbol{\psi}}_1)'\boldsymbol{\beta}_{i} 1_{\{t_{i-1}^{0}\leq t < t_i^0\}}\bigg)={\bf 0}.$$
     Note that
    \begin{equation}
        \begin{aligned}
        &~\quad\bigg\|{1\over T} \sum_{t=1}^T\sum_{i=1}^{m_0+1} (\boldsymbol{z}_t - \bar{\boldsymbol{z}}) (\boldsymbol{\psi}_{1t}-\bar{\boldsymbol{\psi}}_1)'\boldsymbol{\beta}_{i} 1_{\{t_{i-1}^{0}\leq t < t_i^0\}}\bigg\| \\
        &= \bigg\|\dfrac{1}{T}\sum_{i=1}^{m_0+1} \boldsymbol{\beta}_{i}' \sum_{t=t_{i-1}^0}^{t_{i}^0-1} (\boldsymbol{\psi}_{1t}-\bar{\boldsymbol{\psi}}_1)(\boldsymbol{z}_t - \bar{\boldsymbol{z}})'\bigg\|
        \\
         &\leq\max\limits_{1\leq i \leq m_0+1}\|\boldsymbol{\beta}_i\|\bigg({1\over T}  \sum_{i=1}^{m_0+1} \bigg\|\sum_{t=t_{i-1}^0}^{t_{i}^0-1} (\boldsymbol{\psi}_{1t}-\bar{\boldsymbol{\psi}}_1)(\boldsymbol{z}_t - \bar{\boldsymbol{z}})'\bigg\|\bigg).
        \end{aligned}\label{xiz2}
    \end{equation}
    Further, 
    \begin{equation}
        \begin{aligned}
            &\quad~\sum_{t=t_{i-1}^0}^{t_{i}^0-1} (\boldsymbol{\psi}_{1t}-\bar{\boldsymbol{\psi}}_1)(\boldsymbol{z}_t - \bar{\boldsymbol{z}})'\\
            &= \sum_{t=t_{i-1}^0}^{t_{i}^0-1} [(\boldsymbol{\psi}_{1t} - \bar{\boldsymbol{\psi}}_{1}^{(t_{i-1}^0:t_i^0)}) - (\bar{\boldsymbol{\psi}}_1 - \bar{\boldsymbol{\psi}}_{1}^{(t_{i-1}^0:t_i^0)})][(\boldsymbol{z}_t - \bar{\boldsymbol{z}}^{(t_{i-1}^0:t_i^0)} )-(\bar{\boldsymbol{z}}-\bar{\boldsymbol{z}}^{(t_{i-1}^0:t_i^0)})]'\\
            &= \sum_{t=t_{i-1}^0}^{t_{i}^0-1} (\boldsymbol{\psi}_{1t} - \bar{\boldsymbol{\psi}}_{1}^{(t_{i-1}^0:t_i^0)}) (\boldsymbol{z}_t - \bar{\boldsymbol{z}}^{(t_{i-1}^0:t_i^0)} )' + (t_{i}^0-t_{i-1}^0) (\bar{\boldsymbol{\psi}}_1 - \bar{\boldsymbol{\psi}}_{1}^{(t_{i-1}^0:t_i^0)}) (\bar{\boldsymbol{z}}-\bar{\boldsymbol{z}}^{(t_{i-1}^0:t_i^0)})'.
        \end{aligned}\label{piece_con1}
    \end{equation}
    
    Similar to \eqref{crossfe}-\eqref{fe_fnorm}, one can show that
   \beqn
        {\rm E}\bigg\|\sum_{t=t_{i-1}^0}^{t_{i}^0-1} (\boldsymbol{\psi}_{1t} - \bar{\boldsymbol{\psi}}_{1}^{(t_{i-1}^0:t_i^0)}) (\boldsymbol{z}_t - \bar{\boldsymbol{z}}^{(t_{i-1}^0:t_i^0)} )'\bigg\|_F^2 = O(p_zp_w (t_{i}^0-t_{i-1}^0)) = O(p^2 T), \label{piece_con1_1}
        \eeqn
    \beqn
        {\rm E}\|\bar{\boldsymbol{\psi}}_1 - \bar{\boldsymbol{\psi}}_{1}^{(t_{i-1}^0:t_i^0)}\|^2
         &\leq& {\rm E}\|\bar{\boldsymbol{\psi}}_1\|^2 + {\rm E}\|\bar{\boldsymbol{\psi}}_{1}^{(t_{i-1}^0:t_i^0)}\|^2 \nn\\
            &=& T^{-1} {\rm E}\bigg\|{1\over \sqrt{T}}\sum_{t=1}^T (\boldsymbol{Q}^\perp)' \boldsymbol{e}_t\bigg\|^2 + (t_i^0-t_{i-1}^0)^{-1} {\rm E}\bigg\|{1\over \sqrt{t_i^0-t_{i-1}^0}}\sum_{t=t_{i-1}^0}^{t_i^0-1} (\boldsymbol{Q}^\perp)' \boldsymbol{e}_t\bigg\|^2\nn\\
            &=& O(T^{-1} p_w) + O((t_i^0-t_{i-1}^0)^{-1}p_w)\nn\\
            &=& O((t_i^0-t_{i-1}^0)^{-1}p),
       \label{barfdiffbarF0}
    \eeqn
    and
    \begin{equation}
        {\rm E}\|\bar{\boldsymbol{z}} - \bar{\boldsymbol{z}}^{(t_{i-1}^0:t_i^0)}\|^2 = O((t_i^0-t_{i-1}^0)^{-1} p).\label{barzdiffbarz}
    \end{equation}
    Combining \eqref{piece_con1} to \eqref{barzdiffbarz} yields
    \begin{equation}
            {\rm E}\bigg\|\sum_{t=t_{i-1}^0}^{t_{i}^0-1} (\boldsymbol{\psi}_{1t}-\bar{\boldsymbol{\psi}}_1)(\boldsymbol{z}_t - \bar{\boldsymbol{z}})'\bigg\|_F^2
            = O(p^2T) + O(p^2) = O(p^2T).
        \label{xiz2order}
    \end{equation}
    It follows that
    \begin{equation}
    \begin{aligned}
      &\quad~ {\rm E}\bigg({1\over T}  \sum_{i=1}^{m_0+1} \bigg\|\sum_{t=t_{i-1}^0}^{t_{i}^0-1} (\boldsymbol{\psi}_{1t}-\bar{\boldsymbol{\psi}}_1)(\boldsymbol{z}_t - \bar{\boldsymbol{z}})'\bigg\|\bigg)^2\\
        &\leq {1\over T^2} \sum_{i=1}^{m_0+1} \sum_{j=1}^{m_0+1}
        \bigg({\rm E} \bigg\|\sum_{t=t_{i-1}^0}^{t_{i}^0-1} (\boldsymbol{\psi}_{1t}-\bar{\boldsymbol{\psi}}_1)(\boldsymbol{z}_t - \bar{\boldsymbol{z}})'\bigg\|^2_F\bigg)^{1\over 2}
        \bigg({\rm E}\bigg\|\sum_{t=t_{j-1}^0}^{t_{j}^0-1} (\boldsymbol{\psi}_{1t}-\bar{\boldsymbol{\psi}}_1)(\boldsymbol{z}_t - \bar{\boldsymbol{z}})'\bigg\|^2_F\bigg)^{1\over 2}
        \\
        &= O(m_0^2 p^2 T^{-1}).
        \end{aligned}\label{m0unif}
    \end{equation}
    Combining \eqref{xiz2} and \eqref{m0unif} yields
    \begin{equation}
        \bigg\|{1\over T}\sum_{t=1}^T\sum_{i=1}^{m_0+1} (\boldsymbol{z}_t - \bar{\boldsymbol{z}}) (\boldsymbol{\psi}_{1t}-\bar{\boldsymbol{\psi}}_1)'\boldsymbol{\beta}_{i} 1_{\{t_{i-1}^{0}\leq t < t_i^0\}}\bigg\| = O_p(m_0 T^{-{1\over 2}} p)=o_p(1)
\label{piece_convergence}
    \end{equation}
    by Assumption 2 (1) that $m_0pT^{-{1\over 2}}\rightarrow 0$.

    Next, we consider
    \begin{equation}
        \begin{aligned}
            {1\over T}\sum_{t=1}^T\sum_{i=1}^{m_0+1}  (\boldsymbol{z}_t - \bar{\boldsymbol{z}})(\boldsymbol{z}_t - \bar{\boldsymbol{z}})' \boldsymbol{\alpha}_{i} 1_{\{t_{i-1}^{0}\leq t < t_i^0\}} &= \dfrac{1}{T}\sum_{i=1}^{m_0+1}\boldsymbol{\alpha}_{i}' \sum_{t=t_{i-1}^0}^{t_i^0-1} (\boldsymbol{z}_{t} - \bar{\boldsymbol{z}}) (\boldsymbol{z}_{t} - \bar{\boldsymbol{z}})'.
        \end{aligned}
    \end{equation}
    Using similar arguments to \eqref{piece_con1}-\eqref{piece_convergence}, we have
    \begin{equation}
        \begin{aligned}
            &\quad~\bigg\|\dfrac{1}{T}\sum_{i=1}^{m_0+1}\boldsymbol{\alpha}_{i}' \sum_{t=t_{i-1}^0}^{t_i^0-1} [(\boldsymbol{z}_{t} - \bar{\boldsymbol{z}}) (\boldsymbol{z}_{t} - \bar{\boldsymbol{z}})' - {\rm Cov}(\boldsymbol{z}_t)]\bigg\|
            &= O_p(m_0 p_z T^{-1/2})\stackrel{P}{\longrightarrow} 0.
        \end{aligned}   \label{zt1_lim0}
    \end{equation}
    Thus, by Assumption 2 (2),
    \begin{equation}
        \bigg\| {1\over T}\sum_{t=1}^T\sum_{i=1}^{m_0+1}  (\boldsymbol{z}_t - \bar{\boldsymbol{z}})(\boldsymbol{z}_t - \bar{\boldsymbol{z}})' \boldsymbol{\alpha}_{i} 1_{\{t_{i-1}^{0}\leq t < t_i^0\}} - {\rm Cov}(\boldsymbol{z}_t) \tilde{\boldsymbol{\alpha}}\bigg\|=O_p(m_0pT^{-{1\over 2}}) \stackrel{P}{\longrightarrow} 0.\label{zt1_limit}
    \end{equation}
    Combining \eqref{ut_limit}, \eqref{piece_convergence} and \eqref{zt1_limit} yields
    \begin{equation}
        \begin{aligned}
        \|\widehat{\boldsymbol{\Omega}}_{zy} - {\rm Cov}(\boldsymbol{z}_t) \tilde{\boldsymbol{\alpha}} \| =O_p(m_0pT^{-{1\over 2}})
        \stackrel{P}{\longrightarrow} 0.
        \end{aligned}\label{limit_CCp_z_q}
    \end{equation}
    Similarly, one can show that 
    \begin{equation}
        \begin{aligned}
    \|\widehat{\boldsymbol{\Omega}}_{\psi_1,y}^{(l)} - {\rm Cov}(\boldsymbol{\psi}_{1t}^{(l)}, \boldsymbol{\psi}_{1t})\tilde{\boldsymbol{\beta}}\|=O_p(m_0pT^{-{1\over 2}}) \stackrel{P}{\longrightarrow} 0;
        \end{aligned} \label{limit_CCp_f1_q}
    \end{equation}
    and
    \begin{equation}
        \bigg\|T^{-{1\over 2}}\widehat{\boldsymbol{\Omega}}_{\psi_2, y}^{(l)}\bigg\|=O_p(m_0pT^{-{1\over 2}}) \stackrel{P}{\longrightarrow} 0.\label{limit_CCp_f3_q}
    \end{equation}
    Denote $\boldsymbol{q}_l^0=\bigg(\tilde{\boldsymbol{\alpha}}'{\rm Cov}(\boldsymbol{z}_t), \tilde{\boldsymbol{\beta}}'{\rm Cov}(\boldsymbol{\psi}_{1t},\boldsymbol{\psi}_{1t}^{(l)}),\boldsymbol{0}'\bigg)',$ then by \eqref{limit_CCp_z_q} to \eqref{limit_CCp_f3_q},
    \begin{equation}
        \begin{aligned}
       \|\boldsymbol{P}_l \boldsymbol{A}_{l} \widehat{\boldsymbol{\Omega}}_{x_{-l},y}
       - \boldsymbol{q}_l^0\|
        =O_p(m_0pT^{-{1\over 2}})= o_p(1).
    \end{aligned}
        \label{limit_CCp_q}
    \end{equation}
    Finally, we consider $\widehat{\boldsymbol{\Omega}}_y.$ Note that
    \begin{equation}
        \begin{aligned}
        \widehat{\boldsymbol{\Omega}}_y &= {1\over T}\sum_{i=1}^{m_0+1} \sum_{t=t_{i-1}^0}^{t_i^0-1} [\boldsymbol{\alpha}_i' (\boldsymbol{z}_t - \bar{\boldsymbol{z}}) + \boldsymbol{\beta}_i' (\boldsymbol{\psi}_{1t}- \bar{\boldsymbol{\psi}}_1) + (u_{t+1} - \bar{u})]^2 \\
        &= \sum_{i=1}^{m_0+1}{t_{i}^0-t_{i-1}^0 \over T}  \boldsymbol{\alpha}_i'\widehat{\boldsymbol{\Omega}}_{z}^{(t_{i-1}^0:t_i^0)}\boldsymbol{\alpha}_i + \sum_{i=1}^{m_0+1}{t_{i}^0-t_{i-1}^0 \over T} \boldsymbol{\beta}_i' \widehat{\boldsymbol{\Omega}}_{\psi_1}^{(t_{i-1}^0:t_i^0)} \boldsymbol{\beta}_i+ \widehat{\boldsymbol{\Omega}}_{u}+ \\ &~\quad \sum_{i=1}^{m_0+1}  {t_{i}^0-t_{i-1}^0 \over T}\boldsymbol{\alpha}_i' \widehat{\boldsymbol{\Omega}}_{z,u}^{(t_{i-1}^0:t_i^0)}  + \sum_{i=1}^{m_0+1} {t_{i}^0-t_{i-1}^0 \over T}\boldsymbol{\beta}_i'\widehat{\boldsymbol{\Omega}}_{\psi_1,u}^{(t_{i-1}^0:t_i^0)} +\sum_{i=1}^{m_0+1}{t_{i}^0-t_{i-1}^0 \over T} \boldsymbol{\alpha}_i' \widehat{\boldsymbol{\Omega}}_{z,\psi_1}^{(t_{i-1}^0:t_i^0)} \boldsymbol{\beta}_i\\
        &:= L_1 + \cdots + L_6.
        \end{aligned}
    \end{equation}
    By Assumption 2 (3) and similar argument in \eqref{xiz2} to \eqref{piece_convergence}, we have 
    \begin{equation}
        \begin{aligned}
        \bigg\|L_1 + L_2 + L_3 - \sigma_y^2 \bigg\|
        \stackrel{P}{\longrightarrow} 0.
        \end{aligned}\label{Ly1_1}
    \end{equation}
    Similarly, we can show that
    \begin{equation}
        \|L_4\|
             = O_p(m_0 p_z^{ 1 \over 2} T^{-{1 \over 2}}) \stackrel{P}{\longrightarrow} 0,
        \label{Ly4}
    \end{equation}
    \begin{equation}
        \|L_5\| = O_p(m_0 T^{-{1\over 2}} p_w^{1\over 2} ) \stackrel{P}{\longrightarrow} 0,\label{Ly5}
    \end{equation}
    \begin{equation}
        \begin{aligned}
            \|L_6\| = O_p(m_0 p_z^{1/2} p_w^{1/2} T^{-{1\over 2}}) \stackrel{P}{\longrightarrow} 0.
        \end{aligned}\label{Ly6}
    \end{equation}
    Combining \eqref{Ly1_1}$\sim$\eqref{Ly6} yields
    \begin{equation}
    \|\widehat{\sigma}_y^2 - \sigma_y^2\| \stackrel{P}{\longrightarrow} 0,\label{limit_sigma}
    \end{equation}
Combining 
\eqref{limit_CCp_omega}, \eqref{limit_CCp_q} and \eqref{limit_sigma}, we have
\begin{equation}
    \|\widehat{{\rm CC}}_l - {\rm CC}_l\| \stackrel{P}{\longrightarrow} 0\nn
\end{equation}
and complete the proof of (\ref{ccl}).

Using \eqref{ccl}, we will show the sure screening property of our SICS method. For any $l \in  G^0 , m \in (G^0)^c,$ by Assumption 2 (3) that ${\rm CC}_{m} - {\rm CC}_{l} > C_0$, we have
\begin{equation}
    \begin{aligned}
        \widehat{{\rm CC}}_{m} - \widehat{{\rm CC}}_{l} &= ({\rm CC}_{m} - {\rm CC}_{l}) + [\widehat{{\rm CC}}_{m} - \widehat{{\rm CC}}_{l} - ( {\rm CC}_{m} - {\rm CC}_{l})] \\
        &> C_0 - o_p(1)
        > {1\over 2}C_0>0
    \end{aligned}
\end{equation}
holds in probability. This yields
\begin{equation}
    P(\widehat{{\rm CC}}_{m} \leq \widehat{{\rm CC}}_{l})  \longrightarrow 0
    .\label{A272}
\end{equation}
If $G^0\cap(\widehat{G})^c\neq\phi,$ there must exist $l \in G^0 $ such that $\widehat{{\rm CC}}_{l} \geq \widehat{{\rm CC}}_{j_{d_T}},$ where $ \widehat{{\rm CC}}_{j_{d_T}}$ is given in \eqref{dt}. By our Assumption 2 (3) that $|G^0| < \infty$ and $d_T \rightarrow \infty,$ there exists $m_l\in \widehat{G}\cap (G^0)^c,$ such that $ \widehat{{\rm CC}}_{j_{d_T}} > \widehat{{\rm CC}}_{m_l}.$   This implies that
\begin{equation}
     P(G^0\cap(\widehat{G})^c\neq\phi )  
    \leq \sum_{l\in G^0} P(\widehat{{\rm CC}}_{l} \geq \widehat{{\rm CC}}_{j_{d_T}}) \leq \sum_{l\in G^0 } P(\widehat{{\rm CC}}_{l} \geq \widehat{{\rm CC}}_{m_l}).\label{A273}
\end{equation}
    Combining \eqref{A272}, \eqref{A273} and the Assumption 2 (3) that $|G^0|<\infty$ yields
\begin{equation}
    P(G^0\cap(\widehat{G})^c\neq\phi ) \longrightarrow 0.
\end{equation}

This completes our proof of Theorem~3.1.
\end{proof}

\subsection{The Proof of Theorem 3.2}

To prove Theorems 3.2, we need some  lemmas. For convenience, before presenting the lemmas, we first introduce some notations.

 Let  $\tilde{\boldsymbol{x}}_t
:= (\tilde{\boldsymbol{z}}_t', \tilde{\boldsymbol{w}}_t')'$ be the  predictors   selected by one-step SICS procedure (see $(\ref{dt})$),  where $\tilde{\boldsymbol{z}}_t$ is a $\widehat{s}_z-$dimensional $I(0)$ process, $\tilde{\boldsymbol{w}}_t$ is a $\widehat{s}_w-$dimensional $I(1)$ process, and $\widehat{s}_z + \widehat{s}_w=d_T=|\widehat{G}|.$ 
Denote $\tilde{\boldsymbol{Q}}\in\mathbb{R}^{\widehat{s}_w \times r_F}, \tilde{\boldsymbol{e}}_t$ be the matrix and vector containing the rows from $\boldsymbol{Q}, \boldsymbol{e}_t$  corresponding to $\tilde{\boldsymbol{w}}_t,$ respectively. By \eqref{panic}, we have
\begin{equation}
     \tilde{\boldsymbol{w}}_t = \tilde{\boldsymbol{Q}} \boldsymbol{F}_t + \tilde{\boldsymbol{e}}_t. \label{tildestruc}
\end{equation}
Let $\tilde{r} = {\rm rank}(\tilde{\boldsymbol{Q}}),$ 
then $\tilde{\boldsymbol{Q}}$  can be represented as a product of two full-rank matrices, i.e., there exists $\tilde{\boldsymbol{Q}}^*\in\mathbb{R}^{\widehat{s}_w\times\tilde{r}}, \tilde{\boldsymbol{R}}\in\mathbb{R}^{\tilde{r}\times r_F}$ such that
\begin{equation}
    \tilde{\boldsymbol{Q}} = \tilde{\boldsymbol{Q}}^* \tilde{\boldsymbol{R}},\label{QRtilde}
\end{equation}
where $(\tilde{\boldsymbol{Q}}^*)'\tilde{\boldsymbol{Q}}^* = \boldsymbol{I}_{\tilde{r}}, {\rm rank}(\tilde{\boldsymbol{R}}) = \tilde{r}.$ We can also find $\tilde{\boldsymbol{A}},$ such that
\begin{equation}
    \tilde{\boldsymbol{A}} \tilde{\boldsymbol{x}}_t := \begin{bmatrix}
        \boldsymbol{I}_{\widehat{s}_z} & \boldsymbol{O}\\
        \boldsymbol{O} & [(\tilde{\boldsymbol{Q}}^*)^\perp]'\\
        \boldsymbol{O} & [\tilde{\boldsymbol{Q}}^*]'
    \end{bmatrix} \begin{bmatrix}
        \tilde{\boldsymbol{z}}_t \\
        \tilde{\boldsymbol{w}}_t
    \end{bmatrix}  = \begin{bmatrix}
        \tilde{\boldsymbol{z}}_t\\
        \tilde{\boldsymbol{\psi}}_{1t}\\
        \tilde{\boldsymbol{\psi}}_{2t}
    \end{bmatrix} := \tilde{\boldsymbol{\psi}}_{t},~ \tilde{\boldsymbol{A}}'\tilde{\boldsymbol{A}} = \tilde{\boldsymbol{A}}\tilde{\boldsymbol{A}}' = \boldsymbol{I}_{\widehat{s}},\label{newcoint}
\end{equation}
where $\tilde{\boldsymbol{\psi}}_{1t}, \tilde{\boldsymbol{\psi}}_{2t}$ are $I(0), I(1)$ respectively.

Denote the least square estimator (LSE)  for the regression coefficient of $\tilde{\boldsymbol{\psi}}_{t}$ over in the period $(s:l)$ by
\begin{equation}
    \begin{aligned}
        \widehat{\boldsymbol{\beta}}^{(s:l)} &= \biggl[\sum_{t=s}^{l-1} (\tilde{\boldsymbol{\psi}}_{t} - \bar{\tilde{\boldsymbol{\psi}}}^{(s:l)}) (\tilde{\boldsymbol{\psi}}_{t} - \bar{\tilde{\boldsymbol{\psi}}}^{(s:l)})' \biggr]^{-1} \biggl[\sum_{t=s}^{l-1} (\tilde{\boldsymbol{\psi}}_{t} - \bar{\tilde{\boldsymbol{\psi}}}^{(s:l)}) (y_{t+1} - \bar{y}^{(s:l)}) \biggr]
        \\
        &= [\widehat{\boldsymbol{\Omega}}_{\tilde{\psi}}^{(s:l)}]^{-1}  [\widehat{\boldsymbol{\Omega}}_{\tilde{\psi},u}^{(s:l)}],
    \end{aligned} \label{hatbetasl}
\end{equation}
and the Residual Sum of Squares (RSS) by
\begin{equation}
    \begin{aligned}
        {\rm RSS}(s:l)
        &= \sum_{t=s}^{l-1} [y_{t+1} - \bar{y}^{(s:l)} - (\widehat{\boldsymbol{\beta}}^{(s:l)})' (\tilde{\boldsymbol{\psi}}_{t} - \bar{\tilde{\boldsymbol{\psi}}}^{(s:l)})]^2.
    \end{aligned}
    \label{RSS}
\end{equation}

\begin{remark} It is easy to see that the RSS defined in (\ref{RSS}) is equivalent to that given by   \eqref{RSSsl}.
In fact, note that
\begin{equation}
    \begin{aligned}
     \widehat{\boldsymbol{\gamma}}^{(s:l)} &= \biggl[\sum_{t=s}^{l-1} (\tilde{\boldsymbol{x}}_t - \bar{\tilde{\boldsymbol{x}}}^{(s:l)}) (\tilde{\boldsymbol{x}}_t - \bar{\tilde{\boldsymbol{x}}}^{(s:l)})' \biggr]^{-1}
     \biggl[\sum_{t=s}^{l-1} (\tilde{\boldsymbol{x}}_t - \bar{\tilde{\boldsymbol{x}}}^{(s:l)}) (y_{t+1} - \bar{y}^{(s:l)}) \biggr] \\
     &= \tilde{\boldsymbol{A}}'\biggl[\sum_{t=s}^{l-1} (\tilde{\boldsymbol{\psi}}_{t} - \bar{\tilde{\boldsymbol{\psi}}}^{(s:l)}) (\tilde{\boldsymbol{\psi}}_{t} - \bar{\tilde{\boldsymbol{\psi}}}^{(s:l)})' \biggr]^{-1} \tilde{\boldsymbol{A}} \tilde{\boldsymbol{A}}' \biggl[\sum_{t=s}^{l-1} (\tilde{\boldsymbol{\psi}}_{t} - \bar{\tilde{\boldsymbol{\psi}}}^{(s:l)}) (y_{t+1} - \bar{y}^{(s:l)}) \biggr]\\
     &= \tilde{\boldsymbol{A}}'\bigg\{\biggl[\sum_{t=s}^{l-1} (\tilde{\boldsymbol{\psi}}_{t} - \bar{\tilde{\boldsymbol{\psi}}}^{(s:l)}) (\tilde{\boldsymbol{\psi}}_{t} - \bar{\tilde{\boldsymbol{\psi}}}^{(s:l)})' \biggr]^{-1} \biggl[\sum_{t=s}^{l-1} (\tilde{\boldsymbol{\psi}}_{t} - \bar{\tilde{\boldsymbol{\psi}}}^{(s:l)}) (y_{t+1} - \bar{y}^{(s:l)}) \biggr]\bigg\}
     \\ &= \tilde{\boldsymbol{A}}' \widehat{\boldsymbol{\beta}}^{(s:l)}.
    \end{aligned}\label{newcoint2}
\end{equation}
This implies that
\begin{equation}
    \begin{aligned}
        {\rm RSS}(s:l)=&\sum_{t=s}^{l-1} [y_{t+1} - \bar{y}^{(s:l)} -  (\widehat{\boldsymbol{\gamma}}^{(s:l)})' (\tilde{\boldsymbol{x}}_t - \bar{\tilde{\boldsymbol{x}}}^{(s:l)})]^2 \\
        =& \sum_{t=s}^{l-1} [y_{t+1} - \bar{y}^{(s:l)} -  (\widehat{\boldsymbol{\beta}}^{(s:l)})'\tilde{\boldsymbol{A}}\tilde{\boldsymbol{A}}' (\tilde{\boldsymbol{\psi}}_{t} - \bar{\tilde{\boldsymbol{\psi}}}^{(s:l)})]^2\\
        =&\sum_{t=s}^{l-1} [y_{t+1} - \bar{y}^{(s:l)} - (\widehat{\boldsymbol{\beta}}^{(s:l)})' (\tilde{\boldsymbol{\psi}}_{t} - \bar{\tilde{\boldsymbol{\psi}}}^{(s:l)})]^2.
    \end{aligned} \label{ols_rotation}
\end{equation}
\end{remark}

The following lemma is about the  bound of ${\rm RSS}(s:l)$ when $(s:l)$ contains no breaks.

\begin{lemma}
    Suppose $S^0 \cap (s:l)= \phi$ and $ r_F^2 (l-s)^{\tau-{1\over 2}} +r_Fd_T(l-s)^{-{1\over 2}}\rightarrow 0$, then under Assumption 1,
     \begin{equation}
        \sum_{t=s}^{l-1}(u_{t+1}-\bar{u}^{(s:l)})^2 - Cr_F(r_F+d_T) \leq {\rm RSS}(s:l) \leq \sum_{t=s}^{l-1}(u_{t+1}-\bar{u}^{(s:l)})^2 \label{RSSwithoutbreak_sl1}
    \end{equation}
    holds in probability for some $C > 0.$
    \label{lemma_RSS0break}
\end{lemma}

\begin{proof}
The second inequality in \eqref{RSSwithoutbreak_sl1} follows by the definition of least square estimator.
It suffices to prove that
\begin{equation}
    {\rm RSS}(s:l) \geq \sum_{t=s}^{l-1}(u_{t+1}-\bar{u}^{(s:l)})^2 - Cr_F(r_F+d_T) \label{RSSwithoutbreak_sl}
\end{equation}
holds in probability. Let $\boldsymbol{\beta}^{(s:l)}$ be
 the true coefficient of $\tilde{\boldsymbol{\psi}}_{t}$ in $(s:l)$. Then,
\begin{equation}
    \begin{aligned}
        \widehat{\boldsymbol{\beta}}^{(s:l)} &= \biggl[\sum_{t=s}^{l-1} (\tilde{\boldsymbol{\psi}}_{t} - \bar{\tilde{\boldsymbol{\psi}}}^{(s:l)}) (\tilde{\boldsymbol{\psi}}_{t} - \bar{\tilde{\boldsymbol{\psi}}}^{(s:l)})' \biggr]^{-1} \bigg\{\sum_{t=s}^{l-1} (\tilde{\boldsymbol{\psi}}_{t} - \bar{\tilde{\boldsymbol{\psi}}}^{(s:l)}) [(\boldsymbol{\psi}_{t} - \bar{\boldsymbol{\psi}}^{(s:l)})'\boldsymbol{\beta}^{(s:l)} + (u_{t+1}-\bar{u}^{(s:l)})] \bigg\}\\
        &= \boldsymbol{\beta}^{(s:l)} +  [\widehat{\boldsymbol{\Omega}}_{\tilde{\psi}}^{(s:l)}]^{-1} [\widehat{\boldsymbol{\Omega}}_{\tilde{\psi},u}^{(s:l)}].
    \end{aligned}\label{b_omega_u}
\end{equation}
This yields
\begin{equation}
    \begin{aligned}
        {\rm RSS}(s:l) &= \sum_{t=s}^{l-1} [(u_{t+1} - \bar{u}^{(s:l)}) - (\widehat{\boldsymbol{\beta}}^{(s:l)} - \boldsymbol{\beta}^{(s:l)})'(\tilde{\boldsymbol{\psi}}_{t} - \bar{\tilde{\boldsymbol{\psi}}}^{(s:l)})]^2\\
    &= \sum_{t=s}^{l-1}(u_{t+1}-\bar{u}^{(s:l)})^2 - 2(l-s) [\widehat{\boldsymbol{\Omega}}_{\tilde{\psi}, u}^{(s:l)}]' [\widehat{\boldsymbol{\Omega}}_{\tilde{\psi}}^{(s:l)}]^{-1} [\widehat{\boldsymbol{\Omega}}_{\tilde{\psi}, u}^{(s:l)}]
    \\ &~\quad+ (l-s)[\widehat{\boldsymbol{\Omega}}_{\tilde{\psi}, u}^{(s:l)}]' [\widehat{\boldsymbol{\Omega}}_{\tilde{\psi}}^{(s:l)}]^{-1} [\widehat{\boldsymbol{\Omega}}_{\tilde{\psi}}^{(s:l)}] [\widehat{\boldsymbol{\Omega}}_{\tilde{\psi}}^{(s:l)}]^{-1} [\widehat{\boldsymbol{\Omega}}_{\tilde{\psi}, u}^{(s:l)}]  \\
    &= \sum_{t=s}^{l-1}(u_{t+1}-\bar{u}^{(s:l)})^2 -  (l-s)[\widehat{\boldsymbol{\Omega}}_{\tilde{\psi}, u}^{(s:l)}]' [\widehat{\boldsymbol{\Omega}}_{\tilde{\psi}}^{(s:l)}]^{-1} [\widehat{\boldsymbol{\Omega}}_{\tilde{\psi}, u}^{(s:l)}]\\
    &= \sum_{t=s}^{l-1}(u_{t+1}-\bar{u}^{(s:l)})^2 -  (l-s)[\tilde{\boldsymbol{P}}\widehat{\boldsymbol{\Omega}}_{\tilde{\psi}, u}^{(s:l)}]' [\tilde{\boldsymbol{P}}\widehat{\boldsymbol{\Omega}}_{\tilde{\psi}}^{(s:l)}\tilde{\boldsymbol{P}}]^{-1} [\tilde{\boldsymbol{P}}\widehat{\boldsymbol{\Omega}}_{\tilde{\psi}, u}^{(s:l)} ],
    \end{aligned}\label{withoutbreak0}
\end{equation}
where $
    \tilde{\boldsymbol{P}} = {\rm diag}\{\boldsymbol{I}_{\widehat{s}-\tilde{r}}, (l-s)^{-{1\over 2}} \boldsymbol{I}_{\tilde{r}}\}$ is the scaled matrix.
By similar arguments to the proof of Lemma \ref{lemma_cov}, one can show that
\begin{equation}
    \begin{aligned}
        &\quad~\bigg\|\tilde{\boldsymbol{P}} [\widehat{\boldsymbol{\Omega}}_{\tilde{\psi}}^{(s:l)}] \tilde{\boldsymbol{P}} - {\rm diag}\bigg\{{\rm Cov}(\tilde{\boldsymbol{z}}_t), {\rm Cov}(\tilde{\boldsymbol{\psi}}_{1t}), \int_0^1 (\tilde{\boldsymbol{R}}\boldsymbol{W}(s))(\tilde{\boldsymbol{R}}\boldsymbol{W}(s))' {\rm d}s\bigg\}\bigg\| \\
         &= \|\tilde{\boldsymbol{P}} [\widehat{\boldsymbol{\Omega}}_{\tilde{\psi}}^{(s:l)}] \tilde{\boldsymbol{P}} - \tilde{\boldsymbol{\Omega}}_0\| \\ &= O_p(r_F(l-s)^{\tau-{1\over 2}} + d_T (l-s)^{-{1\over 2}})\stackrel{P}{\longrightarrow} 0,
    \end{aligned}\label{covtilde}
\end{equation}
where $d_T$ is the dimension of $\tilde{\boldsymbol{x}}_t.$ By Remark 3.5 in \cite{zhangIdentifyingCointegrationEigenanalysis2019}, we have that $\lambda_{\min}^{-1}\bigg(\int_0^1 \boldsymbol{W}_t \boldsymbol{W}_t' {\rm d}t\bigg) = O_p(r_F).$ It follows from  Assumption~1(2) that
\begin{equation}
    \begin{aligned}
        \lambda_{\min}(\tilde{\boldsymbol{\Omega}}_0) &\geq \min\bigg\{\lambda_{\min}({\rm Cov}(\boldsymbol{z}_t)), \lambda_{\min}
        ({\rm Cov}(\tilde{\boldsymbol{\psi}}_{1t}))
        , \lambda_{\min}\biggl(\int_0^1 [\tilde{\boldsymbol{R}}\boldsymbol{W}(s)][\tilde{\boldsymbol{R}}\boldsymbol{W}(s)]'{\rm d}s\biggr)\bigg\} \\
        &\geq Cr_F^{-1}
    \end{aligned}\label{mineigen0}
\end{equation}
holds for some $C>0.$ This implies that
\begin{equation}
    \begin{aligned}
        \lambda_{\min}(\tilde{\boldsymbol{P}} [\widehat{\boldsymbol{\Omega}}_{\tilde{\psi}}^{(s:l)}] \tilde{\boldsymbol{P}}) &\geq \lambda_{\min}(\tilde{\boldsymbol{\Omega}}_0) - \|\tilde{\boldsymbol{P}} [\widehat{\boldsymbol{\Omega}}_{\tilde{\psi}}^{(s:l)}] \tilde{\boldsymbol{P}} -\tilde{\boldsymbol{\Omega}}_0\| \\
        &\geq Cr_F^{-1} - O_p(r_F(l-s)^{\tau-{1\over 2}} + d_T (l-s)^{-{1\over 2}})\\
        &\geq {Cr_F^{-1}\over 2}
    \end{aligned}
     \label{mineigen}
\end{equation}
holds in probability by nothing that $r_F^2 (l-s)^{\tau-{1\over 2}} +r_Fd_T(l-s)^{-{1\over 2}}\rightarrow 0$. This yields
\begin{equation}
    \|(\tilde{\boldsymbol{P}} [\widehat{\boldsymbol{\Omega}}_{\tilde{\psi}}^{(s:l)}] \tilde{\boldsymbol{P}})^{-1}\| = O_p(r_F).\label{omega_sl_order}
\end{equation}
Similar to the proof of Lemma \ref{lemma_ux}, we can also show that
\begin{equation}
    \|\tilde{\boldsymbol{P}} [\widehat{\boldsymbol{\Omega}}_{\tilde{\psi}, u}^{(s:l)} ]\| = (l-s)^{-{1\over 2}}[O_p(d_T^{1\over 2}) + O_p(d_T^{1\over 2})+ O_p(r_F^{1\over 2}) ] = O_p((l-s)^{-{1\over 2}}(r_F^{1\over 2}+d_T^{1\over 2})).\label{xi_sl_order}
\end{equation}
Combining \eqref{omega_sl_order} and \eqref{xi_sl_order}, we have \eqref{RSSwithoutbreak_sl} and complete the proof of Lemma~\ref{lemma_RSS0break}.
\end{proof}

Next lemma considers the lower bound of the RSS over the period $(t^0-\Delta_l:t^0+\Delta_l).$ 

\begin{lemma}
    Suppose $t^0 \in S^0, \Delta_l$ satisfies $r_F^2 \Delta_l^{\tau- {1\over 2}} + r_Fd_T\Delta_l^{-{1\over 2}} \rightarrow 0$ as $T \rightarrow \infty.$ Under Assumption 1,
    we have  \begin{equation}
         {\rm RSS}(t^0-\Delta_l:t^0+\Delta_l) \geq \sum_{t=t^0-\Delta_l}^{t^0+\Delta_l-1} (u_{t+1}-\bar{u}^{(t^0-\Delta_l:t^0+\Delta_l)})^2 + C \Delta_l r_F^{-1} 
         \label{RSSonebreak_sl}
    \end{equation}
    holds in probability for some $C > 0.$
    \label{lemma_RSS1break}
\end{lemma}

\begin{proof}
    For simplicity, we define
     $\boldsymbol{\beta}_{01}, \boldsymbol{\beta}_{02}$ be the coefficient vector  $\boldsymbol{\beta}(t^0-\Delta_l: t^0)$ and $\boldsymbol{\beta}(t^0: t^0+\Delta_l)$  of $\tilde{\boldsymbol{\psi}}_{t}$ over the period $(t^0-\Delta_l: t^0)$ and $ (t^0: t^0+\Delta_l),$ respectively.
     Denote $\widehat{\boldsymbol{\beta}}_{12}$ be the LSE $\widehat{\boldsymbol{\beta}}^{(t^0-\Delta_l: t^0 + \Delta_l)}$ with samples over the period $(t^0-\Delta_l:t^0+\Delta_l).$
     And  write  the sample means $\bar{\tilde{\boldsymbol{\psi}}}^{(t^0-\Delta_l:t^0)}, \bar{\tilde{\boldsymbol{\psi}}}^{(t^0:t^0+\Delta_l)}, \bar{\tilde{\boldsymbol{\psi}}}^{(t^0-\Delta_l:t^0+\Delta_l)}, \bar{u}^{(t^0-\Delta_l:t^0)}, \bar{u}^{(t^0:t^0+\Delta_l)}, \bar{u}^{(t^0-\Delta_l:t^0+\Delta_l)}, \bar{y}^{(t^0-\Delta_l:t^0+\Delta_l)}$ as $\bar{\tilde{\boldsymbol{\psi}}}_{01}, \bar{\tilde{\boldsymbol{\psi}}}_{02}, \bar{\tilde{\boldsymbol{\psi}}}_{0}, \bar{u}_1, \bar{u}_2, \bar{u}_0, \bar{y}_{0}$
      respectively.

    Note that 
    \begin{equation}
        \begin{aligned}
            {\rm RSS}(t^0-\Delta_l:t^0+\Delta_l)  &=  \bigg(\sum_{t=t^0-\Delta_l}^{t^0-1}+\sum_{t=t^0}^{t^0+\Delta_l-1}\bigg) \bigg((y_{t+1} - \bar{y}_0) - [\widehat{\boldsymbol{\beta}}_{12}]' (\tilde{\boldsymbol{\psi}}_{t} - \bar{\tilde{\boldsymbol{\psi}}}_{0})\bigg)^2\\
            &:=  \widetilde{{\rm RSS}}(t^0-\Delta_l:t^0) + \widetilde{{\rm RSS}}(t^0:t^0+\Delta_l).
        \end{aligned}
    \end{equation}
    The proofs for  $\widetilde{{\rm RSS}}(t^0:t^0+\Delta_l)$ and $\widetilde{{\rm RSS}}(t^0-\Delta_l:t^0)$ are similar, here we only show
    \begin{equation}
          \widetilde{{\rm RSS}}(t^0-\Delta_l:t^0) \geq \sum_{t=t^0-\Delta_l}^{t^0-1} (u_{t+1}-\bar{u}_1)^2 + C \Delta_l r_F^{-1} \label{firstpartRSS}
    \end{equation}
    holds in probability for some $C > 0$.

   Since $\bar{\tilde{\boldsymbol{\psi}}}_{0} = {\bar{\tilde{\boldsymbol{\psi}}}_{01}+ \bar{\tilde{\boldsymbol{\psi}}}_{02}\over 2}, ~\bar{u}_0 = {\bar{u}_1+\bar{u}_2\over 2},~\bar{y}_0={\bar{\tilde{\boldsymbol{\psi}}}_{01}'\boldsymbol{\beta}_{01} + \bar{\tilde{\boldsymbol{\psi}}}_{02}'\boldsymbol{\beta}_{02} \over 2} + \bar{u}_0,$ it follows that for any  $t \in [t^0 -\Delta_l,  t^0),$
    \begin{equation}
        \begin{aligned}
            &\quad~(y_{t+1} - \bar{y}_0) - (\widehat{\boldsymbol{\beta}}_{12})' (\tilde{\boldsymbol{\psi}}_{t} - \bar{\tilde{\boldsymbol{\psi}}}_{0})\\
            &=\bigg[ u_{t+1} + \tilde{\boldsymbol{\psi}}_{t}' \boldsymbol{\beta}_{01} - \bar{u}_0  - {\bar{\tilde{\boldsymbol{\psi}}}_{01}'\boldsymbol{\beta}_{01} + \bar{\tilde{\boldsymbol{\psi}}}_{02}'\boldsymbol{\beta}_{02} \over 2}\bigg] - (\widehat{\boldsymbol{\beta}}_{12})' \bigg(\tilde{\boldsymbol{\psi}}_{t} - {\bar{\tilde{\boldsymbol{\psi}}}_{01} + \bar{\tilde{\boldsymbol{\psi}}}_{02}\over 2}\bigg)\\
            &= u_{t+1} - \bar{u}_1 + {\bar{u}_1 - \bar{u}_2\over 2} + (\tilde{\boldsymbol{\psi}}_{t} - \bar{\tilde{\boldsymbol{\psi}}}_{01})' \boldsymbol{\beta}_{01} + {\bar{\tilde{\boldsymbol{\psi}}}_{01}'\boldsymbol{\beta}_{01} - \bar{\tilde{\boldsymbol{\psi}}}_{02}'\boldsymbol{\beta}_{02} \over 2} \\
            &\qquad- (\widehat{\boldsymbol{\beta}}_{12})' \bigg(\tilde{\boldsymbol{\psi}}_{t} - \bar{\tilde{\boldsymbol{\psi}}}_{01}+ {\bar{\tilde{\boldsymbol{\psi}}}_{01} - \bar{\tilde{\boldsymbol{\psi}}}_{02}\over 2}\bigg)\\
            &= u_{t+1} - \bar{u}_1 + (\tilde{\boldsymbol{\psi}}_{t} - \bar{\tilde{\boldsymbol{\psi}}}_{01})' (\boldsymbol{\beta}_{01}-\widehat{\boldsymbol{\beta}}_{12}) + {\bar{\tilde{\boldsymbol{\psi}}}_{01}'(\boldsymbol{\beta}_{01}-\widehat{\boldsymbol{\beta}}_{12}) - \bar{\tilde{\boldsymbol{\psi}}}_{02}'(\boldsymbol{\beta}_{02}-\widehat{\boldsymbol{\beta}}_{12}) + \bar{u}_1 - \bar{u}_2\over 2}\\
            &:= u_{t+1} - \bar{u}_1 + (\tilde{\boldsymbol{\psi}}_{t} - \bar{\tilde{\boldsymbol{\psi}}}_{01})' (\boldsymbol{\beta}_{01}-\widehat{\boldsymbol{\beta}}_{12}) + E_{0}.
        \end{aligned}\label{A87}
    \end{equation}
Note that $\sum_{t=t^0-\Delta_l}^{t^0-1} (u_{t+1}-\bar{u}_1) = 0, \sum_{t=t^0-\Delta_l}^{t^0-1} (\tilde{\boldsymbol{\psi}}_{t} - \bar{\tilde{\boldsymbol{\psi}}}_{01}) = \boldsymbol{0}.$ It follows that 
    \begin{equation}
        \begin{aligned}
            &\widetilde{{\rm RSS}}(t^0-\Delta_l:t^0)\\
            = &\sum_{t=t^0-\Delta_l}^{t^0-1} \bigg(u_{t+1} - \bar{u}_1 + (\tilde{\boldsymbol{\psi}}_{t} - \bar{\tilde{\boldsymbol{\psi}}}_{01})' (\boldsymbol{\beta}_{01}-\widehat{\boldsymbol{\beta}}_{12}) + E_{0}\bigg)^2\\
            = &\sum_{t=t^0-\Delta_l}^{t^0-1} (u_{t+1} - \bar{u}_1)^2 + \Delta_l (\boldsymbol{\beta}_{01}-\widehat{\boldsymbol{\beta}}_{12})' \widehat{\boldsymbol{\Omega}}_{\tilde{\psi}}^{(t^0-\Delta_l:t^0)} (\boldsymbol{\beta}_{01}-\widehat{\boldsymbol{\beta}}_{12})+ E_{0}^2\\
             &+ 2 \Delta_l (\boldsymbol{\beta}_{01}-\widehat{\boldsymbol{\beta}}_{12})' \widehat{\boldsymbol{\Omega}}_{\tilde{\psi}, u}^{(t^0-\Delta_l:t^0)}
             + 2 E_{0} \sum_{t=t^0-\Delta_l}^{t^0-1} (u_{t+1} - \bar{u}_1) \\
             &+ 2E_{0} (\boldsymbol{\beta}_{01}-\widehat{\boldsymbol{\beta}}_{12})' \sum_{t=t^0-\Delta_l}^{t^0-1}(\tilde{\boldsymbol{\psi}}_{t} - \bar{\tilde{\boldsymbol{\psi}}}_{01})\\
            \geq &\sum_{t=t^0-\Delta_l}^{t^0-1} (u_{t+1} - \bar{u}_1)^2 + \Delta_l (\boldsymbol{\beta}_{01}-\widehat{\boldsymbol{\beta}}_{12})' \widehat{\boldsymbol{\Omega}}_{\tilde{\psi}}^{(t^0-\Delta_l:t^0)} (\boldsymbol{\beta}_{01}-\widehat{\boldsymbol{\beta}}_{12})\\
             &+ 2 \Delta_l (\boldsymbol{\beta}_{01}-\widehat{\boldsymbol{\beta}}_{12})' \widehat{\boldsymbol{\Omega}}_{\tilde{\psi}, u}^{(t^0-\Delta_l:t^0)} \\
            = &\sum_{t=t^0-\Delta_l}^{t^0-1} (u_{t+1} - \bar{u}_1)^2 + \Delta_l \lambda_{\min}\bigg(\boldsymbol{P}_l\widehat{\boldsymbol{\Omega}}_{\tilde{\psi}}^{(t^0-\Delta_l:t^0)}\boldsymbol{P}_l \bigg)\|\boldsymbol{P}_l^{-1}[\boldsymbol{\beta}_{01}-\widehat{\boldsymbol{\beta}}_{12}]\|^2 \\ &~\quad+ 2 \Delta_l [\boldsymbol{P}_l^{-1}(\boldsymbol{\beta}_{01}-\widehat{\boldsymbol{\beta}}_{12})]' \boldsymbol{P}_l\widehat{\boldsymbol{\Omega}}_{\tilde{\psi}, u}^{(t^0-\Delta_l:t^0)}\\
            := &\sum_{t=t^0-\Delta_l}^{t^0-1} (u_{t+1} - \bar{u}_1)^2 + \Gamma_1 + \Gamma_2,
        \end{aligned}\label{RSSleftpart}
    \end{equation}
where $\boldsymbol{P}_l = {\rm diag}\{\boldsymbol{I}_{\widehat{s}-\tilde{r}}, \Delta_l^{-{1\over 2}} \boldsymbol{I}_{\tilde{r}}\}.$

Similar  to the proof of Lemmas \ref{lemma_cov} and \ref{lemma_ux}, we can show that
\begin{equation}
\lambda_{\min}^{-1}\bigg(\boldsymbol{P}_l\widehat{\boldsymbol{\Omega}}_{\tilde{\psi}}^{(t^0-\Delta_l:t^0)}\boldsymbol{P}_l \bigg)=O_p(r_F),~~
\|\Delta_l^{1\over 2}\boldsymbol{P}_l\widehat{\boldsymbol{\Omega}}_{\tilde{\psi}, u}^{(t^0-\Delta_l:t^0-\Delta_l)}\|^2 = O_p(r_F+d_T). \label{aa}
\end{equation}
To show \eqref{firstpartRSS}, it suffices to show that
       \begin{equation} \|\boldsymbol{P}_l^{-1}[\boldsymbol{\beta}_{01}-\widehat{\boldsymbol{\beta}}_{12}]\|\geq C > 0 \label{lem5-1}
    \end{equation}
    for some constant $C>0$  and
  \begin{equation}
        {\|\Gamma_2\|/ \|\Gamma_1\|}=o_p(1). \label{lem5-2}
    \end{equation}

 Next, we show \eqref{lem5-1}. Note that %
\begin{equation}
        \begin{aligned}
       &\widehat{\boldsymbol{\beta}}_{12} - \boldsymbol{\beta}_{01}\\
       =& \biggl[\sum_{t=t^0-\Delta_l}^{t^0+\Delta_l-1} (\tilde{\boldsymbol{\psi}}_{t} - \bar{\tilde{\boldsymbol{\psi}}}_{0}) (\tilde{\boldsymbol{\psi}}_{t} - \bar{\tilde{\boldsymbol{\psi}}}_{0})'\biggr]^{-1} \biggl[\sum_{t=t^0-\Delta_l}^{t^0+\Delta_l-1} (\tilde{\boldsymbol{\psi}}_{t}  - \bar{\tilde{\boldsymbol{\psi}}}_{0}) (y_{t+1} - \bar{y}_0) \biggr] - \boldsymbol{\beta}_{01}\\
       =& \biggl[\sum_{t=t^0-\Delta_l}^{t^0+\Delta_l-1} (\tilde{\boldsymbol{\psi}}_{t} - \bar{\tilde{\boldsymbol{\psi}}}_{0}) (\tilde{\boldsymbol{\psi}}_{t} - \bar{\tilde{\boldsymbol{\psi}}}_{0})'\biggr]^{-1} \biggl[\sum_{t=t^0-\Delta_l}^{t^0+\Delta_l-1} (\tilde{\boldsymbol{\psi}}_{t}  - \bar{\tilde{\boldsymbol{\psi}}}_{0})\bigg\{ (y_{t+1} - \bar{y}_0) - (\tilde{\boldsymbol{\psi}}_{t} - \bar{\tilde{\boldsymbol{\psi}}}_{0})'\boldsymbol{\beta}_{01}\bigg\}  \biggr],
    \end{aligned} \label{a91}
    \end{equation}
and
\begin{equation}
\begin{aligned}
    &\sum_{t=t^0-\Delta_l}^{t^0+\Delta_l-1} (\tilde{\boldsymbol{\psi}}_{t}  - \bar{\tilde{\boldsymbol{\psi}}}_{0})\bigg\{ (y_{t+1} - \bar{y}_0) - (\tilde{\boldsymbol{\psi}}_{t} - \bar{\tilde{\boldsymbol{\psi}}}_{0})'\boldsymbol{\beta}_{01}\bigg\}\\
    =  &\sum_{t=t^0-\Delta_l}^{t^0+\Delta_l-1} (\tilde{\boldsymbol{\psi}}_{t}  - \bar{\tilde{\boldsymbol{\psi}}}_{0})(u_{t+1} - \bar{u}_0) + \sum_{t=t^0-\Delta_l}^{t^0-1} (\tilde{\boldsymbol{\psi}}_{t}  - \bar{\tilde{\boldsymbol{\psi}}}_{0})\bigg[\bigg(\tilde{\boldsymbol{\psi}}_{t}'\boldsymbol{\beta}_{01} - {\bar{\tilde{\boldsymbol{\psi}}}_{01}'\boldsymbol{\beta}_{01} +\bar{\tilde{\boldsymbol{\psi}}}_{02}'\boldsymbol{\beta}_{02} \over 2}\bigg)\\
    &-(\tilde{\boldsymbol{\psi}}_{t}- \bar{\tilde{\boldsymbol{\psi}}}_{0})'\boldsymbol{\beta}_{01}\bigg]+ \sum_{t=t^0}^{t^0+\Delta_l-1} (\tilde{\boldsymbol{\psi}}_{t}  - \bar{\tilde{\boldsymbol{\psi}}}_{0})\bigg[\bigg(\tilde{\boldsymbol{\psi}}_{t}'\boldsymbol{\beta}_{02} - {\bar{\tilde{\boldsymbol{\psi}}}_{01}'\boldsymbol{\beta}_{01} +\bar{\tilde{\boldsymbol{\psi}}}_{02}'\boldsymbol{\beta}_{02} \over 2}\bigg)-(\tilde{\boldsymbol{\psi}}_{t} - \bar{\tilde{\boldsymbol{\psi}}}_{0})'\boldsymbol{\beta}_{01}\bigg]\\
    =&\sum_{t=t^0-\Delta_l}^{t^0+\Delta_l-1} (\tilde{\boldsymbol{\psi}}_{t}  - \bar{\tilde{\boldsymbol{\psi}}}_{0})(u_{t+1} - \bar{u}_0) - {1\over 2}\sum_{t=t^0-\Delta_l}^{t^0-1} (\tilde{\boldsymbol{\psi}}_{t}  - \bar{\tilde{\boldsymbol{\psi}}}_{0}) \bar{\tilde{\boldsymbol{\psi}}}_{02}'(\boldsymbol{\beta}_{02} - \boldsymbol{\beta}_{01})\\
     &+ \sum_{t=t^0}^{t^0+\Delta_l-1} (\tilde{\boldsymbol{\psi}}_{t}  - \bar{\tilde{\boldsymbol{\psi}}}_{0})\tilde{\boldsymbol{\psi}}_{t}'(\boldsymbol{\beta}_{02} - \boldsymbol{\beta}_{01})- {1\over 2}\sum_{t=t^0}^{t^0+\Delta_l-1} (\tilde{\boldsymbol{\psi}}_{t}  - \bar{\tilde{\boldsymbol{\psi}}}_{0}) \bar{\tilde{\boldsymbol{\psi}}}_{02}'(\boldsymbol{\beta}_{02} - \boldsymbol{\beta}_{01})
    \\
    :=&\sum_{t=t^0-\Delta_l}^{t^0+\Delta_l-1} (\tilde{\boldsymbol{\psi}}_{t}  - \bar{\tilde{\boldsymbol{\psi}}}_{0})(u_{t+1} - \bar{u}_0) + \sum_{t=t^0}^{t^0+\Delta_l-1} (\tilde{\boldsymbol{\psi}}_{t}  - \bar{\tilde{\boldsymbol{\psi}}}_{0})\tilde{\boldsymbol{\psi}}_{t}'(\boldsymbol{\beta}_{02} - \boldsymbol{\beta}_{01}),
\end{aligned}\label{a94}
\end{equation}
where the last equality follows by $\sum_{t=t^0-\Delta_l}^{t^0+\Delta_l-1} (\tilde{\boldsymbol{\psi}}_{t}-\bar{\tilde{\boldsymbol{\psi}}}_{0})={\bf 0}.$
Define $\boldsymbol{\theta}_{12}=\boldsymbol{\beta}_{02} - \boldsymbol{\beta}_{01},$ we further write the second term in \eqref{a94} as
\begin{equation}
\begin{aligned}
    &\sum_{t=t^0}^{t^0+\Delta_l-1} (\tilde{\boldsymbol{\psi}}_{t}  - \bar{\tilde{\boldsymbol{\psi}}}_{0})\tilde{\boldsymbol{\psi}}_{t}'\boldsymbol{\theta}_{12}\\
     = &\sum_{t=t^0}^{t^0+\Delta_l-1} \bigg(\tilde{\boldsymbol{\psi}}_{t}  - \bar{\tilde{\boldsymbol{\psi}}}_{02} - {\bar{\tilde{\boldsymbol{\psi}}}_{01}-\bar{\tilde{\boldsymbol{\psi}}}_{02}\over 2}\bigg)\bigg(\tilde{\boldsymbol{\psi}}_{t}  - \bar{\tilde{\boldsymbol{\psi}}}_{02} + \bar{\tilde{\boldsymbol{\psi}}}_{02}\bigg)'\boldsymbol{\theta}_{12}\\
    = &\sum_{t=t^0}^{t^0+\Delta_l-1} (\tilde{\boldsymbol{\psi}}_{t}  - \bar{\tilde{\boldsymbol{\psi}}}_{02})(\tilde{\boldsymbol{\psi}}_{t}  - \bar{\tilde{\boldsymbol{\psi}}}_{02})'\boldsymbol{\theta}_{12} - {1\over 2}\Delta_l(\bar{\tilde{\boldsymbol{\psi}}}_{01}-\bar{\tilde{\boldsymbol{\psi}}}_{02}) \bar{\tilde{\boldsymbol{\psi}}}_{02}'\boldsymbol{\theta}_{12}.
\end{aligned}\label{a95}
\end{equation}
Combining \eqref{a91}, \eqref{a94}, \eqref{a95} yields
\begin{equation}
    \begin{aligned}
        &\widehat{\boldsymbol{\beta}}_{12} - \boldsymbol{\beta}_{01}\\
        =& \biggl[\sum_{t=t^0-\Delta_l}^{t^0+\Delta_l-1} (\tilde{\boldsymbol{\psi}}_{t} - \bar{\tilde{\boldsymbol{\psi}}}_{0}) (\tilde{\boldsymbol{\psi}}_{t} - \bar{\tilde{\boldsymbol{\psi}}}_{0})'\biggr]^{-1} \biggl[\sum_{t=t^0-\Delta_l}^{t^0+\Delta_l-1} (\tilde{\boldsymbol{\psi}}_{t}  - \bar{\tilde{\boldsymbol{\psi}}}_{0})(u_{t+1} - \bar{u}_0)\\
        &+ \sum_{t=t^0}^{t^0+\Delta_l-1} (\tilde{\boldsymbol{\psi}}_{t}  - \bar{\tilde{\boldsymbol{\psi}}}_{02})(\tilde{\boldsymbol{\psi}}_{t}  - \bar{\tilde{\boldsymbol{\psi}}}_{02})'\boldsymbol{\theta}_{12}  - {1\over 2}\Delta_l(\bar{\tilde{\boldsymbol{\psi}}}_{01}-\bar{\tilde{\boldsymbol{\psi}}}_{02}) \bar{\tilde{\boldsymbol{\psi}}}_{02}'\boldsymbol{\theta}_{12} \biggr]\\
        := &\Delta_\beta^1 + \Delta_\beta^2 + \Delta_\beta^3,
    \end{aligned}\label{a68}
\end{equation}
Similar to the proofs of Lemmas \ref{lemma_cov} and \ref{lemma_ux}, we have
\begin{equation}
\begin{aligned}
    \|\boldsymbol{P}_l^{-1} \Delta_\beta^1\| &\leq \bigg\|[\boldsymbol{P}_l\widehat{\boldsymbol{\Omega}}_{\tilde{\psi}}^{(t^0-\Delta_l:t^0+\Delta_l)}\boldsymbol{P}_l]^{-1}\bigg\| \bigg\|\boldsymbol{P}_l \widehat{\boldsymbol{\Omega}}_{\tilde{\psi}, u}^{(t^0-\Delta_l:t^0+\Delta_l)}\bigg\|\\ &=O_p(r_F)\times O_p((r_F^{1\over 2}+d_T^{1\over2}) \times (2\Delta_l)^{-{1\over 2}})
    \stackrel{P}{\longrightarrow}0.
    \end{aligned}
    \label{db1vsdb2}
\end{equation}
For $\Delta_\beta^2,$ since the true coefficient of unit-root components $\tilde{\boldsymbol{\psi}}_{2t}$ are all zero, it follows that the last $\tilde{r}$ components in $\boldsymbol{\beta}_{01}, \boldsymbol{\beta}_{02}$ and $\boldsymbol{\theta}_{12}$ are all zero. Thus,
\begin{equation}\boldsymbol{P}_l \boldsymbol{\theta}_{12} = \boldsymbol{\theta}_{12}, \quad \hbox{and}\quad \boldsymbol{P}_l^{-1} \boldsymbol{\theta}_{12} = \boldsymbol{\theta}_{12}.\label{Plinvariant}
\end{equation}

Consequently,
\begin{equation}
    \begin{aligned}
        \boldsymbol{P}_l^{-1} \Delta_\beta^2 &= [\boldsymbol{P}_l\widehat{\boldsymbol{\Omega}}_{\tilde{\psi}}^{(t^0-\Delta_l:t^0+\Delta_l)}\boldsymbol{P}_l]^{-1}\boldsymbol{P}_l \bigg({1\over 2}\widehat{\boldsymbol{\Omega}}_{\tilde{\psi}}^{(t^0:t^0+\Delta_l)}\bigg) \boldsymbol{\theta}_{12}\\ &= {1\over 2}  [\boldsymbol{P}_l\widehat{\boldsymbol{\Omega}}_{\tilde{\psi}}^{(t^0-\Delta_l:t^0+\Delta_l)}\boldsymbol{P}_l]^{-1}[\boldsymbol{P}_l\widehat{\boldsymbol{\Omega}}_{\tilde{\psi}}^{(t^0:t^0+\Delta_l)}\boldsymbol{P}_l] \boldsymbol{\theta}_{12}.
    \end{aligned}
\end{equation}
Note that when $T \rightarrow \infty,$
\begin{equation}
    \boldsymbol{P}_l\widehat{\boldsymbol{\Omega}}_{\tilde{\psi}}^{(t^0-\Delta_l:t^0+\Delta_l)}\boldsymbol{P}_l \stackrel{P}{\longrightarrow} \tilde{\boldsymbol{\Omega}}_0,~ \, \, \boldsymbol{P}_l\widehat{\boldsymbol{\Omega}}_{\tilde{\psi}}^{(t^0:t^0+\Delta_l)}\boldsymbol{P}_l \stackrel{P}{\longrightarrow} \tilde{\boldsymbol{\Omega}}_0.
\end{equation}
where $\tilde{\boldsymbol{\Omega}}_0$ is defined in \eqref{covtilde}. Therefore, $[\boldsymbol{P}_l\widehat{\boldsymbol{\Omega}}_{\tilde{\psi}}^{(t^0-\Delta_l:t^0+\Delta_l)}\boldsymbol{P}_l]^{-1}[\boldsymbol{P}_l\widehat{\boldsymbol{\Omega}}_{\tilde{\psi}}^{(t^0:t^0+\Delta_l)}\boldsymbol{P}_l]$ converges to identity matrix by continuous mapping theorem, and then
\begin{equation}
    \lambda_{\min}\bigg([\boldsymbol{P}_l\widehat{\boldsymbol{\Omega}}_{\tilde{\psi}}^{(t^0-\Delta_l:t^0+\Delta_l)}\boldsymbol{P}_l]^{-1}[\boldsymbol{P}_l\widehat{\boldsymbol{\Omega}}_{\tilde{\psi}}^{(t^0:t^0+\Delta_l)}\boldsymbol{P}_l]\bigg) \geq {1\over 4}
\end{equation}
holds in probability. This implies that
\begin{equation}
    \|\boldsymbol{P}_l^{-1} \Delta_\beta^2\| \geq C>0, \, \, \hbox{in probability.} \label{db2order}
\end{equation}

For $\Delta_\beta^3,$ note that
\begin{equation}
    \begin{aligned}
        \|\boldsymbol{P}_l^{-1}\Delta_\beta^3\|&\leq\bigg\|[ \boldsymbol{P}_l\widehat{\boldsymbol{\Omega}}_{\tilde{\psi}}^{(t^0-\Delta_l:t^0+\Delta_l)}\boldsymbol{P}_l]^{-1}\bigg\|\bigg\|\boldsymbol{P}_l {1\over 2}(\bar{\tilde{\boldsymbol{\psi}}}_{01}-\bar{\tilde{\boldsymbol{\psi}}}_{02}) \bar{\tilde{\boldsymbol{\psi}}}_{02}'\boldsymbol{\theta}_{12}\bigg\|\\ &\leq O_p(r_F)\times \|\boldsymbol{P}_l (\bar{\tilde{\boldsymbol{\psi}}}_{02} - \bar{\tilde{\boldsymbol{\psi}}}_{01})\| \|\bar{\tilde{\boldsymbol{\psi}}}_{02}' \boldsymbol{\theta}_{12}\|. 
    \end{aligned}\label{a133}
\end{equation}
Write $\boldsymbol{\theta}_{12} = [\boldsymbol{\theta}_z',\boldsymbol{\theta}_f',\boldsymbol{0}']',$ we have
\begin{equation}
    \boldsymbol{P}_l (\bar{\tilde{\boldsymbol{\psi}}}_{02} - \bar{\tilde{\boldsymbol{\psi}}}_{01}) = {1\over \Delta_l^{1\over 2}} \begin{bmatrix}
        {1\over \Delta_l^{1\over 2}} \sum_{t=t^0-\Delta_l}^{t^0-1} (\tilde{\boldsymbol{z}}_t - \tilde{\boldsymbol{z}}_{t+\Delta_l})\\
        {1\over \Delta_l^{1\over 2}} \sum_{t=t^0-\Delta_l}^{t^0-1} (\tilde{\boldsymbol{\psi}}_{1,t} - \tilde{\boldsymbol{\psi}}_{1,t+\Delta_l})\\
        {1\over \Delta_l} \sum_{t=t^0-\Delta_l}^{t^0-1} (\tilde{\boldsymbol{\psi}}_{2,t} - \tilde{\boldsymbol{\psi}}_{2,t+\Delta_l})
    \end{bmatrix},\end{equation}
    \begin{equation}\bar{\tilde{\boldsymbol{\psi}}}_{02}' \boldsymbol{\theta}_{12} = {1\over \Delta_l^{1\over 2}} \bigg[
        {1\over \Delta_l^{1\over 2}} \sum_{t=t^0}^{t^0+\Delta_l-1} \tilde{\boldsymbol{z}}_t'\boldsymbol{\theta}_z+
        {1\over \Delta_l^{1\over 2}} \sum_{t=t^0}^{t^0+\Delta_l-1} \tilde{\boldsymbol{\psi}}_{1,t}'\boldsymbol{\theta}_f\bigg].
 \label{barE1}
\end{equation}
Similar  to the proof of Lemma \ref{lemma_ux}, we have
\begin{equation}
    \|\boldsymbol{P}_l (\bar{\tilde{\boldsymbol{\psi}}}_{02} - \bar{\tilde{\boldsymbol{\psi}}}_{01})\| = O_p(\Delta_l^{-{1\over 2}}) \times [O_p(d_T^{1\over 2}) + O_p(d_T^{1\over 2}) + O_p(r_F^{1\over 2})] = O_p(\Delta_l^{-{1\over 2}} (r_F^{1\over 2}+d_T^{1\over 2})).\label{orderbarE1}
\end{equation}
\begin{equation}
    \|\bar{\tilde{\boldsymbol{\psi}}}_{02}' \boldsymbol{\theta}_{12}\| = O_p(\Delta_l^{-{1\over 2}})\times [O_p(d_T^{1\over 2}) + O_p(d_T^{1\over 2})] = O_p(\Delta_l^{-{1\over 2}} d_T^{1\over 2}).\label{orderbarE2}
\end{equation}
Combining \eqref{a133}, \eqref{orderbarE1}, \eqref{orderbarE2} yields
\begin{equation}
    \begin{aligned}
        \|\boldsymbol{P}_l^{-1}\Delta_\beta^3\| = O_p(\Delta_l^{-1} r_Fd_T^{1\over 2}(r_F^{1\over 2}+d_T^{1\over 2})) \stackrel{P}{\longrightarrow}0.\label{db3vsdb2}
    \end{aligned}
\end{equation}
Combining \eqref{db1vsdb2}, \eqref{db2order}, \eqref{db3vsdb2} yields \eqref{lem5-1}.

{\it Proof of \eqref{lem5-2}.}  Using \eqref{aa}, we have
\begin{equation}
    \begin{aligned}
        {\|\Gamma_2\|^2 \over \|\Gamma_1\|} &\leq 8 {\Delta_l^2\|\boldsymbol{P}_l^{-1} (\boldsymbol{\beta}_{01}-\widehat{\boldsymbol{\beta}}_{12})\|^2 \|\boldsymbol{P}_l\widehat{\boldsymbol{\Omega}}_{\tilde{\psi}, u}^{(t^0-\Delta_l:t^0-\Delta_l)}\|^2\over \Delta_l\|\boldsymbol{P}_l^{-1} (\widehat{\boldsymbol{\beta}}_{12} - \boldsymbol{\beta}_{01})\|^2 \lambda_{\min} (\boldsymbol{P}_l \widehat{\boldsymbol{\Omega}}_{\tilde{\psi}}^{(t^0-\Delta_l:t^0-\Delta_l)}\boldsymbol{P}_l)}\\
        &= O_p\bigg({r_F\|\Delta_l^{1\over 2}\boldsymbol{P}_l\widehat{\boldsymbol{\Omega}}_{\tilde{\psi}, u}^{(t^0-\Delta_l:t^0-\Delta_l)}\|^2}\bigg) = O_p(r_F^{2}+r_Fd_T).
    \end{aligned}\label{gamma4order_0}
\end{equation}
This implies that
\begin{equation}
    {\|\Gamma_2\|^2 / \|\Gamma_1\|^2} = O_p(r_F^{2}+r_Fd_T)\times O_p(r_F\Delta_l^{-1})\stackrel{P}{\longrightarrow} 0\label{gamma4order}
\end{equation}
and proves \eqref{lem5-2}. Combining \eqref{lem5-1}, \eqref{lem5-2} yields \eqref{firstpartRSS}.
\end{proof}

\begin{remark}
    For any $m$ points $1=t_0 < t_1 < t_2 <\cdots <t_m< t_{m+1}=T+1,$ we have
    \begin{equation}
        {1\over m+1}\bigg|\sum_{t=1}^T (u_{t+1} - \bar{u})^2 - \sum_{i=1}^{m+1} \sum_{t=t_{i-1}}^{t_i-1} (u_{t+1} - \bar{u}^{(t_{i-1}:t_i)})^2\bigg| = O_p(1).\label{upiece}
    \end{equation}
    In fact, by some elementary computations, we have
    \begin{equation}
        \sum_{t=1}^T (u_{t+1} - \bar{u})^2 - \sum_{i=1}^{m+1} \sum_{t=t_{i-1}}^{t_i-1} (u_{t+1} - \bar{u}^{(t_{i-1}:t_i)})^2 = \sum_{i=1}^{m+1} \biggl({1\over (t_i-t_{i-1})^{1\over 2}} \sum_{t=t_{i-1}}^{t_i-1}u_{t+1} \biggr)^2 - T\bar{u}^2.\label{upiece0}
    \end{equation}
    Thus, by Assumption 1 (2) and a similar argument to \eqref{eet2}, it follows that
    \begin{equation}
        \begin{aligned}
            &~\quad{\rm E} \bigg\{{1\over m+1}\sum_{i=1}^{m+1} \biggl({1\over (t_i-t_{i-1})^{1\over 2}} \sum_{t=t_{i-1}}^{t_i-1}u_{t+1} \biggr)^2\bigg\}^2\\ &= {1\over (m+1)^2} \sum_{i=1}^{m+1}\sum_{j=1}^{m+1}{\rm E} \biggl({1\over (t_i-t_{i-1})^{1\over 2}} \sum_{t=t_{i-1}}^{t_i-1}u_{t+1} \biggr)^2\biggl({1\over (t_j-t_{j-1})^{1\over 2}} \sum_{t=t_{j-1}}^{t_j-1}u_{t+1} \biggr)^2\\
            &= O(1),
        \end{aligned}\label{upiece1}
    \end{equation}
    and
    \begin{equation}
        {\rm E} \bigg({T\bar{u}^2\over m+1}\bigg)^2 = {1\over (m+1)^2}{\rm E}\bigg(T^{-{1\over 2}}\sum_{t=1}^T u_{t+1}\bigg)^4 = O(1),\label{upiece2}
    \end{equation}
    Combining \eqref{upiece0} to \eqref{upiece2}, we have \eqref{upiece}.
\end{remark}

\begin{lemma} Let $h_T$ be defined as in  Section 2.3,
   $r_F^2 h_T^{\tau-{1\over 2}}+r_Fd_Th_T^{-{1\over 2}} \rightarrow 0$ and $ m_0h_Tr_F=o(\min\limits_{1\leq i \leq m_0+1}|t_i^0-t_{i-1}^0|)$ hold. Then, under  Assumptions~1 and 2, there exists a constant $\nu > 0$ such that
   for all $m < m_0$,
    \begin{equation}
        P\bigg\{
            {\rm RSS}(\{\tilde{t}_1, \cdots, \tilde{t}_m\}) > \sum_{t=1}^{T} (u_{t+1}-\bar{u})^2 + \nu r_F^{-1} \min\limits_{1\leq i \leq m_0+1}|t_i^0-t_{i-1}^0|
        \bigg\} \rightarrow 1,
    \end{equation}
    where $(\tilde{t}_1, \cdots, \tilde{t}_m) = \mathop{\arg\min}_{(t_1,\cdots,t_m): \, h_T=O(\min_{i\neq j}|t_i-t_j|)} {\rm RSS}(\{t_1,\cdots,t_m\}).$
    \label{lemmaA3}
\end{lemma}

\begin{proof}
     Recall that our true breaks set is given by $S^0 = \{t_1^0, \cdots, t_{m_0}^0\}$ with $m_0 = |S^0|.$ 
If $m < m_0,$ it can be noted that there exist at least one true break $t_{i_0}^0, 1\leq i_0 \leq m_0$ subject to $|\tilde{t}_j - t_{i_0}^0| > \Delta_{i_0}$ for $1\leq j \leq m,$ where $\Delta_{i_0} = \dfrac{t_{i_0}^0 - t_{i_0-1}^0}{4}.$ As the true coefficient vector changes in $t_{i_0}^0 - \Delta_{i_0} \leq t < t_{i_0}^0 + \Delta_{i_0},$ the sum of squared residual or RSS in the time period $[t_{i_0}^0-\Delta_{i_0}, t_{i_0}^0 + \Delta_{i_0}]$ should be large.

Define
\begin{equation}
\begin{aligned}
        A_j(T,\Delta_j) &= \{(t_1, \cdots, t_m): 0 < t_1 < \cdots < t_m < T, |t_s - t_j^0|> \Delta_j^0, h_T=O(|t_i-t_k|) \\&for ~i\neq k, 1 \leq s \leq m\},
        \end{aligned}
\end{equation}
since $m < m_0, \{\tilde{t}_1, \cdots, \tilde{t}_m\}$ must belong to one of the $A_j (T, \Delta_j)'s,$ say $A_{i_0}(T, \Delta_{i_0}).$ It's enough to show that $\forall j, 1\leq j \leq m_0,$
\begin{equation}
        P\bigg\{\min\limits_{(t_1, \cdots, t_m)\in A_j(T, \Delta_j)} {\rm RSS}(\{t_1, \cdots, t_m\}) > \sum_{t=1}^{T} (u_{t+1}-\bar{u})^2 + \nu r_F^{-1}\min\limits_{1\leq j \leq m_0+1}|t_j^0 - t_{j-1}^0|\bigg\} \rightarrow 1,\label{a147}
\end{equation}
for some positive constant $\nu > 0.$
Note that
\begin{equation}
        {\rm RSS}(\{\tilde{t}_1, \cdots, \tilde{t}_m\}) \geq {\rm RSS}(\{\tilde{t}_1, \cdots, \tilde{t}_m, t_1^0, t_2^0, \cdots, t_{j-1}^0, t_{j}^0 - \Delta_{j}, t_j^0 + \Delta_j, t_{j+1}^0, \cdots, t_{m_0}^0\}),
\end{equation}
let $T_s$ be the RSS in the time period $[t_{s-1}^0, t_{s}^0)$ for $s = 1,2, \cdots, j-1, j+2, \cdots, m_0 + 1,$ and $T_{j}, T_{j+1}, T_{m_0+2}$ be the RSS in the time period $[t_{j-1}^0, t_{j}^0 - \Delta_{j}), [t_j^0 + \Delta_j, t_{j+1}^0), [t_{j}^0 - \Delta_j, t_{j}^0 + \Delta_j),$ respectively. We focus on them separately.

(I) Consider $T_k, 1\leq k \leq m_0+1.$ For $s = 1, \cdots, j-1, j+2, \cdots, m_0 + 1,$ if the time period $[t_{s-1}^0, t_s^0)$ contains points say $\{\tilde{t}_l, \cdots, \tilde{t}_{l+k_s}\}\subset \{\tilde{t}_{1}, \cdots, \tilde{t}_m\}$ for some $l = 1,2,\cdots,m, k_s = -1,0,\cdots, m-l$ (here, $k_s = -1$ represents the case where $[t_{s-1}^0, t_s^0)$ contains no points among $\{\tilde{t}_1, \cdots,\tilde{t}_m\}$), then we have
\begin{equation}
    \begin{aligned}
         T_s
        &= {\rm RSS}(t_{s-1}^0 : \tilde{t}_l) + \sum_{h=l}^{l+k_s-1} {\rm RSS}(\tilde{t}_h:\tilde{t}_{h+1}) + {\rm RSS}(\tilde{t}_{l+k_s}: t_s^0)\\
    \end{aligned}\label{piecerss}
\end{equation}
By our RCRS procedure or \eqref{hat_tk}, we have $h_T=O(|\tilde{t}_i-\tilde{t}_j|)$ for any $i\neq j$.
Using Lemma \ref{lemma_RSS0break}, we have for $h=l,\cdots,l+k_s-1,$
\begin{equation}
    \begin{aligned}   \sum_{t=\tilde{t}_h}^{\tilde{t}_{h+1}-1} (u_{t+1}-\bar{u})^2 - Cr_F(r_F+d_T) \leq {\rm RSS}(\tilde{t}_h:\tilde{t}_{h+1})\leq \sum_{t=\tilde{t}_h}^{\tilde{t}_{h+1}-1} (u_{t+1}-\bar{u})^2
    \end{aligned}\label{withoutbreak}
\end{equation}
holds for some $C>0.$
If $|\tilde{t}_l - t_{s-1}^0|^{-1}h_T\rightarrow 0,$ we have
\begin{equation}
    {\rm RSS}(t_{s-1}^0 : \tilde{t}_l)\geq \sum_{t=t_{s-1}^0}^{\tilde{t}_l-1} (u_{t+1}-\bar{u}^{(t_{s-1}^0:\tilde{t}_l)})^2- Cr_F(r_F+d_T).\label{ht1}
\end{equation}
Otherwise, we take $0$ as a lower bound for ${\rm RSS}(t_{s-1}^0 : \tilde{t}_l).$ Then we have
\begin{equation}
    \begin{aligned}
        T_s &\geq \bigg[\sum_{t=t_{s-1}^0}^{\tilde{t}_l-1} (u_{t+1}-\bar{u}^{(t_{s-1}^0:\tilde{t}_l)})^2 + \sum_{h=l}^{l+k_s-1} \sum_{t=\tilde{t}_h}^{\tilde{t}_{h+1}-1} (u_{t+1}-\bar{u}^{(\tilde{t}_{h}:\tilde{t}_{h+1})})^2 + \sum_{t=t_{l+k_s}}^{t_s^0-1} (u_{t+1}-\bar{u}^{(\tilde{t}_{l+k_s}:t_s^0)})^2\bigg]\\ &~\quad - \bigg[\sum_{t=t_{s-1}^0}^{\tilde{t}_l-1} (u_{t+1}-\bar{u}^{(t_s^0:\tilde{t}_l)})^2\bigg]{\mathbf{1}}_{\{|\tilde{t}_l-t_{s-1}^0|=O(h_T)\}} -\bigg[\sum_{t=\tilde{t}_{l+k_s}}^{t_s^0-1} (u_{t+1}-\bar{u}^{(\tilde{t}_{l+k_s}:t_s^0)})^2\bigg]{\mathbf{1}}_{\{|t_{s}^0-\tilde{t}_{l+k_s}|=O(h_T)\}}\\ &~\quad- O_p((k_s+2)r_F(r_F+d_T))
        \\
        &=\bigg[
         \sum_{t=t_{s-1}^0}^{t_{s}^0-1}(u_{t+1}-\bar{u}^{(t_{s-1}^0:t_s^0)})^2 -O_p(k_s + 2)\bigg] -O_p((k_s+2)h_T)- O_p((k_s+2)r_F(r_F+d_T))\\
        &= \sum_{t=t_{s-1}^0}^{t_{s}^0-1}(u_{t+1}-\bar{u}^{(t_{s-1}^0:t_s^0)})^2 -O_p((k_s+2)h_T),~ s = 1,\cdots,j-1,j+2,\cdots,m_0+1.
    \end{aligned}
    \label{Ts}
\end{equation}
The first equality above follows by \eqref{upiece}.
Similarly, we have
\begin{equation}
    T_j \geq \sum_{t=t_{j-1}^0}^{t_j^0 -\Delta_j-1} (u_{t+1}-\bar{u}^{(t_{j-1}^0: t_j^0-\Delta_j)})^2 -O_p((k_j+2)h_T),\label{Tj}
\end{equation}
and
\begin{equation}
    T_{j+1} \geq \sum_{t=t_{j}^0+\Delta_j}^{t_{j+1}^0-1} (u_{t+1}-\bar{u}^{(t_{j-1}^0: t_j^0-\Delta_j)})^2-O_p((k_{j+1}+2)h_T).\label{Tjplus1}
\end{equation}
(II) Consider $T_{m_0+2}.$ By Assumption 3, we have $h_T =o(\min\limits_{1\leq i \leq m_0+1}|t_i^0-t_{i-1}^0|),$ which implies that
$r_F^2 (\min\limits_{1\leq i \leq m_0+1}|t_i^0-t_{i-1}^0|)^{\tau-{1\over 2}}+r_Fd_T(\min\limits_{1\leq i \leq m_0+1}|t_i^0-t_{i-1}^0|)^{-{1\over 2}}=o(1).$
Using Lemma \ref{lemma_RSS1break}, we have
\begin{equation}
    T_{m_0+2} \geq \sum_{t=t_j^0-\Delta_j}^{t_j^0 + \Delta_j-1}(u_{t+1}-\bar{u}^{(t_j^0-\Delta_j:t_j^0+\Delta_j)})^2 + \nu_j \Delta_j r_F^{-1} \label{Tm0plus2}
\end{equation}
holds in probability for some $\nu_j > 0.$
Combining \eqref{Ts} to \eqref{Tm0plus2} and \eqref{upiece} yields
\begin{equation}
    \begin{aligned}
          {\rm RSS}(\{\tilde{t}_1,\cdots,\tilde{t}_m\})
         = \sum_{t=1}^{T} (u_{t+1}-\bar{u})^2 + {1\over 2}\nu_j r_F^{-1} \min\limits_{1\leq i \leq m_0+1}|t_i^0 - t_{i-1}^0|
    \end{aligned}
\end{equation}
holds in probability as $m_0h_Tr_F=o(\min\limits_{1\leq i \leq m_0+1}|t_i^0-t_{i-1}^0|)$. Take $\nu = \min\limits_{1\leq j \leq m_0} \nu_j/16,$ we have \eqref{a147}
and the conclusion of Lemma \ref{lemmaA3} holds.

\end{proof}

\begin{lemma}\label{lem_udiff}
    Suppose $\{a_T\}$ be a integer sequence which diverges as $T\rightarrow \infty.$
    Under Assumption 1, we have
    \begin{equation}
       {\rm E} \bigg[\sum_{t=1}^{a_T} (u_t-\bar{u}^{(1:a_T)})^2 - \sum_{t=a_T+1}^{2a_T} (u_t-\bar{u}^{(a_T+1:2a_T)})^2\bigg]^2=O(a_T).
    \end{equation}
\end{lemma}
\begin{proof}
   Since
    \begin{equation}
        \sum_{t=1}^{a_T} (u_t-\bar{u}^{(1:a_T)})^2 -\sum_{t=a_T+1}^{2a_T} (u_t-\bar{u}^{(a_T+1:2a_T)})^2 = \sum_{t=1}^{a_T} (u_t^2 - u_{t+a_T}^2) - a_T[(\bar{u}^{(1:a_T)})^2-(\bar{u}^{(a_T+1:2a_T)})^2],
    \end{equation}
    it follows that
    $
    {\rm E}\bigg[\sum_{t=1}^{a_T} (u_t-\bar{u}^{(1:a_T)})^2 -\sum_{t=a_T+1}^{2a_T} (u_t-\bar{u}^{(a_T+1:2a_T)})^2\bigg] = 0.
    $
    By Assumption 1 (2), we have
    \begin{equation}
    \begin{aligned}
        &~\quad{\rm E} \bigg[\sum_{t=1}^{a_T} (u_t-\bar{u}^{(1:a_T)})^2 - \sum_{t=a_T+1}^{2a_T} (u_t-\bar{u}^{(a_T+1:2a_T)})^2\bigg]^2 \\ &= a_TO\bigg({\rm E}\bigg[{1\over \sqrt{a_T}}\sum_{t=1}^{a_T} (u_t^2 - u_{t+a_T}^2)\bigg]^2\bigg)+ O\bigg(a_T{\rm E}\bigg[\bigg({1\over a_T}\sum_{t=1}^{a_T} u_t\bigg)^2-\bigg({1\over a_T}\sum_{t=a_T+1}^{2a_T} u_t\bigg)^2\bigg]^2\bigg)\\
        &= O(a_T).
    \end{aligned}\label{lemA7}
    \end{equation}

    Then Lemma \eqref{lem_udiff} concludes by using similar argument to \eqref{eet2}.
\end{proof}

\begin{proof}[\textbf{Proof of Theorem 3.2.}]
   We first prove that: if $S^0 \nsubseteq \mathcal{U}(\widehat{S}_k, h_T)$  in the $k$-th step,  then
\begin{equation}
    P(\widehat{t}_{k+1} \in \mathcal{U}(S^0, h_T)) \longrightarrow 1,\label{thm2consistency}
\end{equation}
where $\widehat{S}_k = \{\widehat{t}_1 , \cdots, \widehat{t}_k\}$ and $\mathcal{U}(S,r)$ is the neighborhood  of  $S$ with radius $r$ defined in \eqref{Usr}.

By Assumption 3 that $h_T(\min\limits_{1\leq i \leq m_0+1}|t_{i}^0-t_{i-1}^0|)^{-1}$ tends to zero as $T \rightarrow \infty$, we see that for any $t$, the  period $(t-h_T, t+h_T)$  contains no more than one true break. If $S^0 \nsubseteq \mathcal{U}(\widehat{S}_k, h_T),$ we can find one $t_i^0\in S^0,$ such that $|\widehat{t}_l-t_i^0|\geq h_T$ holds for any $l=1,2,\cdots,k.$ When we select the $(k+1)-$th candidate break $\widehat{t}_{k+1},$ we aim to find $\mathop{\arg\max}_{t \in [\mathcal{U}(\widehat{S}_k, C_h h_T)]^c} {\rm RSS}(t-h_T:t+h_T)$ for $C_h$ in \eqref{hat_tk}. Thus, to show \eqref{thm2consistency}, it suffices to prove that for $t\notin \mathcal{U}(S^0, h_T),$
\begin{equation}
    P\{{\rm RSS}(t_i^0-h_T:t_i^0+h_T)>{\rm RSS}(t-h_T:t+h_T)\}\rightarrow 1. \label{a282}
\end{equation}

By Lemma \ref{lemma_RSS0break}, we have for 
any $t \notin \mathcal{U}(S^0, h_T),$
    \begin{equation}
        {\rm RSS}(t - h_T:t+h_T) \leq \sum_{s=t-h_T}^{t+h_T-1} (u_{s+1}-\bar{u}^{(t-h_T:t+h_T)})^2.\label{thm2withoutbreak}
    \end{equation}

On the other hand, by Lemma \ref{lemma_RSS1break}, it follows that for $i = 1,2,\cdots,m_0,$
    \begin{equation}
        {\rm RSS}(t_i^0-h_T : t_{i}^0 + h_T) \geq \sum_{s=t_{i}^0-h_T}^{t_{i}^0+h_T-1} (u_{s+1}-\bar{u}^{(t_i^0-h_T : t_{i}^0 + h_T)})^2+ \nu_i h_T r_F^{-1}\label{thm2withabreak}
    \end{equation}
holds in probability for some $\nu_i > 0.$
Combining \eqref{thm2withoutbreak} and \eqref{thm2withabreak} yields
\begin{equation}
\begin{aligned}
    &~\quad{\rm RSS}(t_i^0-h_T : t_{i}^0 + h_T) -  {\rm RSS}(t - h_T: t+h_T)\\ &\geq \sum_{s=t_i^0-h_T}^{t_i^0+h_T-1} \bigg[(u_{s+1}-\bar{u}^{(t_i^0-h_T:t_i^0+h_T)})^2 - (u_{s+(t-t_i^0)+1}-\bar{u}^{(t-h_T:t+h_T)})^2\bigg]  + \nu_i h_T r_F^{-1}\label{thm2-4}
\end{aligned}
\end{equation}
holds in probability for any $t \notin \mathcal{U}(S^0, h_T)$. Using Lemma \ref{lem_udiff}, we have
\begin{equation}
    \bigg|\sum_{s=t_i^0-h_T}^{t_i^0+h_T-1} \bigg[(u_{s+1}-\bar{u}^{(t_i^0-h_T:t_i^0+h_T)})^2 - (u_{s+(t-t_i^0)+1}-\bar{u}^{(t-h_T:t+h_T)})^2\bigg]\bigg|=O_p((2h_T)^{1\over 2}).\label{thm2-3}
\end{equation}
By\eqref{thm2-4}, \eqref{thm2-3} and   $r_F^2=o(h_T^{{1\over 2}-\tau})$ given in Assumption 3, it follows that
\begin{equation}
        {\rm RSS}(t_i^0-h_T : t_{i}^0 + h_T) -  {\rm RSS}(t- h_T: t+h_T) \geq{1\over 2} \nu_i h_T r_F^{-1} > 0 \label{A283}
\end{equation}
holds in probability for $t\notin \mathcal{U}(S^0,h_T)$.
 Thus, \eqref{a282} and then \eqref{thm2consistency} is proved.

    Now, suppose  the last true break $t_i^0$ is estimated in the $k_0$-th step, i.e.
    $S^0 \nsubseteq \mathcal{U}(\widehat{S}_{k_0-1}, h_T), S^0 \subset \mathcal{U}(\widehat{S}_{k_0}, h_T).$
    Similar to the proof of Lemma \ref{lemmaA3}, we have for some $\nu_1 >0,$
    \begin{equation}
        \begin{aligned}
        {\rm RSS}(\widehat{S}_{k_0 - 1}) \geq \sum_{t=1}^{T}(u_{t+1}-\bar{u})^2 + \nu_1 \min\limits_{1\leq i \leq m_0+1}|t_{i}^0 - t_{i-1}^0| r_F^{-1}.
        \end{aligned}\label{RSSk0-1}
    \end{equation}
    To establish the order of RSS ratios, we need to derive an upper bound for ${\rm RSS}(\widehat{S}_{k_0})$ by using similar arguments to \cite{chanGroupLASSOStructural2014}. Similar to the proof of Lemma \eqref{lemma_cov}, we have that for $1\leq s<l\leq T$ with $r_F(l-s)^{-{1\over 4}} + d_T(l-s)^{1\over 2}\rightarrow 0,$ it holds that
    \begin{equation}
        \bigg\|\sum_{t=s}^{l-1} (\tilde{\boldsymbol{\psi}}_{t} - \bar{\tilde{\boldsymbol{\psi}}}^{(s:l)})(\tilde{\boldsymbol{\psi}}_{t} - \bar{\tilde{\boldsymbol{\psi}}}^{(s:l)})'\bigg\| = O_p((l-s)^2 r_F).\label{fforder}
    \end{equation}
   Define $R_T (m_0) = \{\{t_1,\cdots,t_{m_0}\}: t_i \in \widehat{S}_{k_0}, |t_i - t_{i}^0| \leq h_T, i = 1,2,\cdots,m_0\}.$ For $\tilde{S} = \{\tilde{t}_1, \cdots, \tilde{t}_{m_0}\} \in R_T(m_0),$
    we have
    \begin{equation}
        \begin{aligned}
            {\rm RSS}(\tilde{S}) &\leq \sum_{t=1}^{t_1^0-h_T - 1}[y_{t+1} - \bar{y}^{(1:\tilde{t}_1)}- [\widehat{\boldsymbol{\beta}}^{(1: \tilde{t}_1)}]' (\tilde{\boldsymbol{\psi}}_{t} - \bar{\tilde{\boldsymbol{\psi}}}^{(1:\tilde{t}_1)})]^2 + \sum_{j=1}^{m_0-1} \sum_{t=t_j^0 + h_T}^{t_{j+1}^0-h_T-1} [y_{t+1} - \bar{y}^{(\tilde{t}_{j}:\tilde{t}_{j+1})}\\ &\quad- [\widehat{\boldsymbol{\beta}}^{(\tilde{t}_{j}:\tilde{t}_{j+1})}]' (\tilde{\boldsymbol{\psi}}_{t}-\bar{\tilde{\boldsymbol{\psi}}}^{(\tilde{t}_{j}:\tilde{t}_{j+1})})]^2
            + \sum_{t=t_{m_0}^0+h_T}^{T}[y_{t+1} - \bar{y}^{(\tilde{t}_{m_0}:T+1)} - [\widehat{\boldsymbol{\beta}}^{(\tilde{t}_{m_0}:T+1)}]' (\tilde{\boldsymbol{\psi}}_{t} \\ &~\quad -  \bar{\tilde{\boldsymbol{\psi}}}^{(\tilde{t}_{m_0}:T+1)})]^2 + \sum_{j=1}^{m_0} \sum_{t = t_j^0 - h_T}^{t_j^0-1} [y_{t+1} - \bar{y}^{(\tilde{t}_{j}:\tilde{t}_{j+1})} - [\widehat{\boldsymbol{\beta}}^{(\tilde{t}_{j}:\tilde{t}_{j+1})}]' (\tilde{\boldsymbol{\psi}}_{t} - \bar{\tilde{\boldsymbol{\psi}}}^{(\tilde{t}_{j}:\tilde{t}_{j+1})})]^2 \\
            &\quad~+ \sum_{j=1}^{m_0} \sum_{t = t_j^0}^{t_j^0+h_T-1} [y_{t+1} - \bar{y}^{(\tilde{t}_{j}:\tilde{t}_{j+1})} - [\widehat{\boldsymbol{\beta}}^{(\tilde{t}_{j}:\tilde{t}_{j+1})}]
        ' (\tilde{\boldsymbol{\psi}}_{t} - \bar{\tilde{\boldsymbol{\psi}}}^{(\tilde{t}_{j}:\tilde{t}_{j+1})})]^2 \\
            &:= (I) + (II) + (III) + (IV) + (V)+ O_p(m_0h_T^2r_F).
        \end{aligned}\label{L1toL5}
    \end{equation}
    where $\widehat{\boldsymbol{\beta}}^{(1: \tilde{t}_1)}, \cdots, \widehat{\boldsymbol{\beta}}^{(\tilde{t}_{m_0}:T+1)}$ are the least square estimators.
    Similar to \cite{chanGroupLASSOStructural2014},
    \begin{equation}
        \begin{aligned}
            (I) + (II) + (III) &\leq \sum_{t=1}^{t_1^0-h_T - 1} (u_{t+1} - \bar{u}^{(1:\tilde{t}_1)})^2 + \sum_{j=1}^{m_0-1} \sum_{t=t_j^0 + h_T}^{t_{j+1}^0-h_T-1} (u_{t+1} - \bar{u}^{(\tilde{t}_{j}:\tilde{t}_{j+1})})^2 + \\
            &\qquad \sum_{t=t_{m_0}^0+h_T}^{T} (u_{t+1} - \bar{u}^{(\tilde{t}_{m_0}:T+1)})^2
        \end{aligned}\label{L1L2L3}
    \end{equation}
    holds in probability. For (IV), we have
    \begin{equation}
        \begin{aligned}
            &\quad~\sum_{t = t_j^0 - h_T}^{t_j^0-1} [y_{t+1} - \bar{y}^{(\tilde{t}_{j}:\tilde{t}_{j+1})} - [\widehat{\boldsymbol{\beta}}^{(\tilde{t}_{j}:\tilde{t}_{j+1})}]' (\tilde{\boldsymbol{\psi}}_{t} - \bar{\tilde{\boldsymbol{\psi}}}^{(\tilde{t}_{j}:\tilde{t}_{j+1})})]^2 \\
            &= \sum_{t = t_j^0 - h_T}^{t_j^0-1} [u_{t+1} - \bar{u}^{(\tilde{t}_{j}:\tilde{t}_{j+1})} - [\widehat{\boldsymbol{\beta}}^{(\tilde{t}_{j}:\tilde{t}_{j+1})} - \boldsymbol{\beta}^{(t_{j-1}^0 : t_j^0)}]' (\tilde{\boldsymbol{\psi}}_{t} - \bar{\tilde{\boldsymbol{\psi}}}^{(\tilde{t}_{j}:\tilde{t}_{j+1})})]^2\\
         &= \sum_{t = t_j^0 - h_T}^{t_j^0-1} [u_{t+1} - \bar{u}^{(\tilde{t}_{j}:\tilde{t}_{j+1})}]^2 + O_p\bigg(\lambda_{\max}\bigg( \sum_{t = t_j^0 - h_T}^{t_j^0-1} (\tilde{\boldsymbol{\psi}}_{t} - \bar{\tilde{\boldsymbol{\psi}}}^{(\tilde{t}_{j}:\tilde{t}_{j+1})})(\tilde{\boldsymbol{\psi}}_{t} - \bar{\tilde{\boldsymbol{\psi}}}^{(\tilde{t}_{j}:\tilde{t}_{j+1})})'
            \bigg)\\ &~\quad\bigg\|\widehat{\boldsymbol{\beta}}^{(\tilde{t}_{j}:\tilde{t}_{j+1})} - \boldsymbol{\beta}^{(t_{j-1}^0 : t_j^0)}\bigg\|^2\bigg)\\
            &= \sum_{t = t_j^0 - h_T}^{t_j^0-1} [u_{t+1} - \bar{u}^{(\tilde{t}_{j}:\tilde{t}_{j+1})}]^2 + O_p(r_Fh_T^2),
        \end{aligned}
    \end{equation}
    where the last step follows by \eqref{fforder}. This implies that there exists $A_4 > 0$ such that
    \begin{equation}
        (IV) \leq \sum_{j=1}^{m_0} \sum_{t = t_j^0 - h_T}^{t_j^0-1} (u_{t+1} - \bar{u}^{(t_j^0-h_T:t_j^0)})^2  + A_4 m_0 r_F h_T^2 \label{L4}
    \end{equation}
    holds in probability. Similarly, we can show that there exists $A_5 > 0$ such that
    \begin{equation}
        (V) \leq \sum_{j=1}^{m_0} \sum_{t = t_j^0}^{t_j^0+h_T} (u_{t+1} - \bar{u}^{(t_j^0:t_j^0 + h_T)})^2 + A_5 m_0 r_F h_T^2 \label{L5}
    \end{equation}
    holds in probability. Combining \eqref{L1L2L3}, \eqref{L4}, \eqref{L5} 
    yields
    \begin{equation}
        {\rm RSS}(\tilde{S}) \leq \sum_{t=1}^{T}(u_{t+1}-\bar{u})^2 + A_{0,0} m_0 r_F h_T^2 \label{RSSofneigh}
    \end{equation}
    holds in probability for some $A_{0,0} > 0$. Hence,
    \begin{equation}
        {\rm RSS}(\widehat{S}_{k_0}) \leq \sum_{t=1}^{T}(u_{t+1}-\bar{u})^2 + A_{0,0} m_0 r_F h_T^2. \label{RSSk_0upper}
    \end{equation}
    By \eqref{RSSk0-1}, \eqref{RSSk_0upper} and ${m_0r_F^2 h_T^2\over \min\limits_{1\leq i \leq m_0+1}|t_i^0-t_{i-1}^0|}\rightarrow 0$ in Assumption 3,  we have
    \beqn
        {\rm RSS}(\widehat{S}_{k_0-1})-{\rm RSS}(\widehat{S}_{k_0}) &\geq& \nu_1 \min\limits_{1\leq i \leq m_0+1}|t_i^0-t_{i-1}^0|r_F^{-1} - A_{0,0} m_0 r_F h_T^2\nn\\
         &\geq& C\min\limits_{1\leq i \leq m_0+1}|t_i^0-t_{i-1}^0|r_F^{-1} \label{a178}
    \eeqn
    holds in probability for some  $C>0.$

    Next, we consider the lower bound of ${\rm RSS}(\widehat{S}_{k_0}).$ Note that
    \begin{equation}
        {\rm RSS}(\widehat{S}_{k_0}) \geq {\rm RSS}(\widehat{S}_{k_0} \cup S^0).
    \end{equation}
  Since  $k_0 \leq M_T= O(\log T), \, m_0 = O(\log T)$ by Assumption~2 (1), using similar arguments to  \eqref{piecerss} to \eqref{Ts},  we can show that
    \begin{equation}
        \begin{aligned}
            {\rm RSS}(\widehat{S}_{k_0} \cup S^0)
            &\geq \sum_{t=1}^{T}(u_{t+1}-\bar{u})^2 - A_1 (m_0+k_0) h_T
            \\ &\geq \sum_{t=1}^{T}(u_{t+1}-\bar{u})^2 - A_1'  h_T \log T\\
            &\geq \sum_{t=1}^{T}(u_{t+1}-\bar{u})^2 - A_1'  h_T^2 r_F\log T
        \end{aligned}
    \end{equation}
    holds in probability for some $A_1, A_1
    ' > 0$.  Thus,
    \begin{equation}
        {\rm RSS}(\widehat{S}_{k_0}) \geq \sum_{t=1}^{T}(u_{t+1}-\bar{u})^2 - A_1'  h_T^2 r_F\log T \label{RSSk_0lower}
    \end{equation}
    holds in probability. Similar to the arguments in \eqref{RSSk_0upper} and \eqref{RSSk_0lower}, we have for all $i \geq 0,$ there exists $A_{i,i} > 0,$ such that
    \begin{equation}
        {\rm RSS}(\widehat{S}_{k_0+i})\geq \sum_{t=1}^{T}(u_{t+1}-\bar{u})^2 - A_{i,i} h_T^2 r_F\log T.\label{upp1}
    \end{equation}
    Meanwhile, by \eqref{RSSk_0upper}, we have for all $i\geq 0,$
    \begin{equation}
           {\rm RSS}(\widehat{S}_{k_0+i})
            \leq {\rm RSS}(\widehat{S}_{k_0})
            \leq \sum_{t=1}^{T}(u_{t+1}-\bar{u})^2 + A_{0,0} m_0 r_F h_T^2
         \label{upp2}
    \end{equation}
    holds in probability.

    Combining \eqref{upp1} and \eqref{upp2} yields
    \begin{equation}
        \begin{aligned}
        \sum_{t=1}^{T}(u_{t+1}-\bar{u})^2 - A_{i,i} h_T^2 r_F\log T \leq {\rm RSS}(\widehat{S}_{k_0+i}) \leq \sum_{t=1}^{T}(u_{t+1}-\bar{u})^2 + A_{0,0} h_T^2 r_F\log T,
        \end{aligned}\label{RSSk0plus1}
    \end{equation}
    \begin{equation}
        \begin{aligned}
        \sum_{t=1}^{T}(u_{t+1}-\bar{u})^2 - A_{i+1,i+1} h_T^2 r_F\log T \leq {\rm RSS}(\widehat{S}_{k_0+i+1}) \leq \sum_{t=1}^{T}(u_{t+1}-\bar{u})^2 + A_{0,0} h_T^2 r_F\log T,
        \end{aligned}\label{RSSk0plus11}
    \end{equation}
    holds in probability for all $i \geq 0$. Combining \eqref{RSSk0plus1} and \eqref{RSSk0plus11} yields
    \begin{equation}
    \begin{aligned}
    (A_{0,0} + A_{i,i}) h_T^2 r_F \log T
        &\leq {\rm RSS}(\widehat{S}_{k_0+i})-{\rm RSS}(\widehat{S}_{k_0+i+1}) \\ &\leq (A_{0,0} + A_{i+1,i+1}) h_T^2 r_F \log T,
        \end{aligned}
    \end{equation}
    holds in probability, or equivalently
    \begin{equation}
        {\rm RSS}(\widehat{S}_{k_0+i})-{\rm RSS}(\widehat{S}_{k_0+i+1})\asymp h_T^2r_F\log T.
        \label{deltaRSS_afterk0}
    \end{equation}
    Combining \eqref{a178} and \eqref{deltaRSS_afterk0} yields
    \begin{equation}
        \begin{aligned}
        {\rm RSSR}_{k_0+1} &= \dfrac{{\rm RSS}(\widehat{S}_{k_0})-{\rm RSS}(\widehat{S}_{k_0+1})}{{\rm RSS}(\widehat{S}_{k_0-1})-{\rm RSS}(\widehat{S}_{k_0})}\\
        &= O_p\bigg({r_F h_T^2 \log T\over \min\limits_{1\leq i \leq m_0+1}|t_i^0-t_{i-1}^0| r_F^{-1}}\bigg)
        \stackrel{P}{\longrightarrow} 0.
        \end{aligned}\label{rssrk0order}
    \end{equation}
    Using \eqref{deltaRSS_afterk0}, we have for any $i \geq 1,$
    \begin{equation}
        {\rm RSSR}_{k_0+i+1} = \dfrac{{\rm RSS}(\widehat{S}_{k_0+i+1})-{\rm RSS}(\widehat{S}_{k_0+i+2})}{{\rm RSS}(\widehat{S}_{k_0+i})-{\rm RSS}(\widehat{S}_{k_0+i+1})}\asymp 1
        \label{deltaRSSk0}
    \end{equation}
    holds in probability.
    This ends the proof of Theorem 3.2.
\end{proof}
\subsection{The Proof of Theorem 3.3}
\begin{proof}[\textbf{Proof of Theorem 3.3.}]

We divide the proof into two parts: (i) $P(\widehat{\widehat{m}}< m_0) \rightarrow 0;$ (ii) $P(\widehat{\widehat{m}}> m_0) \rightarrow 0.$

We firstly prove that $P(\widehat{\widehat{m}}< m_0) \rightarrow 0.$ Theorem 3.2 implies that with probability tending to 1, there exists $\widehat{S}_T = \{\widehat{t}_{T1}, \cdots, \widehat{t}_{T,m_0}\}\subset \widehat{S}$ such that $|\widehat{t}_{Ti} - t_i^0| \leq h_T$ with probability tending to 1. It's enough to show that for all $S$ satisfying $|S| = m<m_0,$
\begin{equation}
    P({\rm IC}_1(S) \geq {\rm IC}_1(\widehat{S}_T)) \rightarrow 1.
\end{equation}
We've proved in \eqref{RSSofneigh} that
\begin{equation}
    {\rm RSS}(\widehat{S}_T) \leq \sum_{t=1}^{T}(u_{t+1}-\bar{u})^2 + A_{0,0} m_0 r_F h_T^2 \label{RSSofneigh1}
\end{equation}
By Lemma \ref{lemmaA3}, we have if $m<m_0,$ then with probability tending to 1,
\begin{equation}
    {\rm RSS}(\{t_1, \cdots, t_m\}) \geq \sum_{t=1}^{T}(u_{t+1}-\bar{u})^2 + \nu \min\limits_{1\leq i \leq m_0} |t_i^0 - t_{i-1}^0| r_F^{-1},
\end{equation}
where $\nu$ is some positive constant. Then, when $\widehat{\widehat{m}} < m_0,$ with probability tending to 1,
\begin{equation}
    \begin{aligned}
        {\rm IC}_1(\widehat{\widehat{S}}) &= {\rm RSS}(\widehat{\widehat{S}}) + \widehat{\widehat{m}}\omega_{1T}\\
        &\geq \sum_{t=1}^{T}(u_{t+1}-\bar{u})^2 + \nu \min\limits_{1\leq i \leq m_0} |t_i^0 - t_{i-1}^0|r_F^{-1} + \widehat{\widehat{m}}\omega_{1T}\\
        &\geq {\rm RSS}(\widehat{S}_T) + m_0 \omega_{1T} + \nu \min\limits_{1\leq i \leq m_0} |t_i^0 - t_{i-1}^0|r_F^{-1} - (m_0 - \widehat{\widehat{m}})\omega_{1T}- A_{0,0}m_0r_Fh_T^2\\
        &\geq {\rm RSS}(\widehat{S}_T) + m_0 \omega_{1T}\\
        &= {\rm IC}_1(\widehat{S}_T),
    \end{aligned}
\end{equation}
where the last inequality is derived by Assumption 4 (1) that $\dfrac{m_0 \omega_{1T} r_F}{\min\limits_{1\leq i\leq m_0+1}|t_{i}^0 - t_{i-1}^0|} \rightarrow 0$ and $\lim\limits_{T\rightarrow \infty} m_0r_Fh_T^2/\omega_{1T} \leq 1.$ This proves $P(\widehat{\widehat{m}}<m_0)\rightarrow 0.$

Next, we turn to show that $P(\widehat{\widehat{m}}>m_0)\rightarrow 0.$ It's enough to prove that if $\widehat{\widehat{m}}>m_0,$
\begin{equation}
    {\rm IC}_1(\{\widehat{\widehat{t}}_1, \cdots, \widehat{\widehat{t}}_{\widehat{\widehat{m}}}\}) > {\rm IC}_1(\{\widehat{\widehat{t}}_1, \cdots, \widehat{\widehat{t}}_{m_0}\})
\end{equation}
with probability tending to 1, where $\{\widehat{\widehat{t}}_1, \cdots, \widehat{\widehat{t}}_{m_0}\} = \arg\min\limits_{(t_1,\cdots,t_{m_0})\in R_T(m_0)}{\rm RSS}(\{t_1,\cdots,t_{m_0}\})$, and $R_T(m_0)$ is defined in the proof of Theorem 3.2.
In fact, when $\widehat{\widehat{m}} > m_0,$ we have
\begin{equation}
    \begin{aligned}
        {\rm RSS}(\{\widehat{t}_{T1},\cdots,\widehat{t}_{T,m_0}\})
        &\geq {\rm RSS}(\{\widehat{\widehat{t}}_1, \cdots, \widehat{\widehat{t}}_{m_0}\})\\
        &\geq {\rm RSS}(\{\widehat{\widehat{t}}_1, \cdots, \widehat{\widehat{t}}_{\widehat{\widehat{m}}}\})\\
        &\geq {\rm RSS}(\{\widehat{\widehat{t}}_1,\cdots,\widehat{\widehat{t}}_{\widehat{\widehat{m}}}\}\cup S^0).\label{a195}
    \end{aligned}
\end{equation}
By similar argument in \eqref{RSSk_0lower} or \eqref{RSSk0plus1}, we have for some positive constant $C,$
\begin{equation}
    {\rm RSS}(\{\widehat{\widehat{t}}_1,\cdots,\widehat{\widehat{t}}_{\widehat{\widehat{m}}}\}\cup S^0)
    \geq \sum_{t=1}^{T}(u_{t+1}-\bar{u})^2 - Cm_0r_F h_T^2 \label{A320}
\end{equation}
holds with probability tending to 1. Combining \eqref{RSSofneigh}, \eqref{a195} and \eqref{A320} yields
\begin{equation}
    \begin{aligned}
        &\quad {\rm RSS}(\{\widehat{\widehat{t}}_1, \cdots, \widehat{\widehat{t}}_{m_0}\})-{\rm RSS}(\{\widehat{\widehat{t}}_1, \cdots, \widehat{\widehat{t}}_{\widehat{\widehat{m}}}\})\\
        &\leq {\rm RSS}(\{\widehat{t}_{T1},\cdots,\widehat{t}_{T,m_0}\}) - {\rm RSS}(\{\widehat{\widehat{t}}_1,\cdots,\widehat{\widehat{t}}_{\widehat{\widehat{m}}}\}\cup S^0) \\
        &=O_p(m_0 r_F h_T^2) ,
    \end{aligned}
\end{equation}
and then
\begin{equation}
    \begin{aligned}
    &\quad {\rm IC}_1(\{\widehat{\widehat{t}}_1, \cdots, \widehat{\widehat{t}}_{\widehat{\widehat{m}}}\}) - {\rm IC}_1(\{\widehat{\widehat{t}}_1, \cdots, \widehat{\widehat{t}}_{m_0}\})\\
    &\geq (\widehat{\widehat{m}} - m_0)\omega_{1T}- A_{0,0} m_0 r_F h_T^2 > 0
    \end{aligned}
\end{equation}
holds in probability, where the last inequality is derived by $\dfrac{m_0r_F h_T^2}{\omega_{1T}} \rightarrow 0$.  
This yields $P(\widehat{\widehat{m}}>m_0) \rightarrow 0$ and complete the proof of $P(\widehat{\widehat{m}}  = m_0)\rightarrow 1$.

Next, we will show that there exists one consistent estimator $\widehat{\widehat{t}}_i\in \widehat{\widehat{S}}$ for each $t_i^0 \in S^0$ by contradictory.
Suppose that we cannot find consistent estimators in $\widehat{\widehat{S}}$ for one $t_i^0 \in S^0,$ then by Assumption 3: $h_T = o(\min\limits_{1\leq i \leq m_0+1}|t_i^0-t_{i-1}^0|),$  there exists $C > 0$ such that $|\widehat{\widehat{t}}_j - t_i^0|\geq C \min\limits_{1\leq i \leq m_0+1}|t_i^0-t_{i-1}^0|$ holds in probability for all $\widehat{\widehat{t}}_j\in \widehat{\widehat{S}}$. Similar to the proof of Lemma \ref{lemmaA3}, we can show that
\begin{equation}
    {\rm RSS}(\widehat{\widehat{S}}) \geq  \sum_{t=1}^{T} (u_{t+1}-\bar{u})^2 + C' \min\limits_{1\leq i \leq m_0+1} |t_i^0-t_{i-1}^0| \label{contraRSS1}
\end{equation}
holds in probability for some $C' > 0.$ By the conclusion of Theorem 3.2, we can find  $\widehat{t}_i \in \widehat{S}$ satisfying $|\widehat{t}_i- t_i^0| = O_p(h_T).$ By \eqref{RSSofneigh} or \eqref{RSSofneigh1}, one can show that
\begin{equation}
    {\rm RSS}(\widehat{\widehat{S}}\cup \{\widehat{t}_i\})\leq  \sum_{t=1}^{T} (u_{t+1}-\bar{u})^2 + A_{0,0} m_0 r_F h_T\label{contraRSS2}
\end{equation}
holds in probability. Combining \eqref{contraRSS1} and \eqref{contraRSS2} yields
\begin{equation}
    {\rm IC}_1 (\widehat{\widehat{S}}\cup \{\widehat{t}_i\})-{\rm IC}_1(\widehat{\widehat{S}})\leq \omega_{1T} + A_{0,0} m_0 r_F h_T  -C' \min\limits_{1\leq i \leq m_0+1} |t_i^0-t_{i-1}^0| \leq 0
\end{equation}
holds in probability, by the Assumption 3 that $m_0 r_F h_T (\min\limits_{1\leq i \leq m_0+1} |t_i^0-t_{i-1}^0|)^{-1}\rightarrow 0$ and the Assumption 4 (1) that $\omega_{1T} (\min\limits_{1\leq i \leq m_0+1} |t_i^0-t_{i-1}^0|)^{-1}\rightarrow 0$. This leads to contradictory and
ends the proof of Theorem 3.3.
\end{proof}

\subsection{The Proof of Theorem 3.4}

Before proving Theorem 3.4, we first give two required lemmas.
The first one (Lemma \ref{lemma_m}) establish the lower bound for ${\rm RSS}_i(\widehat{\widehat{G}}_i)$ with $(\widehat{\widehat{G}}_i)^c\cap G_i^0\neq \phi.$ The second one (Lemma~\ref{lemma_r}) is about the  lower bound of  ${\rm RSS}_i(\widehat{\widehat{G}}_i)$ for the case where $G^0_i \subsetneq \widehat{\widehat{G}}_i$ and $|G^0_i|<|\widehat{\widehat{G}}_i|.$

\begin{lemma}
    Suppose  $(\widehat{\widehat{G}}_i)^c\cap G_i^0\neq \phi$ and Assumptions 1-3 hold, then
    \begin{equation}
        {\rm RSS}_i(\widehat{\widehat{G}}_i)
        \geq 
        \sum_{t=\widehat{\widehat{t}}_{i-1}}^{\widehat{\widehat{t}}_{i}-1}(u_{t+1}-\bar{u}^{(\widehat{\widehat{t}}_{i-1}:\widehat{\widehat{t}}_{i})})^2 + C \min\limits_{1\leq i \leq m_0+1}|t_i^0-t_{i-1}^0|\label{A.7}
    \end{equation}
    holds in probability for some $C > 0.$
    \label{lemma_m}
\end{lemma}

\begin{proof}
 Since $(\widehat{\widehat{G}}_i)^c\cap G_i^0\neq \phi$,  there must exist $m\in (\widehat{\widehat{G}}_i)^c\cap G_i^0.$ 
 We will only consider the RSS under the case where only one predictor $x_{mt}$ is incorrectly eliminated. Specifically, by the fact that
    \begin{equation}
        {\rm RSS}_i(\widehat{\widehat{G}}_i) \geq {\rm RSS}_i\bigg( \widehat{\widehat{G}}_i \cup \bigg\{\{(\widehat{\widehat{G}}_i)^c\cap G_i^0\} \cap \{x_{mt}\}^c\bigg\}\bigg),\label{RSSupp}
    \end{equation}
    it suffices to give a lower bound for ${\rm RSS}_i\bigg(\widehat{\widehat{G}}_i \cup \bigg\{\{(\widehat{\widehat{G}}_i)^c\cap G_i^0\} \cap \{x_{mt}\}^c\bigg\}\bigg)$.

Let $\boldsymbol{x}_t^m$ be the  vector with predictors in $\widehat{\widehat{G}}_i \cup \bigg\{\{(\widehat{\widehat{G}}_i)^c\cap G_i^0\} \cap \{x_{mt}\}^c\bigg\}.$ Similar to \eqref{cointA}, there exists a rotated vector $\boldsymbol{\psi}_{t}^m = ((\boldsymbol{z}_t^m)', (\boldsymbol{\psi}_{1t}^m)',(\boldsymbol{\psi}_{2t}^m)')'$ such that
\begin{equation}
    \begin{aligned}
        &~\quad{\rm RSS}_i\bigg(\widehat{\widehat{G}}_i \cup \bigg\{\{(\widehat{\widehat{G}}_i)^c\cap G_i^0\} \cap \{x_{mt}\}^c\bigg\}\bigg)\\
        &= \sum_{t=\widehat{\widehat{t}}_{i}-1}^{\widehat{\widehat{t}}_{i}-1} \bigg([y_{t+1}-\bar{y}^{(\widehat{\widehat{t}}_{i-1}:\widehat{\widehat{t}}_{i})}] - [\widehat{\boldsymbol{\beta}}^{\tilde{m}}]' [\boldsymbol{\psi}_{t}^{m} - \bar{\boldsymbol{\psi}}^{(\widehat{\widehat{t}}_{i-1}:\widehat{\widehat{t}}_{i})}_m]\bigg)^2\\
        &> \sum_{t=t_{i-1}^0+h_T}^{t_{i}^0-h_T-1} \bigg([y_{t+1}-\bar{y}^{(\widehat{\widehat{t}}_{i-1}:\widehat{\widehat{t}}_{i})}] - [\widehat{\boldsymbol{\beta}}^{\tilde{m}}]' [\boldsymbol{\psi}_{t}^{m} - \bar{\boldsymbol{\psi}}^{(\widehat{\widehat{t}}_{i-1}:\widehat{\widehat{t}}_{i})}_m]\bigg)^2,\label{RSSQm1}
    \end{aligned}
\end{equation}
where
$\widehat{\boldsymbol{\beta}}^{\tilde{m}}$ denotes the LSE based on samples $(y_{t+1},\boldsymbol{\psi}_{t}^m)$ over the period $(\widehat{\widehat{t}}_{i-1}:\widehat{\widehat{t}}_{i}).$

Along the proof line of  \eqref{A87}, we can show that for $t_{i-1}^0+h_T\leq t<t_{i}^0-h_T,$
\begin{equation}
    \begin{aligned}
        &\quad~[y_{t+1}-\bar{y}^{(\widehat{\widehat{t}}_{i-1}:\widehat{\widehat{t}}_{i})}] - [\widehat{\boldsymbol{\beta}}^{\tilde{m}}]' [\boldsymbol{\psi}_{t}^{m} - \bar{\boldsymbol{\psi}}^{(\widehat{\widehat{t}}_{i-1}:\widehat{\widehat{t}}_{i})}_m]\\
        &= [y_{t+1}-\bar{y}^{(t_{i-1}^0+h_T:t_{i}^0-h_T)}] - [\widehat{\boldsymbol{\beta}}^{\tilde{m}}]' [\boldsymbol{\psi}_{t}^m-\bar{\boldsymbol{\psi}}^{(t_{i-1}^0+h_T:t_{i}^0-h_T)}_m] +
        \\ &\quad~\bigg\{[\bar{y}^{(t_{i-1}^0+h_T:t_{i}^0-h_T)} - \bar{y}^{(\widehat{\widehat{t}}_{i-1}:\widehat{\widehat{t}}_{i})}] - [\widehat{\boldsymbol{\beta}}^{\tilde{m}}]' [\bar{\boldsymbol{\psi}}^{(t_{i-1}^0+h_T:t_{i}^0-h_T)}_m - \bar{\boldsymbol{\psi}}^{(\widehat{\widehat{t}}_{i-1}:\widehat{\widehat{t}}_{i})}_m] \bigg\}       \\
        &:= [y_{t+1}-\bar{y}^{(t_{i-1}^0+h_T:t_{i}^0-h_T)}] - [\widehat{\boldsymbol{\beta}}^{\tilde{m}}]' [\boldsymbol{\psi}_{t}^m-\bar{\boldsymbol{\psi}}^{(t_{i-1}^0+h_T:t_{i}^0-h_T)}_m] + E_i.
    \end{aligned}\label{res_m}
\end{equation}
Note that $\sum_{t=t_{i-1}^0+h_T}^{t_{i}^0-h_T-1} \bigg([y_{t+1}-\bar{y}^{(t_{i-1}^0+h_T:t_{i}^0-h_T)}] - [\widehat{\boldsymbol{\beta}}^{\tilde{m}}]' [\boldsymbol{\psi}_{t}^m-\bar{\boldsymbol{\psi}}^{(t_{i-1}^0+h_T:t_{i}^0-h_T)}_m] \bigg) = 0,$ by similar argument to \eqref{RSSleftpart}, we have
\begin{equation}
    \begin{aligned}
    &\quad~{\rm RSS}_i(\widehat{\widehat{G}}_i \cup \{\{(\widehat{\widehat{G}}_i)^c\cap G_i^0\} \cap \{x_{mt}\}^c\}) \\
    &\geq \sum_{t=t_{i-1}^0+h_T}^{t_{i}^0-h_T-1} \bigg([y_{t+1}-\bar{y}^{(t_{i-1}^0+h_T:t_{i}^0-h_T)}] - [\widehat{\boldsymbol{\beta}}^{\tilde{m}}]' [\boldsymbol{\psi}_{t}^m-\bar{\boldsymbol{\psi}}^{(t_{i-1}^0+h_T:t_{i}^0-h_T)}_m] + E_i\bigg)^2\\
    &> \sum_{t=t_{i-1}^0+h_T}^{t_{i}^0-h_T-1} \bigg([y_{t+1}-\bar{y}^{(t_{i-1}^0+h_T:t_{i}^0-h_T)}] - [\widehat{\boldsymbol{\beta}}^{\tilde{m}}]' [\boldsymbol{\psi}_{t}^m-\bar{\boldsymbol{\psi}}^{(t_{i-1}^0+h_T:t_{i}^0-h_T)}_m]\bigg)^2\\
    &\geq \sum_{t=t_{i-1}^0+h_T}^{t_{i}^0-h_T-1} \bigg([y_{t+1}-\bar{y}^{(t_{i-1}^0+h_T:t_{i}^0-h_T)}] - [\widehat{\boldsymbol{\beta}}^{m}]' [\boldsymbol{\psi}_{t}^m-\bar{\boldsymbol{\psi}}^{(t_{i-1}^0+h_T:t_{i}^0-h_T)}_m]\bigg)^2,
\end{aligned}\label{RSSsim2}
\end{equation}
where
$\widehat{\boldsymbol{\beta}}^{m}$ denotes the LSE using samples $(y_{t+1},\boldsymbol{\psi}_{t}^m)$ in the period $(t_{i-1}^0+h_T:t_{i}^0-h_T).$
Denote $\boldsymbol{\beta}^m, \beta_x^m$ be the true coefficient of $\boldsymbol{\psi}_{t}^m$ and $x_{mt}.$
Note that
\begin{equation}\begin{aligned}
\widehat{\boldsymbol{\beta}}^{m} &= \boldsymbol{\beta}^m
+ \bigg(\sum_{t=t_{i-1}^0+h_T}^{t_{i}^0-h_T-1} [\boldsymbol{\psi}_{t}^m-\bar{\boldsymbol{\psi}}^{(t_{i-1}^0+h_T:t_{i}^0-h_T)}_m] [\boldsymbol{\psi}_{t}^m-\bar{\boldsymbol{\psi}}^{(t_{i-1}^0+h_T:t_{i}^0-h_T)}_m]'\bigg)^{-1} \\
&\quad~\bigg\{\bigg(\sum_{t=t_{i-1}^0+h_T}^{t_{i}^0-h_T-1} [\boldsymbol{\psi}_{t}^{m}-\bar{\boldsymbol{\psi}}^{(t_{i-1}^0+h_T:t_{i}^0-h_T)}_m] [x_t^m - \bar{x}^{(t_{i-1}^0+h_T:t_{i}^0-h_T)}_m] \beta^m_x\bigg)+ 
\\ &\quad\bigg(
\sum_{t=t_{i-1}^0+h_T}^{t_{i}^0-h_T-1} [\boldsymbol{\psi}_{t}^m-\bar{\boldsymbol{\psi}}^{(t_{i-1}^0+h_T:t_{i}^0-h_T)}_m] [u_{t+1} -  \bar{u}^{(t_{i-1}^0+h_T:t_{i}^0-h_T)}]\bigg)\bigg\}\\
&= \boldsymbol{\beta}^m + \boldsymbol{P}_m \bigg\{\boldsymbol{P}_m\widehat{\boldsymbol{\Omega}}_{\psi_m}^{(t_{i-1}^0+h_T:t_{i}^0-h_T)}\boldsymbol{P}_m\bigg\}^{-1} \bigg\{\boldsymbol{P}_m [\widehat{\boldsymbol{\Omega}}_{\psi_m, x_m}^{(t_{i-1}^0+h_T:t_{i}^0-h_T)} \beta_{x}^m+\widehat{\boldsymbol{\Omega}}_{\psi_m, u}^{(t_{i-1}^0+h_T:t_{i}^0-h_T)}]\bigg\},
\end{aligned}\label{betasim}
\end{equation}
where $\boldsymbol{P}_m$ is the corresponding scaled matrix.
By similar arguments to the proofs of Lemmas \ref{lemma_cov} and \ref{lemma_ux}, we have
\begin{equation}
\begin{aligned}
    \|\boldsymbol{P}_m^{-1} (\widehat{\boldsymbol{\beta}}^m-\boldsymbol{\beta}^m)\|
    &= O_p(r_F) \times O_p(\Delta_i^{-{1\over 2}}r_F^{1\over 2} + d_T^{1\over 2})
    = O_p(\Delta_i^{-{1\over 2}}r_F(r_F^{1\over 2}+d_T^{1\over 2})),\label{betamissorder_z}
\end{aligned}
\end{equation} where $\Delta_i=t_{i}^0-t_{i-1}^0-2h_T.$
We rewrite \eqref{RSSsim2} by
\begin{equation}
    \begin{aligned}
        &~\quad{\rm RSS}_i(\widehat{\widehat{G}}_i \cup \{\{(\widehat{\widehat{G}}_i)^c\cap G_i^0\} \cap \{x_{mt}\}^c\})\\ &\geq \sum_{t=t_{i-1}^0+h_T}^{t_{i}^0-h_T-1} \bigg((u_{t+1}-\bar{u}^{(t_{i-1}^0+h_T:t_{i}^0-h_T)}) + (\boldsymbol{\beta}^m -\widehat{\boldsymbol{\beta}}^{m})'(\boldsymbol{\psi}_{t}^m - \bar{\boldsymbol{\psi}}^{(t_{i-1}^0+h_T:t_{i}^0-h_T)}_m) \\ &\quad~+ \beta_x^m (x_{mt} - \bar{x}^{(t_{i-1}^0+h_T:t_{i}^0-h_T)}_m)\bigg)^2\\
        &> \sum_{t=t_{i-1}^0+h_T}^{t_{i}^0-h_T-1} [u_{t+1}-\bar{u}^{(t_{i-1}^0+h_T:t_{i}^0-h_T)}]^2 + \Delta_i\bigg\{(\beta_x^{m})^2 \boldsymbol{\widehat{\Omega}}_{x_m}^{(t_{i-1}^0+h_T:t_{i}^0-h_T)}
        + \beta_x^{m}\boldsymbol{\widehat{\Omega}}_{u, x_m}^{(t_{i-1}^0+h_T:t_{i}^0-h_T)}
        \\ &~\quad + \beta_x^{m} (\boldsymbol{\beta}^m -\widehat{\boldsymbol{\beta}}^{m})' \widehat{\boldsymbol{\Omega}}_{\psi_m, x_m}^{(t_{i-1}^0+h_T:t_{i}^0-h_T)}
         + (\boldsymbol{\beta}^m -\widehat{\boldsymbol{\beta}}^{m})'\widehat{\boldsymbol{\Omega}}_{\psi_m, u}^{(t_{i-1}^0+h_T:t_{i}^0-h_T)}\bigg\}\\
        &:= \sum_{t=t_{i-1}^0+h_T}^{t_{i}^0-h_T-1} [u_{t+1}-\bar{u}^{(t_{i-1}^0+h_T:t_{i}^0-h_T)}]^2 + \Delta_i(L_{x1} + L_{x2} + L_{x3} + L_{x4}),
    \end{aligned}\label{res2_m}
\end{equation}
where the second inequality follows by $\sum_{t=t_{i-1}^0+h_T}^{t_{i}^0-h_T-1}(u_{t+1}-\bar{u}^{(t_{i-1}^0+h_T:t_{i}^0-h_T)}) = 0, \sum_{t=t_{i-1}^0+h_T}^{t_{i}^0-h_T-1}(\boldsymbol{\psi}_{t}^m - \bar{\boldsymbol{\psi}}^{(t_{i-1}^0+h_T:t_{i}^0-h_T)}_m) = \boldsymbol{0}, \sum_{t=t_{i-1}^0+h_T}^{t_{i}^0-h_T-1} (x_t^m - \bar{x}^{(t_{i-1}^0+h_T:t_{i}^0-h_T)}_m) = 0, \sum_{t=t_{i-1}^0+h_T}^{t_{i}^0-h_T-1}[(\boldsymbol{\beta}^m -\widehat{\boldsymbol{\beta}}^{m})'(\boldsymbol{\psi}_{t}^m - \bar{\boldsymbol{\psi}}^{(t_{i-1}^0+h_T:t_{i}^0-h_T)}_m)]^2>0.$
Here we will discuss the scenario for both $x_{mt}$ is stationary and non-stationary.

{\bf Case (A). $x_{mt}$ is an $I(0)$ process.}

Consider $L_{x1}.$ By similar arguments to the proofs of Lemma \ref{lemma_cov}, it can be shown that there exists a constant $C_z > 0$ such that
\begin{equation}
    L_{x1}\geq C_z \label{lx1}
\end{equation}
holds in probability.

Consider $L_{x2}.$ By similar arguments to the proofs of Lemma \ref{lemma_ux}, we have
\begin{equation}
    \|L_{x2}\| = O_p(\Delta_i^{-{1\over 2}}). \label{lx2}
\end{equation}
Consider $L_{x3}.$ Using \eqref{betamissorder_z} and similar arguments in the proof of Lemma \ref{lemma_ux}, we have
\begin{equation}
\begin{aligned}
    \|L_{x3}\| &\leq \|\boldsymbol{P}_m^{-1} (\boldsymbol{\beta}^m-\widehat{\boldsymbol{\beta}}^m)\| \|\boldsymbol{P}_m \widehat{\boldsymbol{\Omega}}_{\psi_m, x_m}^{(t_{i-1}^0+h_T:t_{i}^0-h_T)}\| |\beta_x^{m}|\\
    &= O_p(\Delta_i^{-{1\over 2}}r_F (r_F^{1\over 2}+d_T^{1\over 2}))\times O_p(\Delta_i^{-{1\over 2}} (r_F^{1\over 2}+d_T^{1\over 2}))\times O(1)\\ &= O_p(\Delta_i^{-1}r_F(r_F+d_T)) = o_p(1),
\end{aligned}\label{lx3}
\end{equation}
by the Assumption 3 that $\Delta_i = t^0_{i} - t^0_{i-1} - 2h_T\asymp |t_{i+1}^0-t_i^0|, r_F(r_F+d_T) (\min\limits_{1\leq i \leq m_0+1} |t_i^0-t_{i-1}^0|)^{-1}\rightarrow 0.$
Similarly, we have
\begin{equation}
    \begin{aligned}
        \|L_{x4}\| &= O_p(\Delta_i^{-1}r_F(r_F+d_T)) = o_p(1).
    \end{aligned}\label{lx4}
\end{equation}
Combining \eqref{lx1} to \eqref{lx4} and \eqref{upiece} yields
\begin{equation}
\begin{aligned}
     {\rm RSS}_i(\widehat{\widehat{G}}_i)  &\geq\sum_{t=\widehat{\widehat{t}}_{i-1}}^{\widehat{\widehat{t}}_{i}-1} (u_{t+1} - \bar{u}^{(t_{i-1}^0+h_T:t_{i}^0-h_T)})^2 + {1\over 2} C_z \Delta_i
    \\ &\geq \sum_{t=\widehat{\widehat{t}}_{i-1}}^{\widehat{\widehat{t}}_{i}-1} (u_{t+1} - \bar{u}^{(\widehat{\widehat{t}}_{i-1}:\widehat{\widehat{t}}_{i})})^2 +{1\over 4}C_z \Delta_i
\end{aligned}\label{RSSmissz}
\end{equation}
holds in probability, by the Assumption 3 that $h_T=o(\Delta_i).$

{\bf Case (B). $x_{mt}$ (or $w_{m-p_z,t}$) is an $I(1)$ process.}

By the factor structure \eqref{panic}, there exists $\boldsymbol{b}_{m-p_z} \in \mathbb{R}^{r_F}$ and process $\{e_{(m-p_z),t}\}$ such that $x_{mt} = \boldsymbol{b}_{m-p_z}' \boldsymbol{F}_t + \boldsymbol{e}_{m-p_z,t}.$ By similar arguments in the proofs of Lemma \ref{lemma_cov}, we have
\begin{equation}
    L_{x1} = (\beta_x^m)^2{1\over \Delta_i} \sum_{t=t_{i-1}^0+h_T}^{t_{i}^0-h_T-1} (\boldsymbol{b}_{m-p_z}' \boldsymbol{F}_t + \boldsymbol{e}_{m-p_z,t})^2 \geq C_w \Delta_i r_F^{-1} \label{lw1}
\end{equation}
holds in probability for some $C_w > 0.$

Similar to \eqref{lx2} to \eqref{lx4}, it can be shown that
\begin{equation}
    \|L_{x2}\| =O_p(1)=o_p(\|L_{x1}\|); \label{lw2}
\end{equation}
\begin{equation}
        \|L_{x3}\|=O_p(r_F(r_F+d_T))=o_p(\|L_{x1}\|);\label{lw3}
\end{equation}
\begin{equation}
        \|L_{x4}\|=O_p(r_F(r_F+d_T))=o_p(\|L_{x1}\|);\label{lw4}
\end{equation}
Combining \eqref{lw1} to \eqref{lw4} yields
\begin{equation}
    \begin{aligned}
        {\rm RSS}_i(\widehat{\widehat{G}}_i) &\geq \sum_{t=t_{i-1}^0+h_T}^{t_{i}^0-h_T-1} [u_{t+1}-\bar{u}^{(t_{i-1}^0+h_T:t_{i}^0-h_T)}]^2  + {1\over 2} C_w \Delta_i^2 r_F^{-1} \\
        &\geq \sum_{t=\widehat{\widehat{t}}_{i-1}}^{\widehat{\widehat{t}}_{i}-1}(u_{t+1}-\bar{u}^{(\widehat{\widehat{t}}_{i-1}:\widehat{\widehat{t}}_{i})})^2 + {1\over 4}C_w \Delta_i^2r_F^{-1}
    \end{aligned}\label{deltaRSS_mw0}
\end{equation}
holds in probability. Combining \eqref{RSSmissz}, \eqref{deltaRSS_mw0} and the Assumption 3 that $\Delta_i = t_{i}^0 - t_{i-1}^0 - 2h_T \asymp |t_{i}^0-t_{i-1}^0|, r_F^2 \Delta_i^{\tau-{1\over 2}} =o(1)$ concludes \eqref{A.7}.
\end{proof}

\begin{lemma}
    Suppose Assumptions 1-3 holds, and  $G^0_i\subsetneq \widehat{\widehat{G}}_i, |G^0_i|<|\widehat{\widehat{G}}_i|.$
    Then, we have
    \begin{equation}
        {\rm RSS}_i(\widehat{\widehat{G}}_i)
        \geq \sum_{t=\widehat{\widehat{t}}_{i-1}}^{\widehat{\widehat{t}}_{i}-1}(u_{t+1}-\bar{u}^{(\widehat{\widehat{t}}_{i-1}:\widehat{\widehat{t}}_{i})})^2- C' h_T
    \end{equation}
    holds in probability for some $C' > 0.$
    \label{lemma_r}
\end{lemma}

\begin{proof}

Suppose all predictors in $\widehat{\widehat{G}}_i$ are gathered in a vector $\boldsymbol{x}_t^r.$ Similar to \eqref{cointA} and the arguments in \eqref{RSSQm1} to \eqref{RSSsim2}, we can find a rotated predictor vector
$\boldsymbol{\psi}_{t}^r = ([\boldsymbol{z}_t^r]', [\boldsymbol{\psi}_{1t}^r]', [\boldsymbol{\psi}_{2t}^{r}]')'$ such that
\begin{equation}
    \begin{aligned}
    {\rm RSS}_i(\widehat{\widehat{G}}_i)
    &\geq \sum_{t=t_{i-1}^0+h_T}^{t_{i}^0-h_T-1} \bigg([y_{t+1}-\bar{y}^{(t_{i-1}^0+h_T:t_{i}^0-h_T)}] - [\widehat{\boldsymbol{\beta}}^{r}]' [\boldsymbol{\psi}_{t}^{r}-\bar{\boldsymbol{\psi}}_{r}^{(t_{i-1}^0+h_T:t_{i}^0-h_T)}]\bigg)^2,
\end{aligned}\label{res_rz}
\end{equation}
where
$\widehat{\boldsymbol{\beta}}^{r}$ represents the LSE based on the sample $(y_{t+1}, \boldsymbol{\psi}_{t}^r)$ over the period $(t_{i-1}^0+h_T:t_{i}^0-h_T).$
Denote $\boldsymbol{\beta}^r$ be the true coefficient of $\boldsymbol{\psi}_{t}^r.$ Similar to \eqref{betasim}, we have
\begin{equation}
    \begin{aligned}
    \widehat{\boldsymbol{\beta}}^{r}
    &= \boldsymbol{\beta}^r
    + \bigg(\sum_{t=t_{i-1}^0+h_T}^{t_{i}^0-h_T-1} [\boldsymbol{\psi}_{t}^{r}-\bar{\boldsymbol{\psi}}_{r}^{(t_{i-1}^0+h_T:t_{i}^0-h_T)}][\boldsymbol{\psi}_{t}^{r}-\bar{\boldsymbol{\psi}}_{r}^{(t_{i-1}^0+h_T:t_{i}^0-h_T)}]'\bigg)^{-1} \\ &~\quad\bigg(\sum_{t=t_{i-1}^0+h_T}^{t_{i}^0-h_T-1} [\boldsymbol{\psi}_{t}^{r}-\bar{\boldsymbol{\psi}}_{r}^{(t_{i-1}^0+h_T:t_{i}^0-h_T)}] [u_{t+1} - \bar{u}^{(t_{i-1}^0+h_T:t_{i}^0-h_T)}]\bigg).
    \end{aligned}
\end{equation}
Similar to \eqref{betasim}, we have
\begin{equation}
    \|\boldsymbol{P}_r^{-1} (\widehat{\boldsymbol{\beta}}^{r} - \boldsymbol{\beta}^r)\|=O_p(r_F)\times O_p(\Delta_i^{-{1\over 2}} (r_F^{1\over 2} +d_T^{1\over 2})) =O_p(\Delta_i^{-{1\over 2}}r_F(r_F^{1\over 2} +d_T^{1\over 2})),\label{betapartr}
\end{equation}
where $\boldsymbol{P}_r$ denotes the scaled matrix.
Similar to \eqref{res2_m}, we have
\begin{equation}
    \begin{aligned}
        {\rm RSS}_i(\widehat{\widehat{G}}_i) &= \sum_{t=t_{i-1}^0+h_T}^{t_{i}^0-h_T-1} \bigg(u_{t+1} - \bar{u}^{(t_{i-1}^0+h_T:t_{i}^0-h_T)}+(\boldsymbol{\beta}^r -\widehat{\boldsymbol{\beta}}^r)'[\boldsymbol{\psi}_{t}^{r}-\bar{\boldsymbol{\psi}}_{r}^{(t_{i-1}^0+h_T:t_{i}^0-h_T)}]\bigg)^2\\
        &> \sum_{t=t_{i-1}^0+h_T}^{t_{i}^0-h_T-1}(u_{t+1} - \bar{u}^{(t_{i-1}^0+h_T:t_{i}^0-h_T)})^2 + \Delta_i(\boldsymbol{\beta}^r -\widehat{\boldsymbol{\beta}}^r)' \widehat{\boldsymbol{\Omega}}_{\psi_r, u}^{(t_{i-1}^0+h_T:t_{i}^0-h_T)}.
    \end{aligned}\label{res2_r}
\end{equation}
Following the proof lines of Lemma~\ref{lemma_ux}, we have
\begin{equation}
    \|\boldsymbol{P}_r\widehat{\boldsymbol{\Omega}}_{\psi_r, u}^{(t_{i-1}^0+h_T:t_{i}^0-h_T)}\|=O_p(\Delta_i^{-{1\over 2}}(r_F^{1\over 2} +d_T^{1\over 2})).\label{res_2_3}
\end{equation}
Since $r_F(r_F+d_T) h_T^{-{1\over 2}} \rightarrow 0$ by Assumption~3, it follows from  \eqref{betapartr}, \eqref{res2_r} and \eqref{res_2_3} that
\begin{equation}
    \begin{aligned}
        {\rm RSS}_i( \widehat{\widehat{G}}_i)&> \sum_{t=t_{i-1}^0+h_T}^{t_{i}^0-h_T-1}(u_{t+1} - \bar{u}^{(t_{i-1}^0+h_T:t_{i}^0-h_T)})^2 - Cr_F(r_F+d_T)\\
        &\geq \sum_{t=\widehat{\widehat{t}}_{i}-1}^{\widehat{\widehat{t}}_{i}-1} (u_{t+1}-\bar{u}^{(\widehat{\widehat{t}}_{i-1}:\widehat{\widehat{t}}_{i})})^2 - C' h_T
    \end{aligned}
    \label{RSSr_z}
\end{equation}
holds in probability for some $C',C > 0.$   This concludes Lemma \ref{lemma_r}.
\end{proof}

\begin{proof}[\textbf{Proof of Theorem 3.4.}]

Consider the period $(\widehat{\widehat{t}}_{i-1}:\widehat{\widehat{t}}_{i}).$
By similar arguments in \eqref{L1toL5} to \eqref{RSSofneigh}, we have
\begin{equation}
    {\rm RSS}_i(G^0_i) \leq \sum_{t=\widehat{\widehat{t}}_{i}-1}^{\widehat{\widehat{t}}_{i}-1} (u_{t+1}-\bar{u}^{(\widehat{\widehat{t}}_{i-1}:\widehat{\widehat{t}}_{i})})^2 + C_Q r_F h_T^2 \label{RSSQupper}
\end{equation}
holds in probability for some $C_Q > 0.$ Note that
\begin{equation}
    {\rm IC}_2^i(G^0_i) - {\rm IC}_2^i(\widehat{\widehat{G}}_i) = {\rm RSS}_i( G^0_i) - {\rm RSS}_i( \widehat{\widehat{G}}_i) + \omega_{2T}^i (|G^0_i| - |\widehat{\widehat{G}}_i|).\label{icdiff}
\end{equation}
Suppose that the redundant predictors in $\widehat{\widehat{G}}_i$ are gathered in set $\widehat{\widehat{G}}_i^r;$ and
the incorrectly eliminated predictors are gathered in another set $\widehat{\widehat{G}}_i^m.$ Then, we have $G_i^0 = \bigg[(\widehat{\widehat{G}}_i)\cap (\widehat{\widehat{G}}_i^r)^c\bigg]\cup \widehat{\widehat{G}}_i^m.$
It suffices to prove that
\begin{equation}
    P(\widehat{\widehat{G}}_i^r =\widehat{\widehat{G}}_i^m=\phi)\rightarrow 1. \label{thm3.4-0}
\end{equation}
Next, we split the proof of (\ref{thm3.4-0}) into two cases.

\noindent
{\bf Case 1:} If $\widehat{\widehat{G}}_i^m \neq \phi,$ combining Lemma \ref{lemma_m} and \eqref{RSSQupper} yields
\begin{equation}
    {\rm RSS}_i(\widehat{\widehat{G}}_i)-{\rm RSS}_i(G_i^0) \geq C \min\limits_{1\leq i \leq m_0+1} |t_i^0-t_{i-1}^0| \label{rssm}
\end{equation}
holds in probability for some $C>0.$ By Assumption 4 (2), we have $\omega_{2T}^i = o(\min\limits_{1\leq i \leq m_0+1}|t_i^0-t_{i-1}^0|).$ Combining \eqref{icdiff} and \eqref{rssm} yields ${\rm IC}_2^i( G_i^0)-{\rm IC}_2^i(\widehat{\widehat{G}}_i) < 0$ holds in probability, which yields
\begin{equation}
    P(\widehat{\widehat{G}}_i^m \neq \phi) \rightarrow 0. \label{thm3.4-1}
\end{equation}

\noindent
{\bf Case 2:} If $\widehat{\widehat{G}}_i^m = \phi, \widehat{\widehat{G}}_i^r \neq \phi,$ combining Lemma \ref{lemma_r} and \eqref{RSSQupper} yields
\begin{equation}
    {\rm RSS}_i(G^0_i) - {\rm RSS}_i(\widehat{\widehat{G}}_i) \leq C_Q r_Fh_T^2 + C'h_T \leq 2C_Q r_Fh_T^2
\end{equation}
holds in probability. Note that when $\widehat{\widehat{G}}_i^m = \phi, \widehat{\widehat{G}}_i^r \neq \phi,$ we have $|G^0_i| - |\widehat{\widehat{G}}_i| < 0.$ This yields ${\rm IC}_2^i(G_i^0)-{\rm IC}_2^i(\widehat{\widehat{G}}_i) < 0$ holds in probability by Assumption 4 (2) that $r_Fh_T^2 (\omega_{2T}^i)^{-1}\rightarrow 0.$ Thus,
 \begin{equation}
    P(\widehat{\widehat{G}}_i^m = \phi, \widehat{\widehat{G}}_i^r \neq \phi)\rightarrow 0. \label{thm3.4-2}
\end{equation}  Combining \eqref{thm3.4-1} and \eqref{thm3.4-2}
 yields \eqref{thm3.4-0} and concludes Theorem 3.4.
\end{proof}


\section{Modified-\texttt{tcode} transformation}\label{sec:supp_tcode}

This appendix documents the modified transformation convention
\texttt{tcode\_mod} used to construct the FRED-MD working panel
of Section~\ref{sec:emp}. The default FRED-MD
\texttt{tcode} column prescribes a series-specific transformation
designed to render every series approximately stationary. This
default is inappropriate for cointegration-aware methods like
ours because it eliminates the cointegrated $I(1)$ structure
that SICS detects through its canonical-correlation ranking
(Section~\ref{sec:sics}). \texttt{tcode\_mod} replaces the
default with the rule mapping in Table~\ref{tab:tcode_mod}: each
default transformation is shifted one step toward the level so
that the resulting series is in $I(0)$ or $I(1)$ rather than
$I(0)$ alone. A series originally classified as
\texttt{tcode = 5} (log-differenced) is used in log form
(yielding $I(1)$); a series classified as \texttt{tcode = 6}
($\Delta^2 \log$) is used in $\Delta\log$ form (yielding $I(1)$);
and so on. Series with the original \texttt{tcode = 7} are
dropped because their integration order is ambiguous to
back-engineer.

\begin{table}[!ht]
  \centering\small
  \caption{Modified transformation codes (\texttt{tcode\_mod}).}
  \label{tab:tcode_mod}
  \begin{tabular}{ccccc}
    \toprule
    \multicolumn{2}{c}{Default \texttt{tcode}} & & \multicolumn{2}{c}{Modified \texttt{tcode\_mod}} \\
    \cmidrule(lr){1-2} \cmidrule(lr){4-5}
    Code & Operation & & Operation & Resulting integration \\
    \midrule
    1 & level                          & & level             & $I(0)$ \\
    2 & $\Delta$ level                 & & level             & $I(1)$ \\
    3 & $\Delta^2$ level               & & $\Delta$ level    & $I(1)$ \\
    4 & log                            & & log               & $I(0)$ \\
    5 & $\Delta \log$                  & & log               & $I(1)$ \\
    6 & $\Delta^2 \log$                & & $\Delta \log$     & $I(1)$ \\
    7 & $\Delta(x_t/x_{t-1} - 1)$      & & dropped           & --- \\
    \bottomrule
  \end{tabular}

  \vspace{4pt}
  \parbox{0.75\linewidth}{\scriptsize\linespread{0.95}\selectfont\textit{Note:} The default \texttt{tcode} targets approximate stationarity for every series; \texttt{tcode\_mod} shifts each transformation one step toward the level so that the resulting series is $I(0)$ or $I(1)$, preserving the cointegrated structure required by SICS. Series with \texttt{tcode\,=\,7} are dropped because their integration order is ambiguous to back-engineer.}
\end{table}

The working panel is the FRED-MD database from
1960:01 to 2024:12 ($T_{\text{raw}} = 780$ months) under
\texttt{tcode\_mod}. The raw vintage contains 126 series; the one
series with \texttt{tcode\,=\,7} is dropped during transformation
(leaving 125), and a further four series whose missing-value share
exceeds five percent are dropped (leaving 121). The 121 retained
series may individually carry up to five percent missing
observations; these residual gaps are handled by prediction index
alignment and complete-case row deletion. Specifically, with
$\boldsymbol x_t$ the one-period lag of the full 121-series panel
and $y_{t+1}$ the target, any month $t$ for which at least one
entry of $\boldsymbol x_t$ is unobserved is discarded; 31 such
months are removed, yielding $T = 748$ fully observed observations.
$\boldsymbol x_t$ contains all 121 retained series---including
the lagged target $y_t$ as an autoregressive component.

\end{document}